\documentclass[11pt]{article}
\usepackage{amssymb}
\usepackage{colortbl}
\usepackage{amsfonts,amsmath, longtable}
\usepackage{ulem}
\usepackage{comment}
\usepackage{enumerate}

\usepackage{hyperref}

        \def\theequation{\thesection.\arabic{equation}}

\newcommand{\tr}{{\rm tr}}
\newcommand{\ti}[1]{\tilde{#1}}

\newcommand{\mL}{{\mathcal L}}
\newcommand{\mM}{{\mathcal M}}
\newcommand{\mF}{{\mathcal F}}

\newcommand{\mH}{{\mathcal H}}

\newcommand{\la}{\lambda}
\newcommand{\vf}{\varphi}
\newcommand{\al}{\alpha}
\newcommand{\be}{\beta}
\newcommand{\ga}{\gamma}
\newcommand{\om}{\omega}
\newcommand{\vth}{\vartheta}

\newcommand{\bfr}{{\bf{r}}}
\newcommand{\bfy}{{\bf{y}}}
\newcommand{\tbfr}{\tilde{\bf{r}}}

\newcommand{\Mat}{ {\rm Mat}(N,\mathbb C) }

\newcommand{\mC}{\mathbb C}
\newcommand{\mZ}{\mathbb Z}

\newcommand{\ka}{\kappa}

\newcommand{\Om}{\Omega}

\newtheorem{proposition}{Proposition}[section]

\newtheorem{lemma}{Lemma}[section]

\newtheorem{rem}{Remark}[section]

\newenvironment{proof}{\par\noindent{\bf Proof.}}{\hfill$\scriptstyle\blacksquare$}
\def\beq{\begin{equation}}
\def\eq{\end{equation}}
\def\p{\partial}

\newtheorem{theor}{Theorem}%[section]

\newcommand{\mat}[4]{\left(\begin{array}{cc}{#1}&{#2}\\ \ \\{#3}&{#4}
\end{array}\right)}

\newcommand{\mats}[4]{\left(\begin{array}{cc}{#1}&{#2}\\ {#3}&{#4}
\end{array}\right)}

\def\res{\mathop{\hbox{Res}}\limits}

\begin{document}

\begin{center}

\setcounter{page}{1}

\vspace{0mm}

{\Large{\bf
Matrix representations of ${\rm BC}_n$ associative Yang--Baxter equations
}}

\vspace{3mm}

{\Large{\bf
and related long-range spin chains
}}

% \vspace{3mm}

%{\Large{\bf and its multispin generalizations}}

 \vspace{15mm}

 {\Large {M. Matushko}}$\,^{\diamond}$
\qquad
 {\Large {A. Mostovskii}}$\,^{\diamond}$$\,^{*}$
\qquad
 {\Large {A. Zotov}}$\,^{\diamond}$

  \vspace{10mm}

$\diamond$ --
 {\it Steklov Mathematical Institute of the Russian
Academy of Sciences,\\ 8 Gubkina Street, Moscow 119991, Russia}

%$\bullet$ --
%{\it Center for Advanced Studies, Skoltech, 143026, Moscow, Russia}

$*$ --
%{\it Institute for Theoretical and Mathematical Physics,\\ Lomonosov Moscow State University, 119991,  Moscow, Russia}
{\it %Institute for Theoretical and Mathematical Physics,\\
Lomonosov Moscow State University, Moscow, 119991, Russia}

%$\flat$ -- {\it National Research Centre ''Kurchatov Institute'',\\
%Academician Kurchatov square, 1, Moscow, 123182, Russia}

   \vspace{5mm}

 {\small\rm {E-mails: matushko@mi-ras.ru, mostovskii.am21@physics.msu.ru, zotov@mi-ras.ru}}

\end{center}

\vspace{0mm}

\begin{abstract}
We construct a matrix representation of the algebra defined by the associative Yang--Baxter equations of ${\rm BC}_n$ type. In the elliptic case, two sets of generators are represented in terms of the Baxter--Belavin quantum $R$-matrix in the fundamental representation of ${\rm GL}_N$, while the remaining generators are represented by the elliptic $K$-matrix found by Komori and Hikami. We revisit its derivation using the Felder--Pasquier procedure and correct the unitarity relation for odd $N$. We then use these matrix representations to construct examples of long-range quantum spin chains from $R$-matrix-valued Lax pairs for the classical Calogero--Inozemtsev system associated with the ${\rm BC}_n$ root system. We also construct $R$-matrix-valued Lax pairs for Calogero--Moser models in an external field, leading to Haldane--Shastry and Polychronakos--Frahm chains with modified spin positions.
\end{abstract}

\newpage
%\bigskip
{\small{ \tableofcontents }}

\bigskip

%\newpage

\newpage
\subsection*{\ \ \ \ \ \ \underline{Main notation:}}
%\addcontentsline{toc}{section}{ Notations}

\ \vspace{-1mm}

{\small{

$n$ -- number of particles in many-body systems or spins in spin chains;

$N$ -- dimension of the matrix representation for the $R$ and $K$ matrices ($R\in\Mat^{\otimes 2}$, $K\in\Mat$ );

$H$ -- classical Hamiltonian of a many-body system;

$\mH$ -- quantum Hamiltonian of a spin chain;

${\bf r}_{ij}$ -- generators of non-parametric $A$-type (\ref{a101})-(\ref{a103}) or ${\rm BC}_n$-type (\ref{a115})-(\ref{a118})
associative Yang--Baxter equations (AYBE);

${\ti{\bf r}}_{ij}$ and ${\bf y}_i$ -- generators of non-parametric ${\rm BC}_n$-type (\ref{a115})-(\ref{a118})
associative Yang--Baxter equations (AYBE);

$R_{12}^\hbar(z)$ -- quantum $R$-matrix; also a representation of ${\bf r}_{ij}$ in the parametric $A$-type AYBE (\ref{a104})-(\ref{a106})
or parametric ${\rm BC}_n$-type AYBEs (\ref{a122})-(\ref{a131});

$P_{12}$ -- permutation matrix (\ref{P12});

${\ti R}_{12}^\hbar(z)$ -- representation of ${\ti{\bf r}}_{ij}$ in the parametric ${\rm BC}_n$-type AYBEs (\ref{a122})-(\ref{a131});

${\ti K}^\hbar(z)$ -- $K$-matrix solving the reflection equation (\ref{a157}). It also gives a matrix (or operator) representation of ${\bf y}_i$ in ${\rm BC}_n$-type AYBEs (\ref{a122})-(\ref{a131});

${K}^\hbar(z)$ -- $K$-matrix solving the reflection equation (\ref{a156}) related to ${\ti K}^\hbar(z)$ through (\ref{elKfromt});

$J$ -- representation of the reflection ${\hat t}_i$ (\ref{a121}). In the matrix representation, it is the involutive matrix (\ref{a077});

$F_{ij}^z(q)$ and ${\ti F}_{ij}^z(q)$ -- derivatives of $R_{ij}^z(q)$ and ${\ti R}_{ij}^z(q)$ with respect to $q$;

${\ti Y}_j^z(q)$ -- derivative of ${\ti K}_j^z(q)$  with respect to $q$;

${\hat{\bf T}}_i$ or ${\hat{\bf T}}_{N,i}$ -- shift operators introduced before (\ref{fp20});

$T_a$ -- basis matrices in $\Mat$ (\ref{el10});

$Q$ and $\Lambda$ -- matrices entering the definition of $T_a$ (\ref{el11})-(\ref{el12});

$\sigma_k$ -- Pauli matrices (\ref{r723});

$\tau$ -- modulus of the elliptic curve (${\rm Im}(\tau)>0$);

$\Om_a$ -- half-periods of the elliptic curve (\ref{w209});

$\om_a$ -- arguments of the $\phi$-function (\ref{w2023}). For $a=(a_1,a_2)\in\mZ_N\times\mZ_N$ $\om_a=\frac{a_1+a_2\tau}{N}$;

$z_i$ -- equilibrium position of the $i$-th particle;

$L(z)$, $M(z)$ -- matrices forming a classical Lax pair for a many-body system;

$\mL(z)$, $\mM(z)$ -- matrices forming a classical $R$-matrix-valued Lax pair for a many-body system; at equilibrium positions, they form a quantum Lax pair for the corresponding spin chain.

}}

\newpage
%%%%%%%%%%%%%%%%%%%%%%%%%%%%%%%%%%%%%%%%%%%%%%%%%%%%%%%%%%%%%%%%%%%%%%%%%%%%%%%%%%%%%%%%%%%%%%%%%
%%%%%%%%%%%%%%%%%%%%%%%%%%%%%%%%%%%%%%%%%%%%%%%%%%%%%%%%%%%%%%%%%%%%%%%%%%%%%%%%%%%%%%%%%%%%%%%%%
\section{Introduction}\label{sec1}
\setcounter{equation}{0}
The associative Yang--Baxter equation (AYBE)
\beq\label{a01}
\begin{array}{c}
\displaystyle{
\bfr_{ij} \bfr_{jk}=\bfr_{ik}\bfr_{ij}+\bfr_{jk}\bfr_{ik} \qquad \text{ for distinct } i,j,k.
}
\end{array}
\eq
was introduced by S. Fomin and An. Kirillov in \cite{FK} as one of the defining relations of a quadratic algebra. It was then studied from algebraic \cite{Aguiar}
and geometric \cite{Pol} viewpoints. In particular, it was shown in \cite{Pol} that
the parametric version of the relations (\ref{a01})
 \beq\label{a02}
\begin{array}{c}
\displaystyle{
 R_{ij}^z(q_i-q_j) R_{jk}^w(q_j-q_k)=R_{ik}^w(q_i-q_k)R_{ij}^{z-w}(q_i-q_j)+
 R_{jk}^{w-z}(q_j-q_k)R_{ik}^z(q_i-q_k)
}
\end{array}
 \eq
 has a solution given by the elliptic Baxter--Belavin quantum $R$-matrix \cite{BB} in the fundamental representation
 of the Lie group ${\rm GL}_N$. The latter is a solution of the quantum Yang--Baxter equation (QYBE)
  \beq\label{a021}
  \begin{array}{l}
  \displaystyle{
 R^{\hbar}_{12}(q_1-q_2)  R^{\hbar}_{13}(q_1-q_3) R^{\hbar}_{23}(q_2-q_3) =
 R^{\hbar}_{23}(q_2-q_3)  R^{\hbar}_{13}(q_1-q_3)  R^{\hbar}_{12}(q_1-q_2)\,.
 }
 \end{array}
 \eq
 Thus, equations (\ref{a02}) and (\ref{a021}) have common solutions. In fact, any unitary and skew-symmetric solution
 of the AYBE (\ref{a02}) is also a solution of the quantum Yang--Baxter equation (\ref{a021}). A simple proof can be found in
 \cite{LOZ14}.
 Throughout this paper, we consider $R$-matrices of this type.

 Both versions of the AYBE, (\ref{a01}) and (\ref{a02}), have many applications to integrable systems. The original one (\ref{a01})
 can be used to describe operators of Dunkl and Knizhnik--Zamolodchikov type \cite{Kir},
 while the parametric AYBE (\ref{a02}) can be used for constructing classical Lax representations, Sklyanin algebras
 and other integrable structures; see
 \cite{LOZ14,LOZ161,GSZ,KrZ}.
In this paper we focus on $R$-matrix-valued Lax pairs  \cite{LOZ14,GrZ18} and their applications
to the derivation of quantum integrable long-range spin chains \cite{SeZ18}. The resulting Hamiltonian
for a chain of $n$ spins takes the form
 \beq\label{a03}
\begin{array}{c}
\displaystyle{
 \mH=\sum\limits_{i<j}^n F^0_{ij}(z_i-z_j)\,,\qquad z_k=\frac{k}{n}\,,
}
\end{array}
 \eq
where $F^0_{ij}(w)=\p_wR^\hbar_{ij}(w)|_{\hbar=0}$, or equivalently, $F^0_{ij}(w)=\p_wr_{ij}(w)$ for
the classical $r$-matrix $r_{ij}(w)$ appearing in the expansion $R^\hbar_{ij}(w)=\hbar^{-1}{\rm Id}+r_{ij}(w)+O(\hbar)$ of the elliptic Baxter--Belavin $R$-matrix or one of its trigonometric degenerations.
Here we consider $F^0_{ij}(w)$ as an operator acting on the Hilbert space $(\mC^N)^{\otimes n}$. The points $z_k$ are equilibrium positions
in the underlying classical many-body integrable system. For (\ref{a03}), this is the elliptic or trigonometric Calogero--Moser--Sutherland system.
In a particular case, (\ref{a03}) reproduces the Hamiltonian of the Haldane--Shastry chain \cite{HS1}. In the general case,
(\ref{a03}) is the quantum Hamiltonian of an anisotropic long-range integrable spin chain. We briefly recall this construction in the next section.
Further studies of spin chains of this type and their $q$-deformations can be found in \cite{Lam,MZ,KL,Ch3,ChM,LRS}.
The constructions above correspond to periodic (or closed) chains of type $A$.
Long-range Haldane--Shastry-type spin chains with open boundaries were constructed in \cite{BPS}; on the classical side, open and periodic spin Calogero--Moser chains were studied in \cite{Resh}.

In order to obtain long-range spin chains with boundary terms, we turn to the ${\rm BC}_n$-type analogues of the AYBE.
In \cite{Kir} An. Kirillov explained that the relations (\ref{a01}) are associated with the root system of type $A_n$,
and that analogues of (\ref{a01}) can be introduced for other root systems.
The $B$-type version of the associative Yang--Baxter equations includes additional generators
$\tbfr_{ij}$ and $\bfy_i$. For example, the four-term relation involving all types of generators reads
 \beq\label{a04}
\begin{array}{c}
\displaystyle{
\bfr_{ij}\bfy_j=\bfy_i\bfr_{ij}+\tilde{\bfr}_{ij}\bfy_i+\bfy_j\tilde{\bfr}_{ij}\,.
}
\end{array}
 \eq
 In our previous paper \cite{MMZ} we proposed a parametric version of the B-type AYBEs.
 In particular, the relation (\ref{a04}) takes the form
 \beq\label{a05}
 \begin{array}{c}
\displaystyle{
 R_{ij}^{w+z}(q_i-q_j)\tilde{K}^w_j(q_j)=\tilde{K}^w_i(q_i)R_{ij}^{z-w}(q_i-q_j)\!+\!
 \tilde{R}_{ij}^{w-z}(q_i+q_j)\tilde{K}^z_i(q_i)\!+\!
 \tilde{K}_j^{-z}(q_j)\tilde{R}_{ij}^{w+z}(q_i+q_j)\,,
}
\end{array}
 \eq
where the operators $\tilde{K}^z_i(q_i)$ represent the generators $\bfy_i$. Using a representation of Shibukawa--Ueno type \cite{SU},
we also proved that these operators give rise to $K$-matrices solving the reflection equation \cite{Cherednik,Skl-refl}:
 \beq\label{a06}
 \begin{array}{c}
 \displaystyle{
%R_{12}^{\hbar}(u-v)\tilde{K}_1^\hbar(u)\tilde{R}_{12}^{\hbar}(u+v)
%\tilde{K}_2^\hbar(v)=\tilde{K}_2^\hbar(v)\tilde{R}_{12}^{\hbar}(u+v)
%\tilde{K}_1^\hbar(u)R_{12}^{\hbar}(u-v)\,.
%
R_{ij}^{\hbar}(q_i-q_j)K_i^\hbar(q_i)R_{ji}^{\hbar}(q_i+q_j)K_j^\hbar(q_j)=
K_j^\hbar(q_j)R_{ij}^{\hbar}(q_i+q_j)K_i^\hbar(q_i)R_{ji}^{\hbar}(q_i-q_j)\,,
}
\end{array}
 \eq
 where $K^\hbar(w)$ is related to ${\ti K}^\hbar(w)$ by multiplication by a constant matrix.
 Solutions of (\ref{a06}) were used in \cite{Skl-refl} to construct quantum integrable spin chains with boundaries,
 and $K$-matrices provide the boundary terms in the quantum Hamiltonians.
 The most general elliptic solution of (\ref{a06}) for the $K$-matrix
 was proposed by Y. Komori and K. Hikami in \cite{KH2}.

 The parametric $B$-type AYBEs, including (\ref{a05}), allow us to extend the construction of $R$-matrix-valued Lax pairs
 to the ${\rm BC}_n$ Calogero--Inozemtsev system \cite{OP,Inoz89}. As announced at the end of \cite{MMZ}, this
 construction yields a family of quantum long-range spin chains with Hamiltonians of the form
 \beq\label{a07}
\begin{array}{c}
\displaystyle{
 \mH=\sum\limits_{i<j}^n \Big(F^0_{ij}(z_i-z_j)+{\ti F}^0_{ij}(z_i+z_j)\Big)+
 \frac12\sum\limits_{k=1}^n \tilde{Y}_k^0(z_k)\,,
 \qquad
 \tilde{Y}^0(w)=\p_w {\ti K}^\hbar(w)\Big|_{\hbar=0}\,,
}
\end{array}
 \eq
where $z_k$ are equilibrium positions in the spinless classical Calogero--Inozemtsev model \cite{Inoz89} associated with the ${\rm BC}_n$
root system. This model is defined by the Hamiltonian
\begin{equation}
 \label{a071}
   \displaystyle{
 H^{\rm CI} = \frac{1}{2} \sum_{k=1}^{n} p_k^2 - g^2
\sum_{i<j}^{n} \Big(\wp(q_i-q_j) + \wp(q_i+q_j)\Big) -
\frac{1}{2}\sum\limits_{a=0}^3\sum_{k=1}^{n} \nu_a^2\wp(q_k+\Om_a)\,,
}
 \end{equation}
 where $\Om_a$ are the half-periods $0,\frac12,\frac{1+\tau}{2},\frac{\tau}{2}$ of an elliptic curve with modulus $\tau$,
  $\wp$ is the elliptic Weierstrass $\wp$-function (\ref{a905}) and  $g,\nu_0,\nu_1,\nu_2,\nu_3\in\mC$ are five arbitrary constants.
 The four constants $\nu_0,\nu_1,\nu_2,\nu_3$ enter the $K$-matrix ${\ti K}^\hbar(w)$ in (\ref{a07}), and the constant $g$
 can be normalised to $g=1$ when considering equilibrium positions. This is why it is absent in (\ref{a07}).

{\bf The purpose of the paper} is to describe matrix representations of parametric $B$-type AYBEs including (\ref{a05})
and to obtain explicit examples of the corresponding Hamiltonians from (\ref{a07}).

Our first aim is to describe elliptic solutions of the four-term relation (\ref{a05}).
We consider the Baxter--Belavin elliptic $R$-matrix \cite{BB} (in the fundamental representation
of the Lie group ${\rm GL}_N$) since it satisfies both the quantum Yang--Baxter equation
(\ref{a021}) and the associative Yang--Baxter equation (\ref{a02}). In the matrix basis $T_a$, $a=(a_1,a_2)\in\mZ_N\times\mZ_N$
in $\Mat$ (see (\ref{el10})-(\ref{el12})) it takes the following form:
\beq\label{a072}
 \begin{array}{c}
  \displaystyle{
 R_{12}^\hbar(z)=\frac{1}{N}\sum\limits_{a\in\,\mZ_{ N}\times\mZ_{ N}}
  T_a\otimes T_{-a} \vf_a(z,\om_a+\frac{\hbar}{N})
 \in{\rm Mat}(N,\mC)^{\otimes 2}\,,
 }
 \\ \ \\
  \displaystyle{
  \vf_a(z,\om_a+\frac{\hbar}{N})=\exp (2\pi\imath\,\frac{a_2z}{N})\,\phi(z,\frac{a_1+a_2\tau}{N}+\frac{\hbar}{N})\,,\quad
  \phi(z,w)=\frac{\vth'(0)\vth(z+w)}{\vth(z)\vth(w)}\,.
 }
 \end{array}
 \eq
We show that the parametric AYBEs of B type\footnote{We also refer to Kirillov's B-type AYBEs as AYBEs of ${\rm BC}_n$ type because of their relation to the Calogero--Inozemtsev model (\ref{a071}).} are satisfied with
\beq\label{a073}
 \begin{array}{c}
  \displaystyle{
 {\ti R}_{12}^\hbar(z)=\frac{1}{N}\sum\limits_{a\in\,\mZ_{ N}\times\mZ_{ N}}
  T_a\otimes T_{a} \vf_a(z,\om_a+\frac{\hbar}{N})
 \in{\rm Mat}(N,\mC)^{\otimes 2}\,.
 }
 \end{array}
 \eq
We also propose a compact form for the elliptic $K$-matrix:
\beq\label{a074}
  \displaystyle{
    \tilde{K}^{\hbar}(z) = \frac{1}{2N}
    \sum_{\be \in \mZ_{2N} \times \mZ_{2N}} \exp(\frac{2 \pi \imath \be_2 z}{N})
     v (z, \frac{\hbar}{N} + \frac{\om_\be}{2}) T_\be\in \Mat\,,
    }
\eq
where the function $v$ depends on four coupling constants:
\beq\label{a075}
  \displaystyle{
v(z,u)\equiv v(z,u|\nu)=
\sum\limits_{k=0}^3\nu_k \exp(4\pi \imath z\p_\tau\Om_k)\phi(2z,u+\Om_k)\,.
    }
\eq
The matrix $K^\hbar(z)$ entering the reflection equation (\ref{a06}) is given by
\beq\label{a076}
\begin{array}{c}
\displaystyle{
K^{\hbar}(z)=-\tilde{K}^{\hbar}(z)\,J\,,
}
\end{array}
\eq
where $J$ is the involutive matrix
\beq\label{a077}
\displaystyle{
J=e_{1,N-1}+e_{2,N-2}+\dots+e_{N-1,1}+e_{N,N}=\frac{1}{2N}\sum\limits_{a\in\,\mZ_{2N}\times\mZ_{2N}} T_a\in\Mat\,,
}
\eq
It satisfies $J^2=1_N$, where $1_N$ is the $N\times N$ identity matrix,
 and induces the involutive automorphism
 $J\, T_\alpha\, J^{-1}=T_{-\alpha}$. Here $e_{ij}$ are matrix units (the standard matrix basis in $\Mat$).
%
%\beq\label{a0781}
%\displaystyle{
%J^2=1\,,\qquad J\, T_\alpha\, J^{-1}=T_{-\alpha}\,.
%}
%\eq
%
This definition gives the following relation between
 $R_{12}^\hbar(z)$ and ${\ti R}_{12}^\hbar(z)$:\footnote{Let us mention that the appearance of ${\ti R}_{12}^\hbar(z)$ related
 to $R_{12}^\hbar(z)$ by the action of some involution in one of the tensor components is known for reflection equations \cite{FM,KuSkl}.}
\beq\label{a078}
\displaystyle{
    \tilde{R}^{\hbar}_{12} (z)=J_2\,     R^{\hbar}_{12} (z)J_2\,,\qquad J_2=1_N\otimes J\,.
}
\eq
Up to multiplication by constant matrices, (\ref{a074}) coincides with the Komori--Hikami $K$-matrix from \cite{KH2}.
In \cite{KH2}, the authors used the standard matrix basis instead of $T_a$. Their expression is therefore less compact. In \cite{BS}, the Komori--Hikami $K$-matrix was
written in a basis similar to $T_a$. Our version (\ref{a074}) differs from the one obtained in \cite{BS}
in its use of the function $v$ and a sum over $\mZ_{2N} \times \mZ_{2N}$ instead of $\mZ_{N} \times \mZ_{N}$.
This gives a uniform expression for odd and even $N$. We prove that the $R$-matrices and $K$-matrices defined above
satisfy the system of ${\rm BC}_n$ AYBEs including (\ref{a05}). The $K$-matrix (\ref{a076}) also
 satisfies the reflection equation (\ref{a06}).

These results also apply to trigonometric and rational $R$-matrices and $K$-matrices.
Trigonometric $R$-matrices satisfying (\ref{a02}) were classified in \cite{SP} in the ${\rm GL}_N$ case,
and some solutions for Lie supergroups were found in \cite{MZ2}. In the rational case the classification is unknown.
However, known nontrivial solutions include Cherednik's 11-vertex $R$-matrix \cite{Chered} and
its higher-rank generalisation \cite{AtZ5}. It is an interesting problem to find (and possibly classify)
the corresponding ${\ti R}$- and $K$-matrices. Some degenerations of the elliptic expressions (\ref{a072})-(\ref{a074})
 are considered in the final part of our
paper.

Our next aim is to clarify the unitarity condition for the elliptic $K$-matrix.
It has the following form: $K^\hbar(z)K^\hbar(-z)=k(\hbar,z)1_N$,
where $k(\hbar,z)$ is a scalar function and $1_N$ is the $N\times N$ identity matrix.
 We observe that in \cite{KH2}
this function $k(\hbar,z)$ was computed correctly for even $N$ only.
The expression for $k(\hbar,z)$ for odd $N$ in \cite{KH2}
requires a correction, although the expression for the $K$-matrix itself is correct.
The derivation of the $K$-matrix in \cite{KH2}
was based on representations of the $R$- and $K$-matrices by operators of Shibukawa--Ueno type \cite{SU}. Then the Felder--Pasquier procedure \cite{FP} was applied, which
restricts the action of operators to a finite-dimensional space.
In this way one obtains matrix representations.
We rederive the $K$-matrix in this way, choosing an appropriate finite-dimensional space.
This yields the expression (\ref{a074}) and the coefficient function $k(\hbar,z)$.

Finally, we present a set of examples of long-range
spin chains with the quantum Hamiltonian (\ref{a07}).
For this purpose we use known trigonometric and rational solutions of
the AYBE for $R$-matrices in the $A_n$ case.
Here we restrict attention to the ${\rm GL}_2$ case ($N=2$),
for which both the $R$-matrices and the corresponding $K$-matrices are known. We therefore obtain explicit expressions for the quantum Hamiltonian (\ref{a07}).
Another family of spin chains arises from Calogero--Moser models
 in an external field, following \cite{Inoz83}.
We construct $R$-matrix-valued generalisations of their Lax representations.
These yield Haldane--Shastry and Polychronakos--Frahm chains with spin positions
that differ from the standard choices. Presumably, all quantum spin chains
presented in this paper are integrable. However,
we note that our approach, based on the $R$-matrix-valued Lax pairs, leads to quantum Lax equations,
which do not guarantee integrability. A rigorous proof of integrability requires other methods.

Versions of the Haldane--Shastry and Polychronakos--Frahm chains
of ${\rm BC}_n$ type have been extensively studied using a different approach \cite{EFGR}.
The chains obtained here from ${\rm BC}_n$ $R$-matrix-valued Lax pairs are generally anisotropic:
their pairwise spin interactions are not proportional
to the permutation (spin-exchange) operator $P_{ij}$, except in (\ref{Hr6sc1}) and (\ref{Hr6sc2}).

%?? In particular we obtain ... or reproduce?? ... Polychronakos-Frahm \cite{Poly,Frahm}...
%Rodriges \cite{EFGR}... Calogero-Moser model in external field cite Perelomov...

%%%%%%%%%%%%%%%%%%%%%%%%%%%%%%%%%%%%%%%%%%%%%%%%%%%%%%%%%%%%%%%%%%%%%%%%%%%%%%%%%%%%%%%%%%%%%%%%%
%%%%%%%%%%%%%%%%%%%%%%%%%%%%%%%%%%%%%%%%%%%%%%%%%%%%%%%%%%%%%%%%%%%%%%%%%%%%%%%%%%%%%%%%%%%%%%%%%
\section{General construction}\label{sec2}
\setcounter{equation}{0}

In this section we recall how ${\rm A}_{n-1}$-type AYBEs give rise to $R$-matrix-valued Lax pairs for classical integrable systems,
and then how these classical Lax pairs, when evaluated at equilibrium positions,
 are interpreted as quantum Lax pairs for quantum long-range spin chains. This is an alternative
 explanation of the so-called Polychronakos freezing trick \cite{Poly,Poly2}.
 We then turn to AYBEs of ${\rm BC}_n$ type and explain the origin of the Hamiltonian (\ref{a07}).

\subsection{The \texorpdfstring{$A_{n-1}$}{An-1} case: from the AYBE to long-range spin chains}
{\bf The associative classical Yang--Baxter
algebra of type ${\rm A}$ } (in the absence of parameters) generated by the elements ${\bf r}_{ij}$ for $1\le i\neq j\le n $
is defined by (\ref{a01}) together with the skew-symmetry
property\footnote{The definition (\ref{a101})-(\ref{a103}) differs from the Fomin--Kirillov algebra of
\cite{FK} because we do not impose the condition $\bfr_{ij}^2=0$.} \cite{FK}:
\beq\label{a101}
\begin{array}{c}
\bfr_{ij} \bfr_{jk}=\bfr_{ik}\bfr_{ij}+\bfr_{jk}\bfr_{ik} \qquad \text{ for distinct } i,j,k\,,
\end{array}
\eq
\beq\label{a102}
\begin{array}{c}
\bfr_{ij}=-\bfr_{ji}\,,\quad i\neq j
\end{array}
\eq
and
\beq\label{a103}
\begin{array}{c}
{\displaystyle
\bfr_{ij} \bfr_{kl}=\bfr_{kl}\bfr_{ij} \qquad \text{ for distinct } i,j,k,l\,.
}
\end{array}
\eq
{\bf Its parametric version }
was proposed in \cite{Pol}:
\beq\label{a104}
\begin{array}{c}
 R_{ij}^z(q_i-q_j) R_{jk}^w(q_j-q_k)=R_{ik}^w(q_i-q_k)R_{ij}^{z-w}(q_i-q_j)+
 R_{jk}^{w-z}(q_j-q_k)R_{ik}^z(q_i-q_k)
\end{array}
\eq
 for distinct $i,j,k$,
\beq\label{a105}
\begin{array}{c}
 R^z_{ij}(q)=-R_{ji}^{-z}(-q)\,,\quad i\neq j
\end{array}
\eq
and
\beq\label{a106}
\begin{array}{c}
{\displaystyle
R^z_{ij}(q_{ij}) R^w_{kl}(q_{kl})=R^w_{kl}(q_{kl}) R^z_{ij}(q_{ij}) \qquad \text{ for distinct } i,j,k,l\,.
}
\end{array}
\eq
In what follows we use the shorthand notation
\beq\label{a107}
\begin{array}{c}
{\displaystyle
q_{ij}=q_i-q_j\,.
}
\end{array}
\eq
%
%One can also supply the above definition with the unitarity property

%{\bf Representations.}
 The scalar (one-dimensional) representation of the above relations is given by the Kronecker elliptic function (\ref{a901})
 \beq\label{a108}
  \begin{array}{l}
  \displaystyle{
 \phi(z,u)=\frac{\vth'(0)\vth(z+u)}{\vth(z)\vth(u)}\,,
 }
 \end{array}
 \eq
where $\vth(z)$ is the first Jacobi theta function. It is known to satisfy the summation formula
  \beq\label{a109}
  \begin{array}{c}
  \displaystyle{
\phi(z,q_{ij})\phi(w,q_{jk})=\phi(w,q_{ik})\phi(z-w,q_{ij})+\phi(w-z,q_{jk})\phi(z,q_{ik})\,,
 }
 \end{array}
 \eq
which is the scalar (and commutative) version of the AYBE (\ref{a104}). It was shown in \cite{Pol} that
(\ref{a104})-(\ref{a106}) admit a matrix representation given by the suitably normalised Baxter--Belavin elliptic
$R$-matrix \cite{BB} in the fundamental representation of the Lie group ${\rm GL}_N$.
A simple proof of this statement can also be found in \cite{LOZ14}. We describe the elliptic $R$-matrix in the
next section. Here, the chosen normalisation means $R^z_{ij}(q_{ij})|_{N=1}=\phi(z,q_{ij})$, which
is in agreement with passing to the scalar functional equation (\ref{a109}) from (\ref{a104}) for $N=1$.

\noindent{\bf Representation by Shibukawa--Ueno operators.}
The original QYBE (\ref{a021}) admits the following representation by Shibukawa--Ueno operators \cite{SU}:
  \beq\label{a110}
  \begin{array}{c}
  \displaystyle{
R^\hbar_{ij}(q)=\phi(\hbar,x_i-x_j)-\phi(-q,x_i-x_j)\hat{s}_{ij}\,,
 }
 \end{array}
 \eq
where $\hat{s}_{ij}$ are generators of the symmetric group $S_n$ acting on functions
 (of $n$ auxiliary variables $x_1,...,x_n$)
by exchanging the corresponding arguments:
  \beq\label{a111}
  \begin{array}{c}
  \displaystyle{
\hat{s}_{ij}f(\dots,x_i,\dots,x_j,\dots)=f(\dots,x_j,\dots,x_i,\dots)\,.
 }
 \end{array}
 \eq
The Shibukawa--Ueno $R$-operators (\ref{a110}) satisfy the QYBE (\ref{a021}) and the unitarity property
  \beq\label{a112}
  \begin{array}{c}
  \displaystyle{
R_{ij}^\hbar(\xi)R^\hbar_{ji}(-\xi)=(\wp(\hbar)-\wp(\xi)){\rm Id}\,.
 }
 \end{array}
 \eq

Representations of $R$-matrices based on the symmetric group are well known. For example, the non-parametric AYBE (\ref{a101})
has the representation \cite{Kir,MMZ}:
  \beq\label{a113}
  \begin{array}{c}
  \displaystyle{
r_{ij}:=\phi(\xi_i-\xi_j,x_i-x_j)\hat{s}_{ij}\,,
 }
 \end{array}
 \eq
and the AYBE with parameters (\ref{a104}) is represented by the Shibukawa--Ueno $R$-operators (\ref{a110}), see \cite{MMZ}.
The functional equation (\ref{a109}) also arises as the commutativity condition
for the Dunkl operators \cite{BFV} constructed by means of (\ref{a111}).

Using the Shibukawa--Ueno $R$-operators one may reproduce the Baxter--Belavin $R$-matrix in the
fundamental representation by restricting these operators to a suitable $N$-dimensional vector space
spanned by theta functions \cite{FP}.

\noindent{\bf $R$-matrix-valued Lax pairs.}  Consider the classical
elliptic  Calogero--Moser model of ${\rm gl}_n$ type \cite{OP}.
It is defined by the Hamiltonian
  \beq\label{a114}
  \begin{array}{c}
  \displaystyle{
H=\sum\limits_{i=1}^n\frac{p_i^2}{2}-{g}^2\sum\limits_{i>j}^n\wp(q_i-q_j)
 }
 \end{array}
 \eq
 describing the interaction of $n$ particles. All variables and the coupling constant $g$ are complex valued.
 The momenta $p_i$ and particle positions $q_j$ are canonically conjugate:
 $\{p_i,q_j\}=\delta_{ij}$, $\{p_i,p_j\}=\{q_i,q_j\}=0$.
The Lax pair with spectral parameter
$L(z),M(z)\in{\rm Mat}_n$ reproducing the equations of motion through the Lax equation
  \beq\label{a1141}
  \begin{array}{c}
  \displaystyle{
{\dot L}(z)\equiv\{H,L(z)\}=[L(z),M(z)]
 }
 \end{array}
 \eq
was found in \cite{Krich1}. The Lax matrix has elements
  \beq\label{a1142}
  \begin{array}{c}
  \displaystyle{
L_{ij}(z)=\delta_{ij}p_i+{g}(1-\delta_{ij})\phi(z,q_{ij})\,,\quad
q_{ij}=q_i-q_j\,,
 }
 \end{array}
 \eq
where $\phi(z,u)$ is the elliptic Kronecker function (\ref{a108}) or its degeneration.
The accompanying $M$-matrix is
  \beq\label{a1143}
  \begin{array}{c}
  \displaystyle{
M_{ij}(z)={g} d_i\delta_{ij} +{g}(1-\delta_{ij})f(z,q_{ij})\,,\quad
d_i=-\sum\limits_{k: k\neq
i}^n f(0,q_{ik})\,,
 }
 \end{array}
 \eq
where $f(z,q)=\p_q\phi(z,q)$ (\ref{a904}).

The analogy between the function $\phi(z,q)$ and $R$-matrices satisfying the AYBE (\ref{a104})
led to the following higher-dimensional generalisation of (\ref{a1142})-(\ref{a1143}) in \cite{LOZ14}:
  \beq\label{a1144}
  \begin{array}{c}
    \displaystyle{
{\mL}(z)=\sum\limits_{i,j=1}^n E_{ij}\otimes \mL_{ij}(z)\,,
 \quad\quad
\mL_{ij}(z)= 1_{N}^{\otimes n}\, \delta_{ij}p_i
+g(1-\delta_{ij})R^{\,z}_{ij}(q_{ij})\in\Mat^{\otimes n}
 }
 \end{array}
 \eq
  \beq\label{a1145}
  \begin{array}{c}
  \displaystyle{
\bar\mM_{ij}(z)= D_i\delta_{ij}
+g(1-\delta_{ij})F^{\,z}_{ij}(q_{ij})+\delta_{ij}\,\mF^{\,0}\,,\quad
D_i=-g\sum\limits_{k: k\neq i}^n F^{\,0}_{ik}(q_{ik})\,,
 }
 \end{array}
 \eq
where $F^z_{ij}(q)=\p_q R_{ij}^z(q)$ and
  \beq\label{a1147}
  \begin{array}{c}
  \displaystyle{
\mF^{\,0}
 =g\sum\limits_{k>m}^n F^{\,0}_{km}(q_{km})
 =\frac{g}{2}\sum\limits_{k\neq m}^n F^{\,0}_{km}(q_{km})\,.
 }
 \end{array}
 \eq
Notice that
  \beq\label{a1148}
  \begin{array}{c}
  \displaystyle{
\bar\mM(z)=\mM(z)+1_n\otimes\mF^0\,,\quad \mM_{ij}(z)=D_i\delta_{ij}
+g(1-\delta_{ij})F^{\,z}_{ij}(q_{ij})\,,
 }
 \end{array}
 \eq
 where  $\mM(z)$ is a direct generalisation of the $M$-matrix (\ref{a1143}).
Expressions (\ref{a1144})-(\ref{a1145}) form
 the $R$-matrix-valued Lax pair for the classical Calogero--Moser model (\ref{a114}).
 The matrices in this Lax pair have size $nN^n\times nN^n$.
  Its quantum generalisation is obtained in a standard way by replacing $p_i$ with its
  quantum counterpart $\hbar\p_{q_i}$ \cite{GrZ18,ChM}.

In the scalar case ($N$=1), the $R$-matrices reduce to the functions $\phi(z,q_{ij})$, and
the additional term $1_n\otimes\mF^0$ in expression (\ref{a1148}) becomes proportional to
the identity matrix $1_n$ and, consequently, cancels out in the Lax equation.

In the general case ($N>1$) the term $1_n\otimes\mF^0$ is necessary.
Notice that the Lax equation ${\dot \mL}(z)=[\mL,\bar\mM(z)]$ again describes the
classical model (\ref{a114}) as in (\ref{a1141}). This result also admits a different interpretation. Write the Lax equation as
  \beq\label{a1149}
  \begin{array}{c}
  \displaystyle{
{\dot \mL}(z)+[\mF^0,\mL(z)]=[\mL,\mM(z)]\,.
 }
 \end{array}
 \eq
 or
  \beq\label{a11491}
  \begin{array}{c}
  \displaystyle{
\{H,\mL(z)\}+[\mF^0,\mL(z)]=[\mL,\mM(z)]\,.
 }
 \end{array}
 \eq
This equation is called half-quantum since it is partly classical (with respect to
the classical Hamiltonian $H$) and partly quantum (with respect to the quantum Hamiltonian $\mF^0$) \cite{GrZ18,Z18,GSZ}.
It was also shown in these works that $\mF^0$
represents the quantisation of the potential for the model of interacting tops. Further development and a
 general construction of half-quantum type models were presented in \cite{LRS}, where these models are called hybrid
 quantum systems.
 In the fully quantum case, the bracket
  $\{H,\mL(z)\}$ is replaced by the commutator $[{\hat H},\mL(z)]$, resulting in a quantum model of interacting tops; see \cite{GrZ18,ChM}.

\noindent {\bf Quantum long-range spin chains} arise in the following way \cite{SeZ18}.
Consider an equilibrium position of the classical spinless model (\ref{a114}).
For the elliptic and trigonometric $A$-type Calogero--Moser systems, one such configuration consists of
equally spaced points on a circle:
  \beq\label{a11492}
  \begin{array}{c}
  \displaystyle{
z_k=\frac{k}{n}\,,\quad k=1,...,n\,.
 }
 \end{array}
 \eq
 Restricting the half-quantum Lax equation (\ref{a11491}) to (\ref{a11492}) gives
 the quantum Lax equation
  \beq\label{a11493}
  \begin{array}{c}
  \displaystyle{
[\mH,\mL(z)]=[\mL,\mM(z)]\,,\qquad \mH=\mF^0|_{q_k=z_k}
 }
 \end{array}
 \eq
 with the quantum Hamiltonian (\ref{a1147}). In this way one can reproduce the well-known
 Haldane--Shastry spin chain  \cite{HS1} and the elliptic Inozemtsev chain \cite{Inoz90}, see \cite{Ch3,ChM}.
 The construction applies to any solution of the AYBEs (\ref{a104})-(\ref{a106})
 and yields a new family of quantum integrable long-range spin chains  (integrability does not follow
 from a quantum Lax equation alone; it requires a separate proof).
 This procedure for deriving spin chains through $R$-matrix-valued Lax pairs is
 in agreement with the Polychronakos freezing trick \cite{Poly}.

\subsection{\texorpdfstring{$BC_n$}{BCn}-type AYBE and its parametric version}
The algebra defined by the ${\rm B}$ (or ${\rm BC}_n$) analogue of (\ref{a101})-(\ref{a103})
is generated by the elements ${\bf r}_{ij}$, $\tilde{\bf r}_{ij}$ and ${\bf y}_i$
  for $1\le i\neq j\le n $, subject to the following defining relations \cite{Kir}:
\begin{itemize}
\item symmetry and antisymmetry relations for $i\neq j$:
  \beq\label{a115}
  \begin{array}{c}
  \displaystyle{
\bfr_{ij}=-\bfr_{ji}\,, \qquad \tilde{\bfr}_{ij}=\tilde{\bfr}_{ji} \,,
 }
 \end{array}
 \eq
\item permutation relations for distinct  $i, j, k, l$:
\beq\label{a116}
\begin{array}{c}
{\displaystyle \bfy_i \bfy_j=\bfy_j \bfy_i\,,
 \qquad \bfy_{i}\bfr_{kl}=\bfr_{kl}\bfy_i\,, \qquad \bfy_{i}\ti\bfr_{kl}=\ti\bfr_{kl}\bfy_i\,,}
\\ \\
{\displaystyle \bfr_{ij}\bfr_{kl}=\bfr_{kl}\bfr_{ij},\qquad
\tilde{\bfr}_{ij}\bfr_{kl}=\bfr_{kl}\tilde{\bfr}_{ij},\qquad  \tilde{\bfr}_{ij}\tilde{\bfr}_{kl}=\tilde{\bfr}_{kl}\tilde{\bfr}_{ij}\,,}
\end{array}
\eq
\item three-term relations for distinct $i,j,k$:
\beq\label{a117}
\begin{array}{c}
{\displaystyle
 \bfr_{ij} \bfr_{jk}=\bfr_{ik}\bfr_{ij}+\bfr_{jk}\bfr_{ik}\,,
 \qquad
 \bfr_{ij} \tilde{\bfr}_{jk}
 =\tilde\bfr_{ik}{\bfr}_{ij}+\tilde{\bfr}_{jk}\ti\bfr_{ik}\,,
}
\\ \\
{\displaystyle
\bfr_{jk}\tilde{\bfr}_{ik}=\tilde \bfr_{ij}\bfr_{jk}
+\tilde{\bfr}_{ik}\tilde{\bfr}_{ij}\,,
\qquad
\tilde{\bfr}_{jk}\bfr_{ik}
=\tilde{\bfr}_{ij}\tilde{\bfr}_{jk}+{\bfr}_{ik}\tilde{\bfr}_{ij},
}
\end{array}
\eq
\item four-term relation for $i\neq j$:
\beq\label{a118}
  \begin{array}{c}
  \displaystyle{
\bfr_{ij}\bfy_j=\bfy_i\bfr_{ij}+\tilde{\bfr}_{ij}\bfy_i+\bfy_j\tilde{\bfr}_{ij}\,.
}
\end{array}
\eq
\end{itemize}

In \cite{MMZ} a parametric version of (\ref{a115})-(\ref{a118}) was proposed:
 \begin{itemize}
 \item Symmetry and antisymmetry relations:
  \beq\label{a122}
 \tilde{R}^z_{ij}(\xi)=\tilde{R}^z_{ji}(\xi)\,,
 \eq
 \beq\label{a123}
 R^z_{ij}(\xi)=-R_{ji}^{-z}(-\xi)\,,
 \eq
 \item Permutation relations for distinct $i,j,k,l$:
 \beq\label{a124}
  \tilde{K}^z_i(\xi_i) \tilde{K}^w_j(\xi_j)=  \tilde{K}^w_j(\xi_j)\tilde{K}_i^z(\xi_i)\,,
  \qquad
  \tilde{K}^z_i(\xi_i)  R^w_{kl}(\xi_{kl})= R^w_{kl}(\xi_{kl})\tilde{K}_i^z(\xi_i)\,,
 \eq
 \beq\label{a125}
   \tilde{K}^z_i(\xi_i)  \tilde{R}^w_{kl}(\xi_{kl})=
   \tilde{R}^w_{kl}(\xi_{kl})\tilde{K}_i^z(\xi_i)\,,
   \qquad  R^z_{ij}(\xi_{ij}) R^w_{kl}(\xi_{kl})=R^w_{kl}(\xi_{kl}) R^z_{ij}(\xi_{ij})\,,
 \eq
 \beq\label{a126}
 R^z_{ij}(\xi_{ij}) \tilde{R}^w_{kl}(\xi_{kl})=\tilde{R}^w_{kl}(\xi_{kl}) R^z_{ij}(\xi_{ij})\,,
\qquad
 \tilde{R}^z_{ij}(\xi_{ij}) \tilde{R}^w_{kl}(\xi_{kl})=\tilde{R}^w_{kl}(\xi_{kl})
 \tilde{R}^z_{ij}(\xi_{ij})\,,
 \eq
 where $\xi_{ij}=\xi_i-\xi_j$.
  \item Three-term relations:
 \beq\label{a127}
 R_{ij}^z(\xi_i-\xi_j) R_{jk}^w(\xi_j-\xi_k)=R_{ik}^w(\xi_i-\xi_k)R_{ij}^{z-w}(\xi_i-\xi_j)+
 R_{jk}^{w-z}(\xi_j-\xi_k)R_{ik}^z(\xi_i-\xi_k)\,,
 \eq
 \beq\label{a128}
 R_{ij}^z(\xi_i-\xi_j) \tilde{R}_{jk}^w(\xi_j+\xi_k)=
 \tilde{R}_{ik}^w(\xi_i+\xi_k)R_{ij}^{z-w}(\xi_i-\xi_j)
 +\tilde{R}_{jk}^{w-z}(\xi_j+\xi_k)\tilde{R}_{ik}^z(\xi_i+\xi_k)\,,
 \eq
 \beq\label{a129}
 R_{jk}^{z}(\xi_j-\xi_k)\tilde{R}_{ik}^{w}(\xi_i+\xi_k)=\tilde{R}_{ij}^{w}(\xi_i+\xi_j) R_{jk}^{z-w}(\xi_j-\xi_k)+\tilde{R}_{ik}^{w-z}(\xi_i+\xi_k)\tilde{R}_{ij}^{z}(\xi_i+\xi_j)\,,
 \eq
 \beq\label{a130}
 \tilde{R}_{jk}^{z}(\xi_j+\xi_k) R_{ik}^{w}(\xi_i-\xi_k)=\tilde{R}_{ij}^{w}(\xi_i+\xi_j) \tilde{R}_{jk}^{z-w}(\xi_j+\xi_k)+R_{ik}^{w-z}(\xi_i-\xi_k)\tilde{R}_{ij}^{z}(\xi_i+\xi_j)\,,
 \eq
 \item The four-term relation:
 \beq\label{a131}
 R_{ij}^{w+z}(\xi_i-\xi_j)\tilde{K}^w_j(\xi_j)=\tilde{K}^w_i(\xi_i)R_{ij}^{z-w}(\xi_i-\xi_j)+
 \tilde{R}_{ij}^{w-z}(\xi_i+\xi_j)\tilde{K}^z_i(\xi_i)+
 \tilde{K}_j^{-z}(\xi_j)\tilde{R}_{ij}^{w+z}(\xi_i+\xi_j)\,.
 \eq
  \end{itemize}

\subsection{Shibukawa--Ueno-type representation}
It was mentioned in \cite{MMZ} that the algebra (\ref{a115})-(\ref{a118})
admits the following representation, analogous to (\ref{a113}):
\beq\label{a120}
  \begin{array}{c}
  \displaystyle{
\bfr_{ij}:=\phi(\xi_i-\xi_j,x_i-x_j)\hat{s}_{ij},\qquad
 \tilde{\bfr}_{ij}:=\phi(\xi_i+\xi_j,x_i+x_j)\hat{s}_{ij}\hat{t}_i\hat{t}_j,
 \qquad \bfy_i:=v(\xi_i,x_i)\hat{t}_i\,,
}
\end{array}
\eq
where $\hat{t}_i$ is the reflection
\beq\label{a121}
  \displaystyle{
\hat{t}_i f(x_1\dots,x_i,\dots,x_n)=f(x_1\dots,-x_i,\dots,x_n)\,.
}
\eq
The parametric relations (\ref{a122})--(\ref{a131}) are realised by the
Shibukawa--Ueno-type operators
\beq\label{a150}
  \displaystyle{
R^\hbar_{ij}(\xi)=\phi(\hbar,x_i-x_j)-\phi(-\xi,x_i-x_j)\hat{s}_{ij}\,,
}
\eq
\beq\label{a151}
  \displaystyle{
\tilde{R}_{ij}^{\hbar}(\xi)=\hat{t}_j R_{ij}^\hbar(\xi) \hat{t}_j
=\phi(\hbar,x_i+x_j)-\phi(-\xi,x_i+x_j)\hat{s}_{ij}\hat{t}_i\hat{t}_j\,.
}
\eq
Let us introduce the boundary operator
\beq\label{a153}
  \displaystyle{
\tilde{K}^\hbar_i(\xi)
\equiv
\tilde{K}^\hbar_i(\xi|\nu,\nu')
=
v(\hbar,x_i|\nu')-v(-\xi,x_i|\nu)\,\hat{t}_i\,,
}
\eq
where $v(\,\cdot\,,\,\cdot\,|\nu)$ is the function (\ref{q413}) and
$\nu=(\nu_0,\nu_1,\nu_2,\nu_3)$,
$\nu'=(\nu'_0,\nu'_1,\nu'_2,\nu'_3)$
are two (a priori independent) sets of coupling constants.
In \cite{MMZ} it was shown that the operators (\ref{a150})--(\ref{a153})
satisfy the parametric ${\rm BC}_n$ relations (\ref{a122})--(\ref{a131})
for $\nu'=\nu$ in (\ref{a153}); this proof carries over without modification to the case of arbitrary
independent coupling constants $\nu$ and $\nu'$. In particular, the four-term
identity (\ref{a131}) holds with $\tilde{K}^\hbar_i(\xi|\nu,\nu')$
as in (\ref{a153}). At the same time, the reflection equation
(for arbitrary sets $\nu$ and $\nu'$) does not hold.

These operators also satisfy a number of important identities which we formulate below.

\medskip
\noindent\textit{Quantum Yang--Baxter equations.}
Following \cite{SU} one has
\beq\label{a154}
  \displaystyle{
R^\hbar_{ij}(\xi_i-\xi_j)R^\hbar_{ik}(\xi_i-\xi_k)R^\hbar_{jk}(\xi_j-\xi_k)=
R^\hbar_{jk}(\xi_j-\xi_k)R^\hbar_{ik}(\xi_i-\xi_k)R^\hbar_{ij}(\xi_i-\xi_j)\,,
}
\eq
\beq\label{a155}
  \displaystyle{
R^\hbar_{ij}(\xi_i-\xi_j)\tilde{R}^\hbar_{ik}(\xi_i-\xi_k)\tilde{R}^\hbar_{jk}(\xi_j-\xi_k)=
\tilde{R}^\hbar_{jk}(\xi_j-\xi_k)\tilde{R}^\hbar_{ik}(\xi_i-\xi_k)R^\hbar_{ij}(\xi_i-\xi_j)\,.
}
\eq
(These involve only $R$ and $\tilde{R}$, hence are independent of
$\nu$ and $\nu'$.)

\medskip
\noindent\textit{Reflection equations.}
The reflection equation goes back to \cite{Cherednik,Skl-refl}.
In \cite{KH2} Komori and Hikami proposed the boundary operator
\beq\label{a152}
  \displaystyle{
K^\hbar_i(\xi)=
v(-\xi,x_i|\nu)-v(\hbar,x_i|\nu)\,\hat{t}_i
}
\eq
as a solution of the reflection equation with the $R$-operator
(\ref{a150}):
\beq\label{a156}
\displaystyle{
R_{ij}^{\hbar}(\xi_i-\xi_j)K_i^\hbar(\xi_i)R_{ji}^{\hbar}(\xi_i+\xi_j)K_j^\hbar(\xi_j)=
K_j^\hbar(\xi_j)R_{ij}^{\hbar}(\xi_i+\xi_j)K_i^\hbar(\xi_i)R_{ji}^{\hbar}(\xi_i-\xi_j)\,.
}
\eq
This operator is a particular case of (\ref{a153}) when $\nu'=\nu$:
\beq\label{a152b}
\begin{array}{c}
\displaystyle{
K^\hbar_i(\xi)=-\tilde{K}^\hbar_i(\xi|\nu,\nu)\,\hat{t}_i\,.
}
\end{array}
\eq
More generally, it can be shown that the reflection equation for
$\tilde{K}^\hbar_i(\xi|\nu,\nu')$
\beq\label{a157}
  \displaystyle{
R_{ij}^{\hbar}(\xi_i-\xi_j)\tilde{K}_i^\hbar(\xi_i)\tilde{R}_{ij}^{\hbar}(\xi_i+\xi_j)
\tilde{K}_j^\hbar(\xi_j)=\tilde{K}_j^\hbar(\xi_j)\tilde{R}_{ij}^{\hbar}(\xi_i+\xi_j)
\tilde{K}_i^\hbar(\xi_i)R_{ij}^{\hbar}(\xi_i-\xi_j)\,,
}
\eq
holds if and only if the two sets of couplings are related by signs, that is,
\beq\label{a162}
\begin{array}{c}
\displaystyle{
\nu'_a=\varepsilon_a\nu_a\,,
\qquad \varepsilon_a=\pm 1\,,
\qquad a=0,1,2,3\,.
}
\end{array}
\eq
Under the same condition (\ref{a162}) the operator
$K^\hbar_i(\xi|\nu,\nu')=-\tilde{K}^\hbar_i(\xi|\nu,\nu')\,\hat{t}_i$
satisfies the reflection equation (\ref{a156}).

\medskip
\noindent\textit{Unitarity.}
The operators $R$ and $\tilde{R}$ are unitary for any spectral parameters:
\beq\label{a158}
  \displaystyle{
  R_{ij}^{\hbar}(\xi)  R_{ji}^{\hbar}(-\xi)={\rm Id}\Big(\wp(\hbar)-\wp(\xi)\Big)\,,
  }
  \eq
  \beq\label{a159}
    \displaystyle{
  \tilde{R}_{ij}^{\hbar}(\xi)  \tilde{R}_{ji}^{-\hbar}(\xi)={\rm Id}\Big(\wp(\xi)-\wp(\hbar)\Big)\,.
  }
  \eq
For the boundary operators the unitarity relations
\beq\label{a160}
  \displaystyle{
 K^\hbar_i(\xi) K^\hbar_i(-\xi)
 ={\rm Id}\sum_{a=0}^3
 \Big(\bar{\nu'}_a^{\,2}\wp(\hbar+\Omega_a)
 -\bar\nu_a^2\wp(\xi+\Omega_a)\Big)\,,
 }
 \eq
 \beq\label{a161}
   \displaystyle{
 \tilde{K}^\hbar_i(\xi) \tilde{K}^{-\hbar}_i(\xi)=
 {\rm Id}\sum_{a=0}^3
 \Big(\bar\nu_a^2\wp(\xi+\Omega_a)
 -\bar{\nu'}_a^{\,2}\wp(\hbar+\Omega_a)\Big)
 }
 \eq
hold if and only if (\ref{a162}) is satisfied.
Here $\bar\nu$ and $\bar{\nu'}$ are obtained from $\nu$ and $\nu'$ respectively
by the linear transformation (\ref{q416}).
In the special case $\nu'=\nu$ one has $\bar{\nu'}=\bar\nu$, and
the factor in (\ref{a160}) reduces to the expression proposed in \cite{KH2}:
$\sum_{a=0}^3\bar\nu_a^2\bigl(\wp(\hbar+\Omega_a)-\wp(\xi+\Omega_a)\bigr)$.

\medskip
\noindent\textit{Choice of signs for the Felder--Pasquier reduction.}
The signs $\varepsilon_a$ in (\ref{a162}) will be fixed so that, after the
Felder--Pasquier conjugation of Section~\ref{sec25},
the operators $K^{\hbar,N}$ and $\tilde{K}^{\hbar,N}$ preserve the spaces
of entire theta functions $V_N$.
(The finite-dimensional reduction of $K$ goes back to \cite{KH2}, but the
dependence of the couplings on the parity of $N$ must be chosen as below;
cf.\ the residue computation in the proof of Lemma~\ref{propMap}.)
The choice depends on the parity of $N$:
\beq\label{a163}
\begin{array}{c}
\displaystyle{
\begin{aligned}
N\text{ even:}&\qquad
\varepsilon_0=\varepsilon_1=\varepsilon_2=\varepsilon_3=+1
\qquad\bigl(\nu'=\nu\bigr)\,,
\\[1mm]
N\text{ odd:}&\qquad
\varepsilon_0=\varepsilon_1=\varepsilon_3=+1\,,\quad
\varepsilon_2=-1
\qquad\bigl(\nu'_2=-\nu_2,\;
\nu'_a=\nu_a\text{ for }a\neq 2\bigr)\,.
\end{aligned}
}
\end{array}
\eq

\subsection{\texorpdfstring{$R$}{R}-matrix-valued Lax pair for the Calogero--Inozemtsev system}
The elliptic $n$-body Calogero--Moser model (\ref{a114}) is associated with the root system ${\rm A}_{n-1}$.
Extensions to other root systems are described in \cite{OP}.  Here we consider
the model with five arbitrary coupling constants
proposed by V. Inozemtsev in \cite{Inoz89}. It is associated with the ${\rm BC}_n$ root system and is given by the Hamiltonian
\begin{equation}
 \label{a170}
   \displaystyle{
 H = \frac{1}{2} \sum_{k=1}^{n} p_k^2 - g^2
\sum_{i<j}^{n} \Big(\wp(q_i-q_j) + \wp(q_i+q_j)\Big) -
\frac{1}{2}\sum\limits_{a=0}^3\sum_{k=1}^{n} \nu_a^2\wp(q_k+\Om_a)\,,
}
 \end{equation}
 where $\Om_\ga$ are the half-periods (\ref{w209}), and the five arbitrary constants are
 $g,\nu_0,\nu_1,\nu_2,\nu_3\in\mC$. In our construction below we use the $2n\times 2n$
 Lax representation with a spectral parameter due to K. Takasaki \cite{Ta}
 (see also A.J. Bordner, R. Sasaki \cite{BoSa}\footnote{Additionally, see
 \cite{Z04} for the $n=1$ case, \cite{Ch2} for a study of the spectral curve,
and \cite{Ch1} for the relativistic version (the Ruijsenaars-van Diejen model).}). This Lax pair has the form
  \beq\label{a171}
  \begin{array}{c}
  \displaystyle{
L(z)=\mat{L^{11}(z)}{L^{12}(z)}{L^{21}(z)}{L^{22}(z)}\,,\qquad
M(z)=\mat{M^{11}(z)}{M^{12}(z)}{M^{21}(z)}{M^{22}(z)}\,.
 }
 \end{array}
 \eq
The blocks $L^{ab}(z)$ and $M^{ab}(z)$, $a,b=1,2$, are $n\times n$ matrices with
the following entries:
  \beq\label{a172}
  \begin{array}{c}
  \displaystyle{
L^{11}_{ij}(z)=\delta_{ij}p_i+g(1-\delta_{ij})\phi(z,q_{ij})\,,
 }
 \end{array}
 \eq
  \beq\label{a173}
  \begin{array}{c}
  \displaystyle{
L^{12}_{ij}(z)=\delta_{ij}v(z,q_i)+g(1-\delta_{ij})\phi(z,q^+_{ij})\,,
 }
 \end{array}
 \eq
  \beq\label{a174}
  \begin{array}{c}
  \displaystyle{
L^{21}_{ij}(z)=-\delta_{ij}v(-z,q_i)-g(1-\delta_{ij})\phi(-z,q^+_{ij})\,,
 }
 \end{array}
 \eq
  \beq\label{a175}
  \begin{array}{c}
  \displaystyle{
L^{22}_{ij}(z)=-\delta_{ij}p_i-g(1-\delta_{ij})\phi(-z,q_{ij})
 }
 \end{array}
 \eq
and
  \beq\label{a176}
  \begin{array}{c}
  \displaystyle{
M^{11}_{ij}(z)=\delta_{ij}a_i+g(1-\delta_{ij})f(z,q_{ij})\,,
 }
 \end{array}
 \eq
  \beq\label{a177}
  \begin{array}{c}
  \displaystyle{
M^{12}_{ij}(z)=\frac12\,\delta_{ij}v'(z,q_i)+g(1-\delta_{ij})f(z,q^+_{ij})\,,
 }
 \end{array}
 \eq
  \beq\label{a178}
  \begin{array}{c}
  \displaystyle{
M^{21}_{ij}(z)=\frac12\,\delta_{ij}v'(-z,q_i)+g(1-\delta_{ij})f(-z,q^+_{ij})\,,
 }
 \end{array}
 \eq
  \beq\label{a179}
  \begin{array}{c}
  \displaystyle{
M^{22}_{ij}(z)=\delta_{ij}a_i+g(1-\delta_{ij})f(-z,q_{ij})\,,
 }
 \end{array}
 \eq
where
  \beq\label{a180}
  \begin{array}{c}
  \displaystyle{
q_{ij}=q_i-q_j\,,\quad q^+_{ij}=q_i+q_j
 }
 \end{array}
 \eq
and
  \beq\label{a18001}
  \begin{array}{c}
  \displaystyle{
a_i=g\sum\limits_{k:k\neq i}^n\Big(\wp(q_{i}-q_k)+\wp(q_i+q_k)\Big)
+\frac12\sum\limits_{a=0}^3\nu_a\wp(q_i+\Omega_a)\,.
 }
 \end{array}
 \eq

An $R$-matrix-valued generalisation of Takasaki's Lax pair is as follows \cite{MMZ}. Let
  \beq\label{a181}
  \begin{array}{c}
  \displaystyle{
\mL(z)=\mat{\mL^{11}(z)}{\mL^{12}(z)}{\mL^{21}(z)}{\mL^{22}(z)}\,,\qquad
\mM(z)=
\mat{\mM^{11}(z)}{\mM^{12}(z)}{\mM^{21}(z)}{\mM^{22}(z)}
 }
 \end{array}
 \eq
and
  \beq\label{a1811}
  \begin{array}{c}
  \displaystyle{
\bar\mM(z)=\mM(z)+\mF^0\, 1_{2n\times 2n}\,,
 }
 \end{array}
 \eq
where the blocks $\mL^{ab}(z)$ and $\mM^{ab}(z)$, $a,b=1,2$, are given by
  \beq\label{a182}
  \begin{array}{c}
  \displaystyle{
\mL^{11}_{ij}(z)=\delta_{ij}p_i+g(1-\delta_{ij})R_{ij}^z(q_{ij})\,,
 }
 \end{array}
 \eq
  \beq\label{a183}
  \begin{array}{c}
  \displaystyle{
\mL^{12}_{ij}(z)=\delta_{ij}\tilde{K}^z_i(q_i)+g(1-\delta_{ij})\tilde{R}^z_{ij}(q^+_{ij})\,,
 }
 \end{array}
 \eq
  \beq\label{a184}
  \begin{array}{c}
  \displaystyle{
\mL^{21}_{ij}(z)=-\delta_{ij}\tilde{K}^{-z}_i(q_i)-g(1-\delta_{ij})\tilde{R}^{-z}_{ij}(q^+_{ij})\,,
 }
 \end{array}
 \eq
  \beq\label{a185}
  \begin{array}{c}
  \displaystyle{
\mL^{22}_{ij}(z)=-\delta_{ij}p_i-g(1-\delta_{ij})R^{-z}_{ij}(q_{ij})
 }
 \end{array}
 \eq
and
  \beq\label{a186}
  \begin{array}{c}
  \displaystyle{
\mM^{11}_{ij}(z)=\delta_{ij}A_i+g(1-\delta_{ij})F^z_{ij}(q_{ij})\,,
 }
 \end{array}
 \eq
  \beq\label{a187}
  \begin{array}{c}
  \displaystyle{
\mM^{12}_{ij}(z)=\frac12\,\delta_{ij}\tilde{Y}_i^z(q_i)+g(1-\delta_{ij})\tilde{F}^z_{ij}(q^+_{ij})\,,
 }
 \end{array}
 \eq
  \beq\label{a188}
  \begin{array}{c}
  \displaystyle{
\mM^{21}_{ij}(z)=\frac12\,\delta_{ij}\tilde{Y}^{-z}_i(q_i)+g(1-\delta_{ij})\tilde{F}^{-z}_{ij}(q^+_{ij})\,,
 }
 \end{array}
 \eq
  \beq\label{a189}
  \begin{array}{c}
  \displaystyle{
\mM^{22}_{ij}(z)=\delta_{ij}A_i+g(1-\delta_{ij})F^{-z}_{ij}(q_{ij})\,.
 }
 \end{array}
 \eq
In these formulae, we use the notation
\beq\label{a190}
F_{ij}^z(q):=\partial_{q}R_{ij}^z(q)\,, \qquad
\tilde{F}_{ij}^z(q):=\partial_{q}\tilde{R}_{ij}^z(q)\,,
\eq
\beq\label{a191}
Y^\hbar_i(q):=\partial_q K^\hbar_i(q|\nu,\nu')\,, \qquad
\tilde{Y}^\hbar_i(q):=\partial_q \tilde{K}^\hbar_i(q|\nu,\nu')\,.
\eq
  \beq\label{a192}
  \begin{array}{c}
  \displaystyle{
A_i=-\frac12\,\tilde{Y}_i^0(q_i)-g\sum_{k:k\neq i}^n \Big(F^0_{ik}(q_{ik})+{\ti F}^0_{ik}(q_{ik}^+)\Big)
 }
 \end{array}
 \eq
 and
  \beq\label{a193}
  \begin{array}{c}
  \displaystyle{
\mF^0=g\sum_{k<l}^n \Big(F^0_{kl}(q_{kl})+{\ti F}^0_{kl}(q_{kl}^+)\Big)+
\frac12\sum\limits_{k=1}^n \tilde{Y}_k^0(q_k)\,.
 }
 \end{array}
 \eq
 The main result is that the classical Lax equation
  ${\dot \mL}(z)=[\mL(z),\bar\mM(z)]$ is equivalent to
 the equations of motion for the Calogero--Inozemtsev model (\ref{a170}),
 with the constants $\nu_a$ replaced by $\bar\nu_a$, as follows from (\ref{a161}).
 To recover the equations of motion with the couplings in (\ref{a170}), one can
 replace $\nu_a$ by $\bar\nu_a$ in the definition of the ${\ti K}$-matrix,
 since the linear transformation (\ref{q416}) denoted by the bar is involutive.

The long-range spin chains are derived by the same procedure as for the $A$-type models.
Rewrite the classical Lax equation in the form of the half-quantum Lax equation (\ref{a11491}).
Then restriction to an equilibrium position $z_k$
 of the original classical model (\ref{a170}) (or its degeneration)
provides the quantum Lax equation (\ref{a11493}) with the Hamiltonian
 \beq\label{a194}
\begin{array}{c}
\displaystyle{
 \mH=\sum\limits_{i<j}^n \Big(F^0_{ij}(z_i-z_j)+{\ti F}^0_{ij}(z_i+z_j)\Big)+
 \frac12\sum\limits_{k=1}^n \tilde{Y}_k^0(z_k)\,,
 \qquad
 \tilde{Y}^0(w)=\p_w {\ti K}^\hbar(w)\Big|_{\hbar=0}\,,
}
\end{array}
 \eq
 where we set $g=1$, assuming $g\neq 0$ in (\ref{a193}). This is achieved by dividing the Hamiltonian by $g$ and redefining the constants
 $\nu_a\to g\nu_a$.

%\section{Elliptic finite-dimensional $R$-operator and corresponding boundary $K$-operator}\label{sec25}
\section{The elliptic \texorpdfstring{$K$}{K}-matrix via the Felder--Pasquier procedure}\label{sec25}
\setcounter{equation}{0}

In this section we construct finite-dimensional matrix solutions of the
quantum Yang--Baxter and reflection equations by restricting the
Shibukawa--Ueno-type operators from Section~\ref{sec2} to spaces of elliptic
theta functions, following Felder and Pasquier \cite{FP} and Komori and Hikami
\cite{KH2}.

\subsection{Felder--Pasquier procedure}\label{sec251}

The Shibukawa--Ueno-type operators (\ref{a150})--(\ref{a153}) act on (infinite-dimensional)
spaces of functions of the auxiliary variables $x_i$. Finite-dimensional matrix solutions of
the quantum Yang--Baxter and reflection equations are obtained by restricting these operators
to suitable spaces of elliptic theta functions. For the $R$-operator this procedure goes back to
Felder and Pasquier \cite{FP}; the analogous construction for the boundary $K$-operator was
developed by Komori and Hikami \cite{KH2}. Below we recall the main statements;
a complete proof of the passage to the modified operators is given for
Proposition~\ref{propMod}, an outline of the identification of
$R_{12}^{\hbar,N}$ with the Baxter--Belavin matrix via the discrete Fourier identities
(\ref{a911m}), (\ref{a911t}) is given for Proposition~\ref{propFP},
and an analogous outline for the identification of $\tilde{K}_i^{\hbar,N}$
with (\ref{Ktilde}) is given for Proposition~\ref{propFPK}.
Explicit matrix formulae are given below in this section.

\paragraph{Spaces of theta functions.}
Following Felder and Pasquier \cite{FP}, fix an integer $N\ge 1$ and $\tau\in\mC$ with $\mathrm{Im}\,\tau>0$. For a complex parameter $\xi$
let $V_N(\xi)$ be the $N$-dimensional space of entire functions $f(x)$ satisfying
\beq\label{fp10}
\begin{array}{c}
\displaystyle{
f(x+1)=f(x)\,,\qquad
f(x+\tau)=e^{-2\pi\imath N x-\pi\imath N\tau+2\pi\imath\xi}\,f(x)\,.
}
\end{array}
\eq
In particular, we write $V_N:=V_N(0)$. The space $V_N(\xi)$ is $N$-dimensional. Indeed, any entire $1$-periodic function admits a
Fourier expansion $f(x)=\sum_{n\in\mZ}a_ne^{2\pi\imath nx}$; the second condition in
(\ref{fp10}) with $\xi=0$ yields the recurrence
\beq\label{fp10a}
\begin{array}{c}
\displaystyle{
a_{n+N}=e^{\pi\imath(2n+N)\tau}\,a_n\,,
}
\end{array}
\eq
so the coefficients are determined by $N$ independent values $a_0,\dots,a_{N-1}$.
On the space $V_N$ there is an irreducible action of the Heisenberg group $H_N$. It is generated by $A$, $B$ and a central element $\varepsilon$ with
\beq\label{fp12}
\begin{array}{c}
\displaystyle{
A^N=B^N=1\,,\qquad AB=\varepsilon BA\,,\qquad \varepsilon=e^{2\pi\imath/N}\,.
}
\end{array}
\eq
The action on $V_N$ is
\beq\label{fp13}
\begin{array}{c}
\displaystyle{
(Af)(x)=f\!\left(x+\frac1N\right)\,,\qquad
(Bf)(x)=e^{2\pi\imath x+\pi\imath\tau/N}\,f\!\left(x+\frac{\tau}{N}\right)\,.
}
\end{array}
\eq
Following Komori and Hikami \cite{KH2}, we take as a working basis of $V_N$
the theta functions obtained by diagonalising $A$, with the orientation
$B\theta_\al=\theta_{\al-1}$ (the KH basis reindexed by
$\al\mapsto-\al$):
\beq\label{fp17}
\begin{array}{c}
\displaystyle{
\theta_\al(x)
=\sum_{n\in\mZ}
\exp\!\Big(
\pi\imath\frac{(Nn-\al)^2\tau}{N}+2\pi\imath(Nn-\al)x
\Big)
=\vartheta\!\left[\begin{array}{c}-\al/N\\ 0\end{array}\right](Nx\,;\,N\tau)\,,
}
\end{array}
\eq
for $\al\in\mZ_N=\mZ/N\mZ$.
The action of $A$ and $B$ on this basis is
\beq\label{fp14}
\begin{array}{c}
\displaystyle{
A\,\theta_\al=e^{-2\pi\imath\al/N}\,\theta_\al\,,\qquad
B\,\theta_\al=\theta_{\al-1}\,,
\qquad \al\in\mZ_N\,.
}
\end{array}
\eq

\paragraph{Modification of the $R$ and $K$ operators.}
Identify $V_N(\xi_1)\otimes V_N(\xi_2)$ with the space of entire functions of
$(x_1,x_2)$ such that $x_i\mapsto f(x_1,x_2)$ belongs to $V_N(\xi_i)$ for each fixed
value of the other argument. Introduce notation for the lattice of periods
\beq\label{fp141}
\begin{array}{c}
\displaystyle{
\Lambda=\mZ+\tau\mZ\,,
}
\end{array}
\eq
 and fix the half-periods as in
(\ref{w209}). For the boundary operators we take the couplings related by
(\ref{a162})--(\ref{a163}):
\beq\label{fpKcases}
\begin{array}{c}
\displaystyle{
\begin{aligned}
N\text{ even:}&\qquad
\tilde{K}^\hbar(\xi)
=
v(\hbar,x|\nu)-v(-\xi,x|\nu)\,\hat t\,,
\\[1mm]
N\text{ odd:}&\qquad
\tilde{K}^\hbar(\xi)
=
v(\hbar,x|\nu')-v(-\xi,x|\nu)\,\hat t\,,
\quad
\nu'_2=-\nu_2\,,\quad
\nu'_a=\nu_a\ (a\neq 2)\,,
\end{aligned}
}
\end{array}
\eq
and $K^\hbar(\xi)=-\tilde{K}^\hbar(\xi)\,\hat t$ in both cases
(cf.\ (\ref{a152b})).

\begin{lemma}[\cite{FP}]
\label{propMap}
The operators (\ref{a150}), (\ref{a151}) and (\ref{fpKcases}) act on the spaces
$V_N(\,\cdot\,)$ by
\beq\label{fp19b}
\begin{array}{rcl}
R_{12}^{\hbar}(\xi_1-\xi_2)
&\!\!:\!\!&
V_N(\xi_1)\otimes V_N(\xi_2-\hbar)
\;\longrightarrow\;
V_N(\xi_1-\hbar)\otimes V_N(\xi_2)\,,
\\[1.5mm]
\tilde{R}_{12}^{\hbar}(\xi_1+\xi_2)
&\!\!:\!\!&
V_N(\xi_1)\otimes V_N(\xi_2+\hbar)
\;\longrightarrow\;
V_N(\xi_1-\hbar)\otimes V_N(\xi_2)\,,
\\[1.5mm]
K^{\hbar}(\xi)
&\!\!:\!\!&
V_N(-\hbar-\xi)
\;\longrightarrow\;
V_N(-\hbar+\xi)\,,
\\[1.5mm]
\tilde{K}^{\hbar}(\xi)
&\!\!:\!\!&
V_N(\xi+\hbar)
\;\longrightarrow\;
V_N(\xi-\hbar)\,.
\end{array}
\eq
\end{lemma}

\begin{proof}
The statement follows from the results of Felder and Pasquier \cite{FP} and Komori and Hikami \cite{KH2};
the only additional point is that we treat the even and odd cases for $N$
separately.
In particular, some of the residue identities used in \cite{KH2} hold only for
even $N$; for odd $N$ the couplings must be adjusted as in (\ref{fpKcases}).
It is enough to show that the image of an element of the space on the
left-hand side of (\ref{fp19b}) is entire and transforms under
$x\mapsto x+1$, $x\mapsto x+\tau$ with the multipliers of the corresponding
space on the right-hand side of (\ref{fp19b}).

\medskip
\noindent\textit{The $R$-operator.}
Let $f\in V_N(\xi_1)\otimes V_N(\xi_2-\hbar)$ and set
$g=R_{12}^{\hbar}(\xi_1-\xi_2)f$.
As functions of $x_{12}=x_1-x_2$, the coefficients $\phi(\hbar,x_{12})$ and
$\phi(-\xi_1+\xi_2,x_{12})$ have simple poles on the lattice $\Lambda$ (\ref{fp141}) with residues $\exp(-2\pi\imath\hbar m_2)$
and $\exp(2\pi\imath \xi_{12} m_2)$, respectively,
at points $x_{12}=m_1+m_2\tau$, $m_1,m_2\in\mZ$.
Therefore, near $x_1=x_2$ one has
\[
g=\frac{f(x_1,x_2)-f(x_2,x_1)}{x_{12}}+O(1)\,,
\]
so the apparent singularity is removable.
The quasi-periodicity $\phi(z+1,u)=\phi(z,u)$,
$\phi(z+\tau,u)=e^{-2\pi\imath u}\phi(z,u)$ together with
(\ref{fp10}) implies that $g$ is $1$-periodic in each variable and transforms
under $x_i\mapsto x_i+\tau$ with the multipliers of
$V_N(\xi_1-\hbar)\otimes V_N(\xi_2)$.
In particular, $g$ remains regular at all lattice translates
$x_1=x_2+n+m\tau$, hence is entire, and
$g\in V_N(\xi_1-\hbar)\otimes V_N(\xi_2)$.

\medskip
\noindent\textit{The $\tilde{R}$-operator.}
The same argument applies to
$\tilde{R}_{12}^{\hbar}(\xi_1+\xi_2)=\hat t_2 R_{12}^{\hbar}(\xi_1+\xi_2)\hat t_2$,
whose poles on $x_1+x_2\in\Lambda$ cancel because of the combination
$f(x_1,x_2)-f(-x_2,-x_1)$.

\medskip
\noindent\textit{The $\tilde{K}$-operator.}
Let $f\in V_N(\xi+\hbar)$ and
$\tilde{K}^\hbar(\xi)f=v(\hbar,x|\nu')f(x)-v(-\xi,x|\nu)f(-x)$,
with $\nu'$ as in (\ref{fpKcases}).
By (\ref{q413}) each summand has simple poles at $x\in-\Om_a+\Lambda$.
Writing
$M_a(z)=\exp(4\pi\imath z\,\partial_\tau\Om_a)$,
so that $M_0=M_1=1$ and $M_2(z)=M_3(z)=e^{2\pi\imath z}$, the residues of
$v(z,\,\cdot\,|\mu)$ at $x=-\Om_a$ equal $\mu_a M_a(z)$.
Thus the residue of $\tilde{K}^\hbar(\xi)f$ at $x=-\Om_a$ is
\[
\nu'_a M_a(\hbar)\,f(-\Om_a)-\nu_a M_a(-\xi)\,f(\Om_a)\,.
\]
For $f\in V_N(\xi+\hbar)$ the quasi-periodicity (\ref{fp10}) gives
\begin{align*}
f(-\Om_0)&=f(\Om_0)\,,&
f(-\Om_1)&=f(\Om_1)\,,\\
f(-\Om_3)&=e^{-2\pi\imath(\xi+\hbar)}\,f(\Om_3)\,,&
f(-\Om_2)&=(-1)^N e^{-2\pi\imath(\xi+\hbar)}\,f(\Om_2)\,.
\end{align*}
One finds that the residues vanish for
$a=0,1,3$ as soon as $\nu'_a=\nu_a$, while for $a=2$ they vanish if and only if
$\nu'_2=(-1)^N\nu_2$.
This is precisely the parity rule (\ref{fpKcases})\,/\,(\ref{a163}).

The same quasi-periodicity of $\phi$ and of $f$ that was used for $R$ shows that
$\tilde{K}^\hbar(\xi)f$ is $1$-periodic and transforms under $x\mapsto x+\tau$
with the multiplier of $V_N(\xi-\hbar)$.
Hence regularity at the four half-periods implies regularity at all their
lattice translates $-\Om_a+\Lambda$, so the image is entire and lies in $V_N(\xi-\hbar)$.

\medskip
\noindent\textit{The $K$-operator.}
Since $\hat t:V_N(\mu)\to V_N(-\mu)$ and
$K^\hbar(\xi)=-\tilde{K}^\hbar(\xi)\,\hat t$, the mapping properties of $K$
follow from those of $\tilde{K}$.
The same residue computation on the input space $V_N(-\hbar-\xi)$ again forces
the sign choice (\ref{fpKcases}).
\end{proof}

To obtain endomorphisms of a single copy of $V_N$ (respectively $V_N\otimes V_N$),
one conjugates (\ref{a150}), (\ref{a151}) and (\ref{fpKcases}) by the translations that identify the shifted
spaces in (\ref{fp19b}) with $V_N$. Let
${\hat{\bf T}}_N(\la)f(x)=f(x-\la/N)$; then ${\hat{\bf T}}_N(\la):V_N\stackrel{\sim}{\to}V_N(\la)$
and more generally ${\hat{\bf T}}_N(\la):V_N(\mu)\stackrel{\sim}{\to}V_N(\mu+\la)$.
We write ${\hat{\bf T}}_{N,i}(\la)$ for the action of ${\hat{\bf T}}_N(\la)$ on the $i$-th factor.
The modified operators are
\beq\label{fp20}
\begin{array}{rcl}
\displaystyle{
R_{ij}^{\hbar,N}(\xi_i-\xi_j)}
&\!\!=\!\!&
{\hat{\bf T}}_{N,i}(-\xi_i+\hbar)\,{\hat{\bf T}}_{N,j}(-\xi_j)\,
R^{\hbar}(\xi_i-\xi_j)\,
{\hat{\bf T}}_{N,i}(\xi_i)\,{\hat{\bf T}}_{N,j}(\xi_j-\hbar)\,,
\\[2mm]
\displaystyle{
\tilde{R}_{ij}^{\hbar,N}(\xi_i+\xi_j)}
&\!\!=\!\!&
{\hat{\bf T}}_{N,i}(-\xi_i+\hbar)\,{\hat{\bf T}}_{N,j}(-\xi_j)\,
\tilde{R}^{\hbar}(\xi_i+\xi_j)\,
{\hat{\bf T}}_{N,i}(\xi_i)\,{\hat{\bf T}}_{N,j}(\xi_j+\hbar)\,,
\\[2mm]
\displaystyle{
K_i^{\hbar,N}(\xi)}
&\!\!=\!\!&
{\hat{\bf T}}_{N,i}(-\xi+\hbar)\,
K^{\hbar}(\xi)\,
{\hat{\bf T}}_{N,i}(-\hbar-\xi)\,,
\\[2mm]
\displaystyle{
\tilde{K}_i^{\hbar,N}(\xi)}
&\!\!=\!\!&
{\hat{\bf T}}_{N,i}(-\xi+\hbar)\,
\tilde{K}^{\hbar}(\xi)\,
{\hat{\bf T}}_{N,i}(\xi+\hbar)\,.
\end{array}
\eq
\begin{proposition}
\label{propMod}
The modified operators (\ref{fp20}) preserve the corresponding theta spaces:
$R_{ij}^{\hbar,N}$ and $\tilde{R}_{ij}^{\hbar,N}$ act on $V_N\otimes V_N$, while
$K_i^{\hbar,N}$ and $\tilde{K}_i^{\hbar,N}$ act on $V_N$
(with the sign choice (\ref{a163}) for the couplings in (\ref{fpKcases})).
They satisfy the same ${\rm BC}_n$ associative Yang--Baxter equations
(\ref{a122})--(\ref{a131}), the Yang--Baxter equations
(\ref{a154}), (\ref{a155}), the reflection equations
(\ref{a156}), (\ref{a157}), and the unitarity relations
(\ref{a158})--(\ref{a161}) as the unrestricted operators.
\end{proposition}

\begin{proof}
  The analogous statement is contained in
  \cite{FP,KH2}.
  The preservation of $V_N$ and $V_N\otimes V_N$ is immediate from Lemma~\ref{propMap}
and the definitions (\ref{fp20}), once the couplings obey (\ref{a163}).
The identities for the modified operators follow by conjugation with
the translations ${\hat{\bf T}}_N$, using the multiplicativity
\beq\label{fpTmult}
\begin{array}{c}
\displaystyle{
{\hat{\bf T}}_N(\la+\mu)={\hat{\bf T}}_N(\la)\,{\hat{\bf T}}_N(\mu)
}
\end{array}
\eq
and the intertwining relations
\beq\label{fpTcomm}
\begin{array}{c}
\displaystyle{
\begin{aligned}
R^{\hbar}(\xi)\,&\bigl({\hat{\bf T}}_N(\eta)\otimes {\hat{\bf T}}_N(\eta)\bigr)
=\bigl({\hat{\bf T}}_N(\eta)\otimes {\hat{\bf T}}_N(\eta)\bigr)\,R^{\hbar}(\xi)\,,
\\[1mm]
\tilde{R}^{\hbar}(\xi)\,&\bigl({\hat{\bf T}}_N(\eta)\otimes {\hat{\bf T}}_N(-\eta)\bigr)
=\bigl({\hat{\bf T}}_N(\eta)\otimes {\hat{\bf T}}_N(-\eta)\bigr)\,\tilde{R}^{\hbar}(\xi)\,,
\end{aligned}
}
\end{array}
\eq
together with the analogous intertwining relations for $K^{\hbar}$ and $\tilde{K}^{\hbar}$.
Expanding a product of modified operators by (\ref{fp20}) and cancelling
adjacent translations reduces the claim to the unrestricted one.

As an illustration we verify the four-term relation (\ref{a131}).
Expanding the modified product on the left-hand side by (\ref{fp20}) one has
\begin{align}
&\quad R_{ij}^{w+z,N}(\xi_i-\xi_j)\,\tilde{K}_j^{w,N}(\xi_j)=
\nonumber\\
&\quad=
{\hat{\bf T}}_{N,i}(-\xi_i+w+z)\,{\hat{\bf T}}_{N,j}(-\xi_j)\,
R_{ij}^{w+z}(\xi_i-\xi_j)\,
{\hat{\bf T}}_{N,i}(\xi_i)\,{\hat{\bf T}}_{N,j}(\xi_j-w-z)\times
\nonumber\\
&\qquad\times
{\hat{\bf T}}_{N,j}(-\xi_j+w)\,
\tilde{K}^{w}_j(\xi_j)\,
{\hat{\bf T}}_{N,j}(\xi_j+w)=
\nonumber\\
&\quad=
{\hat{\bf T}}_{N,i}(-\xi_i+w+z)\,{\hat{\bf T}}_{N,j}(-\xi_j)\,
R_{ij}^{w+z}(\xi_i-\xi_j)\,
{\hat{\bf T}}_{N,i}(\xi_i)\,{\hat{\bf T}}_{N,j}(-z)\,
\tilde{K}^{w}_j(\xi_j)\,
{\hat{\bf T}}_{N,j}(\xi_j+w)=
\nonumber\\
&\quad=
{\hat{\bf T}}_{N,i}(-\xi_i+w)\,{\hat{\bf T}}_{N,j}(-\xi_j-z)\,
R_{ij}^{w+z}(\xi_i-\xi_j)\,
\tilde{K}^{w}_j(\xi_j)\,
{\hat{\bf T}}_{N,i}(\xi_i+z)\,
{\hat{\bf T}}_{N,j}(\xi_j+w)\,,
\label{fp131LHS}
\end{align}
after cancelling adjacent translations by (\ref{fpTmult}) and passing a
diagonal factor through $R_{ij}^{w+z}$ by (\ref{fpTcomm}). The three
summands on the right-hand side reduce analogously to the same outer
conjugation of the corresponding unrestricted terms. Thus each of the
four terms takes the form
\begin{equation*}
{\hat{\bf T}}_{N,i}(-\xi_i+w)\,{\hat{\bf T}}_{N,j}(-\xi_j-z)
\;(\text{corresponding unrestricted summand})\;
{\hat{\bf T}}_{N,i}(\xi_i+z)\,{\hat{\bf T}}_{N,j}(\xi_j+w)\,,
\end{equation*}
and (\ref{a131}) follows for the modified operators. The remaining
identities are obtained in the same way.
\end{proof}

\medskip
\noindent\textit{Some useful properties.}
From the local expansion $\phi(z,u)=z^{-1}+O(1)$ one has
\beq\label{fpResR}
\begin{array}{c}
\displaystyle{
\res_{\hbar=0}R^{\hbar}_{ij}(\xi)=\mathrm{Id}\,,
\qquad
\res_{\xi=0}R^{\hbar}_{ij}(\xi)=\hat s_{ij}\,,
}
\end{array}
\eq
and likewise
\beq\label{fpResRt}
\begin{array}{c}
\displaystyle{
\res_{\hbar=0}\tilde{R}^{\hbar}_{ij}(\xi)=\mathrm{Id}\,,
\qquad
\res_{\xi=0}\tilde{R}^{\hbar}_{ij}(\xi)=\hat s_{ij}\hat t_i\hat t_j\,.
}
\end{array}
\eq
The same residues hold for the modified operators (\ref{fp20}), since the
translations ${\hat{\bf T}}_N$ are entire in the spectral parameters and cancel in pairs
at $\hbar=0$ and at $\xi=0$.

The reflection $\hat t$, $(\hat t f)(x)=f(-x)$, preserves $V_N$. On the basis
(\ref{fp17}) it acts by reversing the label,
\beq\label{fp33}
\begin{array}{c}
\displaystyle{
\hat t\,\theta_\al=\theta_{-\al}\,,\qquad
\theta_\al(-x)=\theta_{-\al}(x)\,.
}
\end{array}
\eq
Using ${\hat{\bf T}}_N(-\la)\,\hat t=\hat t\,{\hat{\bf T}}_N(\la)$, the relations
$\tilde{R}^{\hbar}_{ij}(\xi)=\hat t_j R^{\hbar}_{ij}(\xi)\hat t_j$
and
$K_i^{\hbar}(\xi)=-\tilde{K}_i^{\hbar}(\xi)\,\hat t_i$
pass to the modified operators:
\beq\label{fpRtN}
\begin{array}{c}
\displaystyle{
\tilde{R}_{ij}^{\hbar,N}(\xi)=\hat t_j R_{ij}^{\hbar,N}(\xi)\hat t_j\,,
}
\end{array}
\eq
\beq\label{fpKtN}
\begin{array}{c}
\displaystyle{
K_i^{\hbar,N}(\xi)=-\tilde{K}_i^{\hbar,N}(\xi)\,\hat t_i\,.
}
\end{array}
\eq

Since $\hat s_{ij}^2=1$ one also has the Fourier-type identity
$\hat s_{ij} R^{\hbar}_{ij}(\xi)=R^{\xi}_{ij}(\hbar)$, which holds for the modified operators
(\ref{fp20}):
\beq\label{fpFour}
\begin{array}{c}
\displaystyle{
\hat s_{ij} R^{\hbar,N}_{ij}(\xi)=R^{\xi,N}_{ij}(\hbar)\,.
}
\end{array}
\eq
An analogous relation for the boundary operators depends on the parity of $N$.
Using (\ref{a9ntiphi}), i.e.\ $v(-z,-u)=-v(z,u)$, one finds
\[
\hat t_i\,\tilde{K}^{\hbar}_i(\xi|\nu,\nu')
=\tilde{K}^{\xi}_i(\hbar|\nu',\nu)\,.
\]
For even $N$ (when $\nu'=\nu$) the modified operators therefore satisfy
\beq\label{fpFourKeven}
\begin{array}{c}
\displaystyle{
\hat t_i\,\tilde{K}^{\hbar,N}_i(\xi)=\tilde{K}^{\xi,N}_i(\hbar)\,,
}
\end{array}
\eq
while for odd $N$, with $\nu'=(\nu_0,\nu_1,-\nu_2,\nu_3)$ as in (\ref{a163}),
\beq\label{fpFourKodd}
\begin{array}{c}
\displaystyle{
\hat t_i\,\tilde{K}^{\hbar,N}_i(\xi|\nu_0,\nu_1,\nu_2,\nu_3)
=\tilde{K}^{\xi,N}_i(\hbar|\nu_0,\nu_1,-\nu_2,\nu_3)\,,
}
\end{array}
\eq
Now let us write down the explicit form of the modified operators (\ref{fp20})
after the conjugations by ${\hat{\bf T}}_N$. Writing $x_{12}=x_1-x_2$, one has
\beq\label{fp21}
\begin{array}{c}
\begin{array}{l}
\displaystyle{
\big(R_{12}^{\hbar,N}(\xi)f\big)(x_1,x_2)
=\phi\!\Big(\hbar,\,x_{12}+\frac{\xi-\hbar}{N}\Big)
f\!\Big(x_1-\frac{\hbar}{N},\,x_2+\frac{\hbar}{N}\Big)
}
\\[2mm]
\displaystyle{\qquad
-\,\phi\!\Big(-\xi,\,x_{12}+\frac{\xi-\hbar}{N}\Big)
f\!\Big(x_2-\frac{\xi}{N},\,x_1+\frac{\xi}{N}\Big)\,,
}
\end{array}
\end{array}
\eq
\beq\label{fp21b}
\begin{array}{c}
\begin{array}{l}
\displaystyle{
\big(\tilde{R}_{12}^{\hbar,N}(\xi)f\big)(x_1,x_2)
=\phi\!\Big(\hbar,\,x_1+x_2+\frac{\xi-\hbar}{N}\Big)
f\!\Big(x_1-\frac{\hbar}{N},\,x_2-\frac{\hbar}{N}\Big)
}
\\[2mm]
\displaystyle{\qquad
-\,\phi\!\Big(-\xi,\,x_1+x_2+\frac{\xi-\hbar}{N}\Big)
f\!\Big(-x_2-\frac{\xi}{N},\,-x_1-\frac{\xi}{N}\Big)\,,
}
\end{array}
\end{array}
\eq
\beq\label{fp31}
\begin{array}{c}
\displaystyle{
\big(K_i^{\hbar,N}(\xi)f\big)(x)
=
v\!\left(-\xi,\,x+\frac{\xi-\hbar}{N}\Big|\nu\right)
f\!\left(x+\frac{2\xi}{N}\right)
-
v\!\left(\hbar,\,x+\frac{\xi-\hbar}{N}\Big|\nu'\right)
f\!\left(-x+\frac{2\hbar}{N}\right)\,,
}
\end{array}
\eq
\beq\label{fp31b}
\begin{array}{c}
\displaystyle{
\big(\tilde{K}_i^{\hbar,N}(\xi)f\big)(x)
=
v\!\left(\hbar,\,x+\frac{\xi-\hbar}{N}\Big|\nu'\right)
f\!\left(x-\frac{2\hbar}{N}\right)
-
v\!\left(-\xi,\,x+\frac{\xi-\hbar}{N}\Big|\nu\right)
f\!\left(-x-\frac{2\xi}{N}\right)\,.
}
\end{array}
\eq

\paragraph{$\mZ_N$ symmetry.}
As in the construction of Felder and Pasquier \cite{FP}, the modified operator
$R_{12}^{\hbar,N}(\xi)$ on $V_N\otimes V_N$ commutes with the diagonal
action of the Heisenberg generators (\ref{fp13}):
\beq\label{fpZn}
\begin{array}{c}
\displaystyle{
\bigl[R_{12}^{\hbar,N}(\xi),\,A\otimes A\bigr]=0\,,
\qquad
\bigl[R_{12}^{\hbar,N}(\xi),\,B\otimes B\bigr]=0\,.
}
\end{array}
\eq
Thus $R_{12}^{\hbar,N}(\xi)\,(\theta_\be\otimes\theta_\ga)$ is a linear combination of
$\theta_\mu\otimes\theta_\nu$ with $\mu+\nu=\be+\ga$, and the matrix elements depend
only on the differences of the indices. In the matrix realisation of the next
subsection this gives the $\mZ_N$ symmetry
of the Baxter--Belavin $R$-matrix.

\subsection{The Baxter--Belavin \texorpdfstring{$R$}{R}-matrix and the corresponding \texorpdfstring{${\ti R}$}{R}-matrix}

The goal of this subsection is to identify the modified operator $R_{12}^{\hbar,N}$ with the Baxter--Belavin $R$-matrix and to obtain the corresponding $\tilde{R}$.

\paragraph{Matrix basis.}
For the finite-dimensional matrix realisation we use the matrices
in ${\rm End}(\mC^N)$:
\beq\label{el10}\begin{array}{c} \displaystyle{
        T_\al = \exp \left( -\al_1 \al_2 \frac{\pi \imath}{N} \right) Q^{\al_1} \Lambda^{\al_2}, \quad \al = (\al_1, \al_2)\in \mZ_N \times \mZ_N\,,
}\end{array}\eq
where $Q,\Lambda\in{\rm End}( \mC^N)$, with entries\footnote{Notice that in our previous papers we used a slightly different definition.
Previously, we had $ Q_{kl}=\delta_{kl}\exp(\frac{2\pi
 \imath}{{ N}}k)$ instead of (\ref{el11}). The present convention is chosen to match the Felder--Pasquier procedure.}
\beq\label{el11}
 \begin{array}{c}
  \displaystyle{
 Q_{kl}=\delta_{kl}\exp(-\frac{2\pi
 \imath}{{ N}}k)\,,\ \ \
 \Lambda_{kl}=\delta_{k-l+1=0\,{\hbox{\tiny{mod}}}\,
 { N}}\,,\quad Q^{ N}=\Lambda^{ N}=1_N
 }
 \end{array}
 \eq
 generate a finite-dimensional representation of the Heisenberg group:
 \beq\label{el12}
 \begin{array}{c}
  \displaystyle{
 \Lambda^{\alpha_2} Q^{\alpha_1}=\exp\left(-\frac{2\pi\imath}{{ N}}\,\alpha_1
 \alpha_2\right)Q^{\alpha_1} \Lambda^{\alpha_2}\,,\ \alpha_{1},\,\alpha_2\in\mZ\,.
 }
 \end{array}
 \eq
 In particular, $T_0=T_{(0,0)}=1_N$ is the $N\times N$ identity matrix.
 For the product of the basis matrices we have:
\beq\label{Tcond}\begin{array}{c} \displaystyle{
    T_\al T_{\be} = \ka_{\al, \be} T_{\al + \be},
    \quad \ka_{\al, \be} = \exp \left( \frac{\pi i}{N}(\al_1 \be_2 - \al_2 \be_1) \right)\,,
}\end{array}\eq
where $\al+\be=(\al_1+\be_1,\al_2+\be_2)$. In particular, (\ref{Tcond}) means that $T_{-\al}=T_\al^{-1}$.
This basis has the property $\tr(T_\al T_{-\be})=N\delta_{\al,\be}$.

\begin{rem}
Identifying $V_N\simeq\mC^N$ via the ordered basis
$\{\theta_1,\dots,\theta_{N-1},\theta_0\}$, and comparing with
(\ref{el11}), one obtains
\beq\label{fp15}
\begin{array}{c}
\displaystyle{
A\;\longleftrightarrow\;Q\,,\qquad
B\;\longleftrightarrow\;\Lambda\,,
}
\end{array}
\eq
so that the Baxter--Belavin matrices $T_\al$ of (\ref{el10}) act on $V_N$ by the
corresponding Heisenberg operators. (Here $B\theta_\al=\theta_{\al-1}$ matches
$\Lambda e_j=e_{j-1}$, while $A\leftrightarrow Q$ matches
$A\theta_\al=e^{-2\pi\imath\al/N}\theta_\al$, with $Q$ as in (\ref{el11}).
The relation $AB=e^{2\pi\imath/N}BA$ is compatible with
$\Lambda Q=e^{-2\pi\imath/N}Q\Lambda$ from (\ref{el12}).)
\end{rem}

\begin{rem}
Instead, Felder and Pasquier \cite{FP} diagonalise $B$, using
\beq\label{fp11}
\begin{array}{c}
\displaystyle{
\theta^\mathrm{FP}_\al(x)=\sum_{n\in\mZ}
\exp\!\Big(
\pi\imath\frac{n^2\tau}{N}+2\pi\imath n\Big(x-\frac{\al}{N}\Big)
\Big)\,,
\qquad \al\in\mZ_N\,,
}
\end{array}
\eq
for which $A\,\theta^\mathrm{FP}_\al=\theta^\mathrm{FP}_{\al-1}$ and
$B\,\theta^\mathrm{FP}_\al=e^{2\pi\imath\al/N}\,\theta^\mathrm{FP}_\al$, and the
identification is $A\leftrightarrow\Lambda$, $B\leftrightarrow Q^{-1}$. The two bases are
related by the discrete Fourier transform on $\mZ_N$ (up to normalisation),
which interchanges the roles of $Q$ and $\Lambda$.
\end{rem}

\paragraph{Permutation operator.}
On functions of the auxiliary variables the permutation operators $\hat s_{ij}$ from
(\ref{a111}) simply exchange the arguments. In particular, on pure tensors of basis
theta functions (\ref{fp17}) one has
\beq\label{fpShat}
\begin{array}{c}
\displaystyle{
\hat s_{12}\,(\theta_\be\otimes\theta_\ga)=\theta_\ga\otimes\theta_\be\,.
}
\end{array}
\eq
In the matrix realisation $V_N\simeq\mC^N$ we associate to $\hat s_{12}$ the
permutation operator $P_{12}\in{\rm End}(\mC^N\otimes\mC^N)$ acting by
$P_{12}(a\otimes b)=(b\otimes a)$ for any $a,b\in\mC^N$.
Also, for any $A,B\in{\rm End}(\mC^N)$ one has
$P_{12}(A\otimes B)=(B\otimes A)P_{12}$. The explicit form of $P_{12}$ in the
matrix basis is
\begin{equation}\label{P12}
\begin{array}{c}
    \displaystyle{
    P_{12}=\sum\limits_{k,l=1}^N e_{kl}\otimes e_{lk}=\frac{1}{N}\sum\limits_{\al\in\,\mZ_N\times\mZ_N}T_\al\otimes T_{-\al}\,,
    }
\end{array}
\end{equation}
where $\{e_{kl}\}$ is the standard matrix basis in ${\rm End}( \mC^N)$. In particular, $P_{12}^2=1_{N^2}$.

\paragraph{The Baxter--Belavin elliptic $R$-matrix.}
Introduce the following meromorphic function of two complex variables $\xi$ and $\hbar$ taking values in ${\rm End}(\mC^N\otimes \mC^N)$:
\begin{equation}\label{BB}
\begin{array}{c}
    \displaystyle{
    R^{\hbar}_{12} (\xi) = \frac{1}{N}
  \sum_{\al \in \mZ_N \times \mZ_N} \varphi_\al (\xi, \frac{\hbar}{N} + \om_\al) T_\al \otimes T_{-\al}}\,\in{\rm End}(\mC^N\otimes \mC^N)\,,
\end{array}
\end{equation}
where
 \beq\label{w2022}
 \begin{array}{c}
  \displaystyle{
 \vf_a(\xi,\om_a+\hbar)=\exp (2\pi\imath\,\frac{a_2\xi}{N})\,\phi(\xi,\om_a+\hbar)=
 \exp (2\pi\imath \xi\p_\tau\om_a)\,\phi(\xi,\om_a+\hbar)\,,
 \quad a\in\,\mZ_{ N}\times\mZ_{ N}

 }
 \end{array}
 \eq
  and
 \beq\label{w2023}
 \begin{array}{c}
  \displaystyle{
 \om_a=\frac{a_1+a_2\tau}{N}\,,
 \quad a\in\,\mZ_{ N}\times\mZ_{ N}\,.
 }
 \end{array}
 \eq
\begin{proposition}[\cite{FP}]
\label{propFP}
In the basis (\ref{fp17}) of $V_N$ the matrix of
$R_{12}^{\hbar,N}$ coincides with the Baxter--Belavin
elliptic $R$-matrix (\ref{BB}).
For $N=2$ one recovers Baxter's $R$-matrix of the eight-vertex model.
\end{proposition}

\begin{proof}
Identify $V_N\simeq\mC^N$ via the ordered basis
$\{\theta_1,\dots,\theta_{N-1},\theta_0\}$ of (\ref{fp17}), with
$A\leftrightarrow Q$ and $B\leftrightarrow\Lambda$ as in (\ref{fp15}).
We compute the matrix of (\ref{BB}) in this basis and then that of the
functional operator (\ref{fp21}), and compare the two.

From (\ref{fp14}) and (\ref{el10}), for $\al=(\al_1,\al_2)$ and
$\be\in\mZ_N$,
\beq\label{fpTact}
\begin{array}{c}
\displaystyle{
T_\al\,\theta_\be
=\exp\!\Big(-\frac{\pi\imath}{N}\al_1(2\be-\al_2)\Big)\,
\theta_{\be-\al_2}\,.
}
\end{array}
\eq
Hence on pure tensors
\beq\label{fpTTact}
\begin{array}{c}
\begin{array}{l}
\displaystyle{
\bigl(T_\al\otimes T_{-\al}\bigr)(\theta_\be\otimes\theta_\ga)
=
\exp\!\Big(-\frac{2\pi\imath}{N}\al_1(\be-\ga-\al_2)\Big)
\theta_{\be-\al_2}\otimes\theta_{\ga+\al_2}\,.
}
\end{array}
\end{array}
\eq
Consequently the Baxter--Belavin operator (\ref{BB}) acts by
\beq\label{fpBBact}
\begin{array}{c}
\displaystyle{
R_{12}^{\hbar}(\xi)\,(\theta_\be\otimes\theta_\ga)
=\frac1N\sum_{\al\in\mZ_N\times\mZ_N}
\exp\!\Big(-\frac{2\pi\imath}{N}\al_1(\be-\ga-\al_2)\Big)
\vf_{\al}\!\Big(\xi,\,\frac{\hbar}{N}+\om_{\al}\Big)
\theta_{\be-\al_2}\otimes\theta_{\ga+\al_2}\,.
}
\end{array}
\eq
The coefficient of $\theta_{\be-\al_2}\otimes\theta_{\ga+\al_2}$ is
\beq\label{fpBBcoeff}
\begin{array}{c}
\displaystyle{
C^{\mathrm{BB}}_{\al_2}(\be,\ga;\xi,\hbar)
:=
\frac1N\sum_{\al_1\in\mZ_N}
\exp\!\Big(-\frac{2\pi\imath}{N}\al_1(\be-\ga-\al_2)\Big)
\vf_{(\al_1,\al_2)}\!\Big(\xi,\,\frac{\hbar}{N}+\om_{(\al_1,\al_2)}\Big)
}
\\[2mm]
\displaystyle{\qquad
=
\frac1N\,e^{\frac{2\pi\imath}{N}\al_2\xi}
\sum_{\al_1\in\mZ_N}
\exp\!\Big(-\frac{2\pi\imath}{N}\al_1(\be-\ga-\al_2)\Big)
\phi\!\Big(\xi,\,\frac{\hbar+\al_2\tau}{N}+\frac{\al_1}{N}\Big)\,.
}
\end{array}
\eq
Using $\phi(\xi,u)=\phi(u,\xi)$ and (\ref{a911m}) with
$m=-\be+\ga+\al_2$ one obtains
\beq\label{fpBBsimp}
\begin{array}{c}
\displaystyle{
C^{\mathrm{BB}}_{\al_2}(\be,\ga;\xi,\hbar)
=
}
\\ \ \\
\displaystyle{
=
e^{\frac{2\pi\imath}{N}\al_2\xi}\,
\exp\!\Big(\frac{2\pi\imath}{N}(\be-\ga-\al_2)(\hbar+\al_2\tau)\Big)\,
\phi\!\Big(\hbar+\al_2\tau,\,\xi+(\be-\ga-\al_2)\tau;\,N\tau\Big)\,,
}
\\ \ \\
\displaystyle{
R_{12}^{\hbar}(\xi)\,(\theta_\be\otimes\theta_\ga)
=\sum_{\al_2\in\mZ_N}
C^{\mathrm{BB}}_{\al_2}(\be,\ga;\xi,\hbar)\,
\theta_{\be-\al_2}\otimes\theta_{\ga+\al_2}\,.
}
\end{array}
\eq

Now turn to the functional side. Set $\rho:=x_{12}+(\xi-\hbar)/N$ and
substitute $f=\theta_\be\otimes\theta_\ga$ into (\ref{fp21}):
\beq\label{fpRth}
\begin{array}{c}
\begin{array}{l}
\displaystyle{
\big(R_{12}^{\hbar,N}(\xi)\,(\theta_\be\otimes\theta_\ga)\big)(x_1,x_2)
=
\phi\!\Big(\hbar,\,\rho\Big)
\theta_\be\!\Big(x_1-\frac{\hbar}{N}\Big)
\theta_\ga\!\Big(x_2+\frac{\hbar}{N}\Big)
}
\\[2mm]
\displaystyle{\qquad
-\,
\phi\!\Big(-\xi,\,\rho\Big)
\theta_\be\!\Big(x_2-\frac{\xi}{N}\Big)
\theta_\ga\!\Big(x_1+\frac{\xi}{N}\Big)\,.
}
\end{array}
\end{array}
\eq
The right-hand side lies in $V_N\otimes V_N$ by
Proposition~\ref{propMod}. By the $\mZ_N$ symmetry (\ref{fpZn}),
the operator (\ref{fp21}) commutes with $A\otimes A$ and $B\otimes B$,
so the expansion in the basis (\ref{fp17}) reduces to
\beq\label{fpRmat}
\begin{array}{c}
\displaystyle{
R_{12}^{\hbar,N}(\xi)\,(\theta_\be\otimes\theta_\ga)
=\sum_{\mu\in\mZ_N}
R^{\be\ga}_{\mu,\,\be+\ga-\mu}(\xi,\hbar)\,
\theta_\mu\otimes\theta_{\be+\ga-\mu}\,.
}
\end{array}
\eq
Equivalently, the only possibly nonzero matrix elements are those of the form
$R^{\be\ga}_{\be-\al_2,\,\ga+\al_2}$.
To extract them, apply (\ref{a911t}) with $m=0$
(after $\tau\mapsto N\tau$ and $\phi(z,u)=\phi(u,z)$) to the two Kronecker
factors in (\ref{fpRth}). This yields
\beq\label{fpRsum}
\begin{array}{c}
\begin{array}{l}
\displaystyle{
R_{12}^{\hbar,N}(\xi)\,(\theta_\be\otimes\theta_\ga)
=\sum_{\mu\in\mZ_N}
\Bigg(
e^{2\pi\imath(\be-\mu)\rho}\,
\phi\!\Big(\hbar+(\be-\mu)\tau,\,N\rho;\,N\tau\Big)\,
\theta_\be\!\Big(x_1-\frac{\hbar}{N}\Big)
\theta_\ga\!\Big(x_2+\frac{\hbar}{N}\Big)
}
\\[2mm]
\displaystyle{\qquad
-\,
e^{2\pi\imath(\ga-\mu)\rho}\,
\phi\!\Big(-\xi+(\ga-\mu)\tau,\,N\rho;\,N\tau\Big)\,
\theta_\be\!\Big(x_2-\frac{\xi}{N}\Big)
\theta_\ga\!\Big(x_1+\frac{\xi}{N}\Big)
\Bigg)\,.
}
\end{array}
\end{array}
\eq
Each summand in (\ref{fpRsum}) is an eigenvector of $A\otimes\mathrm{Id}$ with
eigenvalue $e^{-2\pi\imath\mu/N}$ (since $A\theta_\al=e^{-2\pi\imath\al/N}\theta_\al$),
and is therefore proportional to $\theta_\mu\otimes\theta_{\be+\ga-\mu}$.
To find the proportionality factor it is enough to evaluate the summand at
any convenient point where it does not vanish. Choose
\beq\label{fpPtspec}
\begin{array}{c}
\displaystyle{
N\rho=-\hbar+(\mu-\be)\tau\,,
}
\end{array}
\eq
so that the first term vanishes (because $\phi(z,-z)=0$).
On this locus
$x_2-\xi/N=x_1-(\mu-\be)\tau/N$ and
$x_1+\xi/N=x_2+(\mu-\be)\tau/N$, and using (\ref{fp10}) one has
\beq\label{fpThid}
\begin{array}{c}
\displaystyle{
\theta_\be\!\Big(x_2-\frac{\xi}{N}\Big)
\theta_\ga\!\Big(x_1+\frac{\xi}{N}\Big)
=
\exp\!\Big(
-\frac{2\pi\imath}{N}(\mu-\be)\xi
\Big)
\theta_\mu(x_1)\,
\theta_{\be+\ga-\mu}(x_2)\,.
}
\end{array}
\eq
Only the second term survives. Substituting (\ref{fpPtspec}) gives
\beq\label{fpRme}
\begin{array}{c}
\begin{array}{l}
\displaystyle{
R^{\be\ga}_{\mu,\,\be+\ga-\mu}(\xi,\hbar)
=
\exp\!\Big(
-\frac{2\pi\imath}{N}(\ga-\mu)\hbar
+\frac{2\pi\imath}{N}(\ga-\mu)(\mu-\be)\tau
-\frac{2\pi\imath}{N}(\mu-\be)\xi
\Big)
}
\\[2mm]
\displaystyle{\qquad
\times\,
\phi\!\Big(
\hbar-(\mu-\be)\tau,\,
\xi-(\ga-\mu)\tau;\,
N\tau
\Big)\,.
}
\end{array}
\end{array}
\eq
For $\mu=\be-\al_2$ this becomes
\beq\label{fpRme2}
\begin{array}{c}
\begin{array}{l}
\displaystyle{
R^{\be\ga}_{\be-\al_2,\,\ga+\al_2}(\xi,\hbar)
=
\exp\!\Big(
\frac{2\pi\imath}{N}\al_2\xi
+\frac{2\pi\imath}{N}(\be-\ga-\al_2)(\hbar+\al_2\tau)
\Big)
}
\\[2mm]
\displaystyle{\qquad
\times\,
\phi\!\Big(
\hbar+\al_2\tau,\,
\xi+(\be-\ga-\al_2)\tau;\,
N\tau
\Big)\,.
}
\end{array}
\end{array}
\eq
Comparing (\ref{fpRme2}) with (\ref{fpBBsimp}) one finds
$R^{\be\ga}_{\be-\al_2,\,\ga+\al_2}=C^{\mathrm{BB}}_{\al_2}$,
hence $R_{12}^{\hbar,N}(\xi)=R_{12}^{\hbar}(\xi)$ in the basis (\ref{fp17}).
\end{proof}

The expression (\ref{BB}) is the elliptic $\mZ_N$-symmetric Baxter--Belavin $R$-matrix.
Under the identification $A\leftrightarrow Q$, $B\leftrightarrow\Lambda$ the relations
(\ref{fpZn}) become
\begin{equation}\label{R}
\begin{array}{c}
    \displaystyle{
    (Q\otimes Q) R^{\hbar}_{12} (z) = R^{\hbar}_{12} (z) (Q\otimes Q) \,,\qquad
    (\Lambda\otimes \Lambda) R^{\hbar}_{12} (z) = R^{\hbar}_{12} (z) (\Lambda\otimes \Lambda) \,.
    }
\end{array}
\end{equation}

\paragraph{Important properties.}
The $R$-matrix (\ref{BB}) satisfies the Fourier symmetry
(which follows from (\ref{fpFour}) under the identification $\hat s_{12}\leftrightarrow P_{12}$):
\begin{equation}\label{el30}\begin{array}{c}
  P_{12} R^{\hbar}_{12}(\xi) = R^{\xi}_{12}(\hbar)\,,
\end{array}\end{equation}
 The $R$-matrix has simple poles in both $\hbar$ and $\xi$
(as follows from (\ref{fpResR}) under $\hat s_{12}\leftrightarrow P_{12}$):
 \beq\label{el32}
 \begin{array}{c}
  \displaystyle{
\res\limits_{\xi=0}R^{\hbar}_{12}(\xi)=
P_{12}\,,
  }
 \end{array}
 \eq
 \beq\label{el33}
 \begin{array}{c}
  \displaystyle{
\res\limits_{\hbar=0} R^{\hbar}_{12}(\xi)=1_N\otimes 1_N\,.
}
 \end{array}
 \eq
 The quasi-periodic behaviour on the lattice of periods is as follows:
 \beq\label{el34}
 \begin{array}{c}
  \displaystyle{
 R_{12}^\hbar(\xi+1)=(Q\otimes 1_N)R_{12}^\hbar(\xi)(Q^{-1}\otimes 1_N)\,,
  }
 \\ \ \\
  \displaystyle{
 R_{12}^\hbar(\xi+\tau)=\exp(-2\pi\imath\frac{\hbar}{N})\,(\Lambda\otimes 1_N)R_{12}^\hbar(\xi)(\Lambda^{-1}\otimes
 1_N)\,.
  }
 \end{array}
 \eq

The reflection $\hat t$ from (\ref{fp33}) acts on the basis theta functions by
$\hat t\,\theta_\al=\theta_{-\al}$.
In the matrix realisation $V_N\simeq\mC^N$ we associate to $\hat t$ the
matrix $J\in{\rm End}(\mC^N)$,
\beq\label{el21}
J=e_{1,N-1}+e_{2,N-2}+\dots+e_{N-1,1}+e_{N,N}=\frac{1}{2N}\sum\limits_{a\in\,\mZ_{2N}\times\mZ_{2N}} T_a\,,
\eq
with $J^2=1_N$. Conjugation by $J$ realises the involutive automorphism
$\mathrm{j}$ of ${\rm End}(\mC^N)$,
\begin{equation}\label{jAut}
\mathrm{j}(A)= J\, A\, J^{-1}\,,
\qquad
\mathrm{j}(T_\alpha)=T_{-\alpha}\,.
\end{equation}
Denote by $\tilde{R}_{12}^\hbar(\xi)$ the operator $(\mathrm{id}\otimes\mathrm{j}) R_{12}^\hbar(\xi)$:
  \begin{equation}\label{Rtilde}
\begin{array}{c}
    \displaystyle{
    \tilde{R}^{\hbar}_{12} (\xi):=J_2\,     R^{\hbar}_{12} (\xi)\,J_2 = \frac{1}{N}
    \sum_{\al \in \mZ_N \times \mZ_N} \varphi_\al (\xi, \frac{\hbar}{N} + \om_\al)\,
     T_\al \otimes T_{\al}}\,\in{\rm End}(\mC^N\otimes \mC^N )\,.
\end{array}
\end{equation}
This is the matrix counterpart of (\ref{fpRtN}), obtained under the
identification $\hat t\leftrightarrow J$.
Consequently, (\ref{Rtilde}) is the matrix form of the modified operator
(\ref{fp21b}) in the basis (\ref{fp17}).

Since $(\rm{id}\otimes \mathrm{j})$ commutes with $Q^{\pm 1}\otimes 1_N$ and
$\Lambda^{\pm 1}\otimes 1_N$, the operator $\tilde{R}_{12}^{\hbar}(\xi)$ has the
same quasi-periodic behaviour as (\ref{el34}). Its simple poles have residues
(as follows from (\ref{fpResRt}) under $\hat s_{12}\leftrightarrow P_{12}$ and
$\hat t\leftrightarrow J$):
 \beq\label{el35}
 \begin{array}{c}
  \displaystyle{
\res\limits_{\xi=0}\tilde{R}^{\hbar}_{12}(\xi)=
P_{12} J_1 J_2\,,
  }
 \end{array}
 \eq
 \beq\label{el36}
 \begin{array}{c}
  \displaystyle{
\res\limits_{\hbar=0} \tilde{R}^{\hbar}_{12}(\xi)=1_N\otimes 1_N\,.
}
 \end{array}
 \eq

\subsection{The elliptic boundary \texorpdfstring{$K$}{K}-matrix}\label{sec25K}
We now give the matrix formula for $\tilde{K}$ and outline its identification with the modified boundary operator.
Introduce a meromorphic function $\tilde{K}^{\hbar}(\xi)$ of two variables $\xi$ and $\hbar$, taking values in ${\rm End}(\mC^N)$:
   \begin{equation}\label{Ktilde}
    \tilde{K}^{\hbar}(\xi) = \frac{1}{2N}
    \sum_{\be \in \mZ_{2N} \times \mZ_{2N}} \exp(\frac{2 \pi \imath \be_2 \xi}{N}) v (\xi, \frac{\hbar}{N} + \frac{\om_\be}{2}) T_\be\in {\rm End}(\mC^N)\,.
\end{equation}
As in \cite{KH2}, the linearity of $v$ in the couplings (\ref{q413}) gives
\beq\label{KtildeSum}
\begin{array}{c}
\displaystyle{
\tilde{K}^{\hbar}(\xi)
=\sum_{a=0}^3\nu_a\,\tilde{K}_a^{\hbar}(\xi)\,,
\qquad
\tilde{K}_a^{\hbar}(\xi)
=\frac{1}{2N}
\sum_{\be\in\mZ_{2N}\times\mZ_{2N}}
\exp\!\Big(\frac{2\pi\imath\,\be_2 \xi}{N}\Big)
\vf_a\!\Big(2\xi,\,\frac{\hbar}{N}+\frac{\om_\be}{2}+\Om_a\Big)
T_\be\,.
}
\end{array}
\eq

\begin{proposition}
\label{propFPK}
In the basis (\ref{fp17}) of $V_N$ (with $A\leftrightarrow Q$,
$B\leftrightarrow\Lambda$), the matrix (\ref{Ktilde}) coincides with the
modified operator $\tilde{K}_i^{\hbar,N}$
(depending on four arbitrary parameters $\nu_a$), with the couplings fixed by
(\ref{a163}).
\end{proposition}

\begin{proof}
{\it Outline.}
We use the same identification $V_N\simeq\mC^N$ and the basis
(\ref{fp17}) as in Proposition~\ref{propFP}, with
$A\leftrightarrow Q$ and $B\leftrightarrow\Lambda$.
Both sides are linear in the couplings $\nu_a$, so it is enough to match
each summand of (\ref{KtildeSum}) with the corresponding term of the modified
operator (\ref{fp31b}) under the sign condition (\ref{a163}).

On the matrix side, (\ref{fpTact}) expands (\ref{KtildeSum}) in the basis
$\{\theta_\ga\}$, with coefficients given by finite Fourier sums over
$\mZ_{2N}\times\mZ_{2N}$.
On the functional side, one substitutes $f=\theta_\ga$ into (\ref{fp31b})
and applies the discrete Fourier identity (\ref{a911t}) as in the proof of
Proposition~\ref{propFP}.
The resulting coefficients are then compared by evaluating at a convenient
point (where one of the two Kronecker terms vanishes) and using the
quasi-periodicity of the theta basis together with the residue constraints
that fix (\ref{a163}).
This yields the claimed identification.
\end{proof}

\paragraph{Properties of $\tilde K$.}
Under the identification $\hat t\leftrightarrow J$, the relations
(\ref{fpKtN}), (\ref{fpFourKeven}) and (\ref{fpFourKodd}) become
\beq\label{elKfromt}
\begin{array}{c}
\displaystyle{
K^{\hbar}(\xi)=-\tilde{K}^{\hbar}(\xi)\,J
}
\end{array}
\eq
and
\beq\label{elFourKeven}
\begin{array}{c}
\displaystyle{
\tilde{K}^{\hbar}(\xi)=J\tilde{K}^{\xi}(\hbar)
\qquad\text{(even $N$)}\,,
}
\end{array}
\eq
\beq\label{elFourKodd}
\begin{array}{c}
\displaystyle{
\tilde{K}^{\hbar}(\xi|\nu_0,\nu_1,\nu_2,\nu_3)
=J\tilde{K}^{\xi}(\hbar|\nu_0,\nu_1,-\nu_2,\nu_3)
\qquad\text{(odd $N$)}\,.
}
\end{array}
\eq
The reflection equation and the unitarity relations for these matrices are recorded in the Summary below; see also (\ref{a160})--(\ref{a161}).

\subsection{The case \texorpdfstring{$N=2$}{N=2}}
We specialise the matrix formulae to $N=2$, as needed for the degenerations in Section~\ref{sec4}.
In this case $R^{\hbar}_{12} (z)$ is an element of ${\rm End}(\mC^2\otimes \mC^2)$, i.e.\ it is the $4\times 4$ matrix known as Baxter's $R$-matrix for the eight-vertex model \cite{BB}. Namely,
for $N=2$, the matrices $Q$, $\Lambda$ (\ref{el11}) take the form
 \beq\label{r722}
 \begin{array}{c}
  \displaystyle{
 Q=\mats{-1}{0}{0}{1}\,,\quad \Lambda=\mats{0}{1}{1}{0}
  }
 \end{array}
 \eq
and the basis matrices (\ref{el10}) are $T_{00}=1_2=\sigma_0$, $T_{10}=-\sigma_3$, $T_{01}=\sigma_1$ and
$T_{11}=-\sigma_2$, where $\{\sigma_a\}$, $a=0,1,2,3$ are the Pauli matrices:
 \beq\label{r723}
 \begin{array}{c}
  \displaystyle{
 \sigma_0=\sigma_4=\mats{1}{0}{0}{1}\,,\quad
 \sigma_1=\mats{0}{1}{1}{0}\,,\quad
 \sigma_2=\mats{0}{-\imath}{\imath}{0}\,,\quad
 \sigma_3=\mats{1}{0}{0}{-1}\,.
  }
 \end{array}
 \eq
 The permutation operator takes the form
 \beq\label{r724}
 \begin{array}{c}
  \displaystyle{
P_{12}=\frac12\sum\limits_{k=0}^3\sigma_k\otimes\sigma_k\,.
  }
 \end{array}
 \eq

For $N=2$, the $R$-matrix (\ref{BB}) has the form \cite{BB}
\beq\label{w810}
 \begin{array}{c}
  \displaystyle{
 R_{12}^\hbar(z)
 =\frac12\sum\limits_{k=0}^3\vf_{k}\Big(z,\Om_{k}+\frac{\hbar}{2}\Big)\sigma_{4-k}\otimes\sigma_{4-k}\,,
  }
 \end{array}
\eq
 and the reflection $K$-matrix \cite{IK,Vega} is
  \beq\label{w811}
 \begin{array}{c}
  \displaystyle{
 {\tilde K}^\hbar(z)=\sum\limits_{k=0}^3\rho_{k}
 e^{-2\pi\imath\hbar\p_\tau\Om_{k}}\vf_{k}(z-\Om_{k},\hbar+\Om_{k})\sigma_{4-k}=
}
\\ \ \\\
  \displaystyle{
 =
 \sum\limits_{k=0}^3\rho_{k}
 e^{2\pi\imath (z+\hbar+\Om_k)\p_\tau\Om_{k}}\phi(z+\Om_{k},\hbar+\Om_{k})\sigma_{4-k}\,,
  }
 \end{array}
\eq
 where $\rho_1=-\bar\nu_1$, $\rho_2=-\bar\nu_2$, $\rho_3=\bar\nu_3$,
 $\rho_4=\rho_0=\bar\nu_0$ is assumed, and $\p_\tau\Om_{k}$ is equal to either $0$ (for $\Om_0$, $\Om_1$) or
 $1/2$ (for $\Om_2$, $\Om_3$).

In this case $J=\sigma_0$, and we have $\tilde{R}^\hbar_{12}(z)=R^\hbar_{12}(z)$,
${\tilde K}^\hbar(z)={\tilde K}^z(\hbar)$ and ${K}^\hbar(z)=-{\tilde K}^\hbar(z)$.

\subsection{Summary}
We collect the matrix representations obtained above and record the identities they satisfy, including the comparison with the Komori--Hikami $K$-matrix.
We have recalled the construction of the Baxter--Belavin $R$-matrix
\beq\label{b072}
 \begin{array}{c}
  \displaystyle{
 R_{12}^\hbar(z)=\frac{1}{N}\sum\limits_{a\in\,\mZ_{ N}\times\mZ_{ N}}
  T_a\otimes T_{-a} \vf_a(z,\om_a+\frac{\hbar}{N})
 \in{\rm Mat}(N,\mC)^{\otimes 2}
 }
 \end{array}
 \eq
and obtained the corresponding matrix ${\ti R}_{12}^\hbar(z)$
\beq\label{b073}
 \begin{array}{c}
  \displaystyle{
 {\ti R}_{12}^\hbar(z)=J_2\,     R^{\hbar}_{12} (z)J_2=\frac{1}{N}\sum\limits_{a\in\,\mZ_{ N}\times\mZ_{ N}}
  T_a\otimes T_{a} \vf_a(z,\om_a+\frac{\hbar}{N})
 \in{\rm Mat}(N,\mC)^{\otimes 2}\,.
 }
 \end{array}
 \eq
and the $K$-matrix
\beq\label{b074}
  \displaystyle{
    \tilde{K}^{\hbar}(z) = \frac{1}{2N}
    \sum_{\be \in \mZ_{2N} \times \mZ_{2N}} \exp(\frac{2 \pi \imath \be_2 z}{N})
     v (z, \frac{\hbar}{N} + \frac{\om_\be}{2}) T_\be\in \Mat\,,
    }
\eq
\beq\label{b076}
\begin{array}{c}
\displaystyle{
K^{\hbar}(z)=-\tilde{K}^{\hbar}(z)\,J\,.
}
\end{array}
\eq
Here $J$ is the matrix (\ref{el21}).

Applying the Felder--Pasquier procedure to the operators of Shibukawa--Ueno type
in (\ref{a150})--(\ref{a153}) gives these matrix representations. The operators preserve
the corresponding spaces of theta functions, so their matrix representations satisfy the same relations
as the operators themselves.
\begin{theor}
The matrix representations (\ref{b072})-(\ref{b076}) satisfy the ${\rm BC}_n$ parametric AYBEs (\ref{a122})-(\ref{a131}),
the Yang--Baxter equations
\beq\label{b154}
  \displaystyle{
R^\hbar_{ij}(\xi_i-\xi_j)R^\hbar_{ik}(\xi_i-\xi_k)R^\hbar_{jk}(\xi_j-\xi_k)=
R^\hbar_{jk}(\xi_j-\xi_k)R^\hbar_{ik}(\xi_i-\xi_k)R^\hbar_{ij}(\xi_i-\xi_j)\,,
}
\eq
\beq\label{b155}
  \displaystyle{
R^\hbar_{ij}(\xi_i-\xi_j)\tilde{R}^\hbar_{ik}(\xi_i-\xi_k)\tilde{R}^\hbar_{jk}(\xi_j-\xi_k)=
\tilde{R}^\hbar_{jk}(\xi_j-\xi_k)\tilde{R}^\hbar_{ik}(\xi_i-\xi_k)R^\hbar_{ij}(\xi_i-\xi_j)\,.
}
\eq
 the reflection equations
\beq\label{b156}
\displaystyle{
R_{ij}^{\hbar}(\xi_i-\xi_j)K_i^\hbar(\xi_i)R_{ji}^{\hbar}(\xi_i+\xi_j)K_j^\hbar(\xi_j)=
K_j^\hbar(\xi_j)R_{ij}^{\hbar}(\xi_i+\xi_j)K_i^\hbar(\xi_i)R_{ji}^{\hbar}(\xi_i-\xi_j)\,.
}
\eq
\beq\label{b157}
  \displaystyle{
R_{ij}^{\hbar}(\xi_i-\xi_j)\tilde{K}_i^\hbar(\xi_i)\tilde{R}_{ij}^{\hbar}(\xi_i+\xi_j)
\tilde{K}_j^\hbar(\xi_j)=\tilde{K}_j^\hbar(\xi_j)\tilde{R}_{ij}^{\hbar}(\xi_i+\xi_j)
\tilde{K}_i^\hbar(\xi_i)R_{ij}^{\hbar}(\xi_i-\xi_j)\,,
}
\eq
and the unitarity properties
\beq\label{b158}
  \displaystyle{
  R_{ij}^{\hbar}(\xi)  R_{ji}^{\hbar}(-\xi)={\rm Id}\Big(\wp(\hbar)-\wp(\xi)\Big)\,,
  }
  \eq
  \beq\label{b159}
    \displaystyle{
  \tilde{R}_{ij}^{\hbar}(\xi)  \tilde{R}_{ji}^{-\hbar}(\xi)={\rm Id}\Big(\wp(\xi)-\wp(\hbar)\Big)\,,
  }
  \eq
\beq\label{b160}
  \displaystyle{
 K^\hbar_i(\xi) K^\hbar_i(-\xi)
 ={\rm Id}\sum_{a=0}^3
 \Big(\bar{\nu'}_a^{\,2}\wp(\hbar+\Omega_a)
 -\bar\nu_a^2\wp(\xi+\Omega_a)\Big)\,,
 }
 \eq
 \beq\label{b161}
   \displaystyle{
 \tilde{K}^\hbar_i(\xi) \tilde{K}^{-\hbar}_i(\xi)=
 {\rm Id}\sum_{a=0}^3
 \Big(\bar\nu_a^2\wp(\xi+\Omega_a)
 -\bar{\nu'}_a^{\,2}\wp(\hbar+\Omega_a)\Big)
 }
 \eq
 with the condition
(\ref{a163}) on the constants $\nu'$.
\end{theor}

Notice that for odd $N$ the coefficient function on the right-hand side of (\ref{b161}) can be expressed in terms of the function $v$; see (\ref{q419a}):
 \beq\label{b162}
   \displaystyle{
\sum_{a=0}^3
 \Big(\bar\nu_a^2\wp(z+\Omega_a)
 -\bar{\nu'}_a^{\,2}\wp(\hbar+\Omega_a)\Big)=v(z,\hbar|\mu)v(z,-\hbar|\mu)
 }
 \eq
with the set of constants $\mu$
 \beq\label{b163}
 \begin{array}{c}
   \displaystyle{
\mu_0=\frac12\Big(\nu_0+\nu_1-\nu_2+\nu_3\Big)\,,
 }
\quad
   \displaystyle{
\mu_1=\frac12\Big(\nu_0+\nu_1+\nu_2-\nu_3\Big)\,,
 }
  \\ \ \\
   \displaystyle{
\mu_2=\frac12\Big(-\nu_0+\nu_1+\nu_2+\nu_3\Big)\,,
 }
\quad
   \displaystyle{
\mu_3=\frac12\Big(\nu_0-\nu_1+\nu_2+\nu_3\Big)
 }
 \end{array}
 \eq
since for odd $N$ $\mu_0=\bar\nu'_0$, $\mu_1=\bar\nu'_1$, $\mu_2=-\bar\nu'_2$, $\mu_3=\bar\nu'_3$.

In their paper \cite{KH2} Komori and Hikami derived the following expression for the $K$-matrix:
 \beq\label{b170}
   \displaystyle{
K^{KH}(\hbar,z|\nu)=\sum\limits_{b=0}^3\nu_b \sum\limits_{i,j=1}^N K_{ij}^b(\hbar,z)e_{ij}\,,
 }
 \eq
 \beq\label{b171}
   \displaystyle{
K_{ij}^b(\hbar,z)=\frac{ \vth'(0|N\tau)\,\theta\left[\begin{array}{c}-\frac{2(i-1)}{N}+\frac12\\  \frac12\end{array}\right](2\hbar+2z|N\tau) }
{\theta\left[\begin{array}{c}-\frac{i+j-2}{N}+\frac12\\  \frac12\end{array}\right](2\hbar|N\tau)
\theta\left[\begin{array}{c}\frac{j-i}{N}+\frac12\\  \frac12\end{array}\right](2z|N\tau)  }
\,
\frac{ \theta\left[\begin{array}{c}-\frac{j-1}{N}+\frac{b_2-1}{2}\\  N\frac{b_1-1}{2}\end{array}\right](\hbar-z|N\tau) }
{ \theta\left[\begin{array}{c}-\frac{i-1}{N}+\frac{b_2-1}{2}\\  N\frac{b_1-1}{2}\end{array}\right](\hbar+z|N\tau) }\,,
 }
 \eq
where $b=0$ corresponds to $(b_1,b_2)=(0,0)$, $b=1$ to $(b_1,b_2)=(1,0)$, $b=2$ to $(b_1,b_2)=(1,1)$ and
 $b=3$ to $(b_1,b_2)=(0,1)$.

\begin{theor}
  The $K$-matrix (\ref{b076}) and the Komori--Hikami $K$-matrix (\ref{b170})-(\ref{b171})
  are related by
 \beq\label{b173}
   \displaystyle{
\Big({\ti K}^{\hbar}(z|\nu)\Big)^T=-J\,\Lambda K^{KH}(\hbar,z|S\nu)\Lambda^{-1}\,,
 }
 \eq
where $S$ permutes the couplings according to $S(\nu_0,\nu_1,\nu_2,\nu_3)=(\nu_2,\nu_3,\nu_0,\nu_1)$.
\end{theor}
This relation follows by comparing the quasi-periodicity properties and the residues at the half-periods
in the variable $z$.

\noindent {\bf Remark.} As noted at the beginning of this section, we define the
basis $T_a$ using the matrices $Q,\Lambda$ in (\ref{el10})-(\ref{el11}). This convention differs from the one used in our previous work.
If we use our standard notation
\beq\label{a971}
\begin{array}{c}
\displaystyle{
        T_\al = \exp \left( \al_1 \al_2 \frac{\pi \imath}{N} \right) Q^{\al_1} \Lambda^{\al_2},
        \quad \al = (\al_1, \al_2)\in \mZ_N \times \mZ_N\,,
}\end{array}\eq
with
\beq\label{a041}
 \begin{array}{c}
  \displaystyle{
 Q_{kl}=\delta_{kl}\exp(\frac{2\pi
 \imath}{{ N}}k)\,,\ \ \
 \Lambda_{kl}=\delta_{k-l+1=0\,{\hbox{\tiny{mod}}}\,
 { N}}\,,\quad Q^{ N}=\Lambda^{ N}=1_N\,,
 }
 \end{array}
 \eq
then some relations change. In particular, instead of (\ref{elFourKeven}), (\ref{elFourKodd}) and (\ref{b173})
we have the following set of relations for the $\ti K$-matrix (defined as in (\ref{Ktilde}) but with the basis (\ref{a971}), (\ref{a041})):
\beq\label{elFourKeven2}
\begin{array}{c}
\displaystyle{
\tilde{K}^{\hbar}(\xi)=\tilde{K}^{\xi}(\hbar)J
\qquad\text{(even $N$)}\,,
}
\end{array}
\eq
\beq\label{elFourKodd2}
\begin{array}{c}
\displaystyle{
\tilde{K}^{\hbar}(\xi|\nu_0,\nu_1,\nu_2,\nu_3)
=\tilde{K}^{\xi}(\hbar|\nu_0,\nu_1,-\nu_2,\nu_3)J
\qquad\text{(odd $N$)}
}
\end{array}
\eq
and
 \beq\label{b1732}
   \displaystyle{
{\ti K}^{\hbar}(z|\nu)=-\Lambda K^{KH}(\hbar,z|S\nu)\Lambda^{-1}J
 }
 \eq
with $S(\nu_0,\nu_1,\nu_2,\nu_3)=(\nu_2,\nu_3,\nu_0,\nu_1)$.
The matrices $A=\sum_a A_aT_a$ obtained by using the same coefficients in the two bases (\ref{el10})--(\ref{el11}) and (\ref{a971})--(\ref{a041})
are related by $A\to JA^TJ$.

%%%%%%%%%%%%%%%%%%%%%%%%%%%%%%%%%%%%%%%%%%%%%%%%%%%%%%%%%%%%%%%%%%%%%%%%%%%%%%%%%%%%%%%%%%%%%%%%%
%%%%%%%%%%%%%%%%%%%%%%%%%%%%%%%%%%%%%%%%%%%%%%%%%%%%%%%%%%%%%%%%%%%%%%%%%%%%%%%%%%%%%%%%%%%%%%%%%
\section{Trigonometric and rational degenerations}\label{sec4}
\setcounter{equation}{0}
In this section we derive trigonometric and rational degenerations of the elliptic $K$-matrix in the
${\rm GL}_2$ case ($N=2$). In Sections~\ref{sec311} and~\ref{sec321}, we take the standard trigonometric and rational limits. In Sections~\ref{sec312} and~\ref{sec322}, we obtain nonstandard degenerations by applying the singular gauge transformations described in \cite{AHZ,Smirnov}.

\subsection{Trigonometric degenerations}
\subsubsection{The XXZ six-vertex model}\label{sec311}
The standard trigonometric XXZ $R$-matrix has the form \cite{Jimbo}:
\beq\label{trigR}
 \begin{array}{c}
 \displaystyle{
  R_{12}^\hbar(z) = \frac12\Big(\pi \cot\pi \hbar + \pi \cot \pi z\Big)
  \Big(\sigma_0 \otimes \sigma_0 + \sigma_3 \otimes \sigma_3\Big)+
 }
 \\ \ \\
 \displaystyle{
  + \frac{\pi}{2\sin \pi \hbar}\Big(\sigma_0 \otimes \sigma_0 - \sigma_3 \otimes \sigma_3\Big)
  + \frac{\pi}{2\sin \pi z}\Big(\sigma_1 \otimes \sigma_1 + \sigma_2 \otimes \sigma_2\Big)\,.
 }
 \end{array}
\eq
This is the $R$-matrix for the quantum affine algebra ${\hat{\mathcal U}}_q({\rm gl}_2)$.
It can be obtained in the  limit ${\rm Im}(\tau)\rightarrow\infty$ from Baxter's eight-vertex $R$-matrix,
which is (\ref{b072}) for $N=2$:
 \beq\label{w8100}
 \begin{array}{c}
  \displaystyle{
 R_{12}^\hbar(z)
 =\frac12\sum\limits_{k=0}^3\vf_{k}\Big(z,\Om_{k}+\frac{\hbar}{2}\Big)\sigma_{4-k}\otimes\sigma_{4-k}\,,
  }
 \end{array}
\eq
 or explicitly
 \beq\label{w207}
 \begin{array}{c}
 R_{12}^\hbar(z)=\displaystyle{\frac12}\left(
 \begin{array}{cccc}
 \vf_{00}+\vf_{10} & 0 & 0 & \vf_{01}-\vf_{11}
 \\
 0 & \vf_{00}-\vf_{10} & \vf_{01}+\vf_{11} & 0
  \\
 0 & \vf_{01}+\vf_{11} & \vf_{00}-\vf_{10} & 0
 \\
 \vf_{01}-\vf_{11} & 0 & 0 & \vf_{00}+\vf_{10}
 \end{array}
 \right)
  \end{array}
 \eq
  with the definitions
 \beq\label{w208}
 \begin{array}{c}
  \displaystyle{
 \vf_{00}=\phi\Big(z,\frac{\hbar}{2}\Big)=\vf_{0}\Big(z,\Om_{0}+\frac{\hbar}{2}\Big)\,,\quad
 \vf_{10}=\phi\Big(z,\frac{1+\hbar}{2}\Big)=\vf_{1}\Big(z,\Om_{1}+\frac{\hbar}{2}\Big)\,,
 }
 \\ \ \\
  \displaystyle{
 \vf_{01}=e^{\pi\imath z}\phi\Big(z,\frac{\tau+\hbar}{2}\Big)=\vf_{3}\Big(z,\Om_{3}+\frac{\hbar}{2}\Big)\,,\quad
 \vf_{11}=e^{\pi\imath z}\phi\Big(z,\frac{1+\tau+\hbar}{2}\Big)=\vf_{2}\Big(z,\Om_{2}+\frac{\hbar}{2}\Big)\,,
 }
  \end{array}
 \eq
see (\ref{w209})-(\ref{w210}).
Introduce\footnote{This $q$ should not be confused with the deformation parameter in ${\hat{\mathcal U}}_q({\rm gl}_2)$.}
$q=e^{\pi i \tau}$.
The most general elliptic $K$-matrix in the ${\rm GL}_2$ ($N=2$) case was found in \cite{IK,Vega}.
We write it as\footnote{The form (\ref{w8110}) is obtained from (\ref{b074}) for $N=2$.}
\beq\label{w8110}
 \begin{array}{c}
  \displaystyle{
 {\tilde K}^\hbar(z)=
 \sum\limits_{k=0}^3\nu_{k}
 e^{2\pi\imath (z+\hbar)\p_\tau\Om_{k}}\phi(z+\Om_{k},\hbar+\Om_{k})\sigma_{4-k}\,.
  }
 \end{array}
\eq
This version differs from (\ref{w811}) by a rescaling of the coupling constants:
$\nu_a=\rho_a\exp(2\pi\imath\Om_a\partial_\tau\Om_a)$.
%: $\nu_1\rightarrow{q^{-\frac12}\nu_1}$ and $\nu_2\rightarrow{q^{-\frac12}\nu_2}$.
Taking the limit ${\rm Im}(\tau)\rightarrow\infty$, that is, $q\to 0$,
we obtain the trigonometric $K$-matrix\footnote{The constants $\nu_a$ in (\ref{trigK}) are not exactly the same as in (\ref{w8110}) due to
rescaling in the limit ${\rm Im}(\tau)\rightarrow\infty$.}:
\beq\label{trigK}
 \begin{array}{c}
 \displaystyle{
 		{\tilde K}^\hbar(z)
 		=
 }
 \\ \ \\
  \displaystyle{
 =
 \nu_0(\pi\cot(\pi z)+\pi\cot(\pi\hbar))\sigma_{0}+\nu_3(\pi\tan(\pi z)+\pi\tan(\pi\hbar))\sigma_{3}+\pi\sin\pi(\hbar+z)(\nu_1\sigma_{1}+\nu_2\sigma_{2})\,.
 }
  \end{array}
 \eq
 Up to multiplication by a scalar function, this is the trigonometric $K$-matrix
 obtained by T. Inami and H. Konno in \cite{IK}:
\beq\label{trigK2}
 \begin{array}{c}
 \displaystyle{
 		{\tilde K}^\hbar(z)=\frac{\pi\nu_0\sin(u+\eta)}{\sin(\eta)\sin(u)\cos(u)}\, K^{IK}(u,\eta)\,,\quad
 \eta=\pi\hbar\,,\ u=\pi z\,.
 }
  \end{array}
 \eq

 \subsubsection{The trigonometric seven-vertex model}\label{sec312}
 It is easy to verify that the QYBE (\ref{a021}) and the AYBE (\ref{a02}) are invariant under
 the gauge transformations
%:
 \beq\label{w8111}
  \displaystyle{
 R_{12}^{'}=AR_{12}A^{-1}\,, \qquad  A=G_{1}G_{2}\,, \qquad  G_{1}=G\otimes1\,,\ G_{2}=1\otimes G\,,
 }
 \eq
 where $G$ is an arbitrary invertible $2\times 2$ matrix. The reflection equation (\ref{a156})
 and the ${\rm BC}_n$ AYBEs (\ref{a122})-(\ref{a131}) are also invariant
 under simultaneous conjugation of $R$- and $K$-matrices:
 \beq\label{w8112}
R_{12}^{'}=AR_{12}A^{-1}, \qquad K^{'}=GKG^{-1}\,.
 \eq
 For our purposes, we use the following gauge transformation matrix:
 \beq\label{Gt}
 G_t=\mats{q^{\frac{1}{4}}\alpha^{-\frac12}}{0}{0}{q^{-\frac{1}{4}}\alpha^{\frac12}}\,.
 \eq
 Conjugation by this matrix acts nontrivially on the components $\sigma_1$ and $\sigma_2$ only:
 \beq
 \sigma_1 \rightarrow{\mats{0}{q^{\frac12}\alpha^{-1}}{q^{-\frac12}\alpha}{0}}, \qquad \sigma_2 \rightarrow{\mats{0}{-\imath q^{\frac12}\alpha^{-1}}{\imath q^{-\frac12}\alpha}{0}}\,.
 \eq
 Conjugating Baxter's $R$-matrix (\ref{w810}) by this matrix and taking the trigonometric limit $q\rightarrow{0}$ gives
 \beq\label{trigR7}
  \begin{array}{c}
  	\displaystyle{
 R_{12}^\hbar(z) = \frac12 \Big(\pi \cot\pi \hbar + \pi \cot \pi z\Big)
 \Big(\sigma_0 \otimes \sigma_0 + \sigma_3 \otimes \sigma_3\Big)
 + \frac{\pi}{2\sin \pi \hbar}\Big(\sigma_0 \otimes \sigma_0 - \sigma_3 \otimes \sigma_3\Big)+
 }
 \\ \ \\
 \displaystyle{
 + \frac{\pi}{2\sin \pi z}\Big(\sigma_1 \otimes \sigma_1 +
 \sigma_2 \otimes \sigma_2\Big)+ 4\alpha^2\pi\sin\pi(z+\hbar)\sigma_{-} \otimes \sigma_{-}\,,
 }
\end{array}
\eq
where $ \sigma_{-} = \frac12(\sigma_{1}-i\sigma_{2})$. This is the 7-vertex trigonometric $R$-matrix found by I. Cherednik \cite{Chered}.
Its explicit $4\times 4$ form is
 \beq\label{trigR71}
  \begin{array}{c}
  	\displaystyle{
 R_{12}^\hbar(z) = \pi
 \begin{pmatrix}
 \cot\pi z+ \cot\pi \hbar & 0 & 0 & 0 \\
 0 & \displaystyle{\frac{1}{\sin\pi \hbar}} & \displaystyle{\frac{1}{\sin\pi z}} & 0 \\
 0 & \displaystyle{\frac{1}{\sin\pi z}} & \displaystyle{\frac{1}{\sin\pi \hbar}} & 0 \\
 4\alpha^2\sin\pi\left(\hbar+z\right) & 0 & 0 &  \cot\pi z+ \cot\pi \hbar
 \end{pmatrix}\,.
 }
\end{array}
\eq
Applying the gauge transformation (\ref{Gt}) to the elliptic $K$-matrix (\ref{w811}) and taking the trigonometric limit yields
 \beq\label{trigK7}
 \begin{array}{c}
 	\displaystyle{
 		{\tilde K}^\hbar(z)
 		=\nu_0\Big(\pi\cot(\pi z)+\pi\cot(\pi\hbar)\Big)\sigma_{0}+\nu_3\Big(\pi\tan(\pi z)+\pi\tan(\pi\hbar)\Big)\sigma_{3}+\alpha\nu\pi\sin\pi(\hbar+z)\sigma_{-}},
 \end{array}
 \eq
 where a new coupling constant $ \nu=\nu_1+\imath\nu_2$ appears. This $K$-matrix (up to parameter redefinitions and conjugation)
 was proposed in \cite{Cong}.

 The $R$-matrix (\ref{trigR71}) and the trigonometric $K$-matrix (\ref{trigK7}) satisfy the QYBE (\ref{a021}), the AYBE (\ref{a02}),
  the ${\rm BC}_n$ AYBEs (\ref{a122})-(\ref{a131}) and the reflection equation (\ref{a156})  since the gauge-transformed elliptic $R$- and $K$-matrices satisfy these equations for any $q$. However, (\ref{trigR7}) and (\ref{trigK7}) are not gauge-equivalent to (\ref{trigR}) and (\ref{trigK}) because $G_{t}$ is singular at $q=0$.

\subsection{Rational degenerations}
Here we describe rational $R$- and $K$-matrices satisfying the QYBE (\ref{a021}), the AYBE (\ref{a02}),
  the ${\rm BC}_n$ AYBEs (\ref{a122})-(\ref{a131}) and the reflection equation (\ref{a156}).
\subsubsection{The XXX six-vertex model}\label{sec321}
The rational limit ($\cot{x}\rightarrow{\frac{1}{x}}\,,\ \sin{x}\rightarrow{x}$) applied to the XXZ $R$-matrix (\ref{trigR})
gives Yang's rational XXX $R$-matrix \cite{Yang}:
\beq\label{ratR}
\displaystyle{
		R_{12}^\hbar(z)
		=\frac{\sigma_{0}\otimes\sigma_{0}}{\hbar}+\frac{P_{12}}{z}\,,\qquad
 P_{12} = \frac12\sum\limits_{k=0}^3\sigma_{k}\otimes\sigma_{k}\,,
	}
\eq
where $P_{12}$ is the permutation operator.
The corresponding $K$-matrix is derived from (\ref{trigK}) by a similar limit ($\cot{x}\rightarrow{\frac{1}{x}}, \sin{x}\rightarrow{x}, \tan{x}\rightarrow{x}$):
\beq\label{ratK}
\displaystyle{
		{\tilde K}^\hbar(z)
		=\nu_0\Big(\frac{1}{\hbar}+\frac{1}{z}\Big)\sigma_{0}+(\hbar+z)\sum_{k=1}^{3}\nu_k \sigma_{k}}\,.
\eq
If $\nu_0\neq 0$, it can be normalised to $1$ by rescaling the $K$-matrix. Using the gauge symmetry (\ref{w8112}), one can also set $\nu_1=\nu_2=0$, leaving only $\nu_3$ (when $\nu_1^2+\nu_2^2+\nu_3^2\neq 0$).
This example was found and studied in the original paper by E. Sklyanin \cite{Skl-refl}.

\subsubsection{The rational eleven-vertex model}\label{sec322}
A nonstandard rational degeneration can be obtained by a procedure similar to that in Section~\ref{sec312}. Consider the following gauge transformation matrix, which is singular at $x=0$:
 \beq\label{Gr}
 G_r=\mats{x\beta^{-\frac{1}{2}}\alpha^{\frac12}}{0}{2x^{-1}\beta^{\frac{1}{2}}\alpha^{\frac12}}{x^{-1}\beta^{\frac{1}{2}}\alpha^{-\frac12}}\,.
 \eq
 Conjugation by this matrix acts on the Pauli matrices as follows:
 \beq
 \begin{array}{c}
 \displaystyle{
 \sigma_1 \rightarrow{\mats{-2\alpha}{x^{2}\alpha\beta^{-1}}{x^{-2}\beta(\alpha^{-1}-4\alpha)}{2\alpha}}, \qquad \sigma_2 \rightarrow{\mats{2\imath\alpha}{-\imath x^{2}\alpha\beta^{-1}}{\imath x^{-2}\beta(\alpha^{-1}+4\alpha)}{-2\imath\alpha}},}
 \\ \ \\
 \sigma_3\rightarrow{\mats{1}{0}{4x^{-2}\beta}{-1}}\,.
 \end{array}
 \eq
 Replacing $\pi$ in (\ref{trigR7}) and (\ref{trigK7}) by a formal variable $x$, applying the gauge transformation (\ref{Gr}), and taking the limit $x\rightarrow{0}$, we obtain
 \beq\label{w8115}
\begin{array}{c}
	\displaystyle{
		R_{12}^\hbar(z)=
 }
 \\ \ \\
  \displaystyle{
		=\frac{\sigma_{0}\otimes\sigma_{0}}{\hbar}+\frac{P_{12}}{z}
-\beta^2\Big(z^3+\hbar^3 + 2z^2\hbar + 2z\hbar^2\Big)\sigma_{-}\otimes\sigma_{-}
	}-\beta\Big(z+\hbar\Big)\Big(\sigma_{-}\otimes\sigma_{3}+\sigma_{3}\otimes\sigma_{-}\Big)
\end{array}
\eq
and
 \beq\label{w8116}
\begin{array}{c}
	\displaystyle{
		{\tilde K}^\hbar(z)
		=\nu_0\Big(\frac{1}{\hbar}+\frac{1}{z}\Big)\sigma_{0}
+\beta\tilde{\nu}\Big(\hbar+z\Big)\sigma_{-}}, \quad\tilde{\nu}=\nu+4\nu_3\,.
\end{array}
\eq
Explicitly, we have:
 \beq\label{w8117}
R_{12}^{\hbar}(z) =
\begin{pmatrix}
1/\hbar + 1/z & 0 & 0 & 0 \\
-\beta (z+ \hbar) & 1/\hbar & 1/z & 0 \\
-\beta (z+ \hbar) & 1/z & 1/\hbar & 0 \\
-\beta^2(z^3 + \hbar^3 + 2z^2\hbar + 2z\hbar^2) & \beta (z+ \hbar) & \beta (z+ \hbar) & 1/\hbar + 1/z
\end{pmatrix}
\eq
and
 \beq\label{w8118}
{\tilde K}^\hbar(z)=\mat{\nu_0(\frac{1}{\hbar}+\frac{1}{z})}{0}{\beta\tilde{\nu}(\hbar+z)}{\nu_0(\frac{1}{\hbar}+\frac{1}{z})}\,.
\eq
This limiting procedure was carried out in \cite{AHZ} and \cite{Smirnov}.
The resulting rational $R$-matrix (\ref{w8117}) is an eleven-vertex matrix. It was proposed in \cite{Chered}.
The AYBE (\ref{a02}) for this $R$-matrix and its higher-rank generalisations was proved in \cite{AtZ5}.

The $K$-matrix (\ref{w8116}) itself is known (see, e.g., \cite{Vega}) and can be written as
 \beq\label{w8119}
\begin{array}{c}
	\displaystyle{
		{\tilde K}^\hbar(z)
		=\frac{z+\hbar}{z\hbar}\Big(\nu_0\sigma_0+\be\hbar{\ti \nu}\, z\,\sigma_-\Big)\,.
}
\end{array}
\eq
It is known as a solution of the reflection equation for Yang's six-vertex $R$-matrix (\ref{ratR}).
To the best of our knowledge, its compatibility with the eleven-vertex $R$-matrix has not been reported previously. In addition to the reflection equation, the pair (\ref{w8117})-(\ref{w8118}) satisfies the ${\rm BC}_n$ AYBEs.

%%%%%%%%%%%%%%%%%%%%%%%%%%%%%%%%%%%%%%%%%%%%%%%%%%%%%%%%%%%%%%%%%%%%%%%%%%%%%%%%%%%%%%%%%%%%%%%%%
%%%%%%%%%%%%%%%%%%%%%%%%%%%%%%%%%%%%%%%%%%%%%%%%%%%%%%%%%%%%%%%%%%%%%%%%%%%%%%%%%%%%%%%%%%%%%%%%%
\section{Calogero--Moser model in an external field}\label{sec5}
\setcounter{equation}{0}
Calogero--Moser models are associated with Lie algebras and root systems \cite{OP}.
In this section we consider a class of such models in an external field, with Hamiltonians of the form
\beq\label{w8221}
	\displaystyle{
H= \frac{1}{2} \sum_{k=1}^{n} p_k^2-\sum\limits_{i>j}^{n}V(x_i-x_j)+\sum_{k=1}^{n}W(x_k)\,,
}
\eq
These are $A$-type models with an additional external potential acting on each particle \cite{Inoz83,Woj}.
Our aim is to generalise the results of those papers by
constructing $R$-matrix-valued Lax pairs for these models.

We consider three combinations of the potentials in (\ref{w8221}).
\beq\label{w8222}
  \begin{array}{c}
\hbox{Types of $V(x)$ and $W(x)$:}
 \end{array}
 \eq
\begin{itemize}
%   \item
  \item I.\quad $V(x)=g^2x^{-2}$, $W(x)=\nu^2x^2$;
  \item II.\quad $V(x)=g^2x^{-2}$, $W(x)=\nu^2x^2+\gamma^2x^4$, $\gamma \neq 0$;
  \item III.\footnote{In case III, the dependence on $\delta$ can be removed from the corresponding Hamiltonian (\ref{HS}); we retain this parameter here.}
  \quad $ \displaystyle{V(x)=g^2\frac{1}{\sinh^2(\frac{\be x}{2})}}$, $\displaystyle{W(x)=\gamma\rho\cosh(\beta x+\delta)+\frac{\gamma^2}{4}\cosh(2\beta x+2\delta)+{\rm const}}$.
\end{itemize}
%\gamma\rho\cosh(\beta q_k+\delta)+\frac{\gamma^2}{4}\cosh(2\beta q_k+2\delta)

The reason for treating case $I$ separately will become clear below.
The Lax pairs for these systems were proposed in \cite{Inoz83} and their integrability was
 proved in \cite{Woj}\footnote{The integrability of system $II$ was also independently discussed in \cite{Poly3}.}.
  We propose an $R$-matrix generalisation of the Lax pairs in \cite{Inoz83}, with the following block-matrix structure:
 \beq\label{q502}
  \begin{array}{c}
  \displaystyle{
\mL(z)=\mat{\mL^{11}(z)}{\mL^{12}(z)}{\mL^{21}(z)}{\mL^{22}(z)}\,,\qquad
\mM(z)=\mat{\mM^{11}(z)}{\mM^{12}(z)}{\mM^{21}(z)}{\mM^{22}(z)}+\mH\, 1_{2n\times 2n}\,,
 }
 \end{array}
 \eq
 This structure is chosen to reproduce the original Lax pairs of \cite{Inoz83} for $N=1$,
 when the $R$- and $K$-matrices become scalar functions. The similarity between the Lax pairs for models $I$-$III$ and
 the elliptic Inozemtsev model (\ref{a170}) reflects the fact that these models arise as degenerations of (\ref{a170}), as shown in \cite{Ta2}, \cite{Osh}. Also, the relationship between these models was examined in the context of multicomponent Painlev\'e equations \cite{Ta2}, see also \cite{BCR}.

 Verification of the Lax equations is based on the AYBE and the unitarity conditions for $R$- and $K$-matrices.
 We also use the following relation between them:
 \beq\label{q500}
   \displaystyle{
-2{F}_{ij}^{\hbar}\tilde{K}_j^{\hbar}+2\tilde{K}_i^\hbar
 {F}_{ij}^{-\hbar}+{R}_{ij}^{\hbar}\tilde Y_j^{\hbar}+\tilde Y_i^\hbar{R}_{ij}^{-\hbar}=0\,,
 }
 \eq
 where $\tilde{Y}^\hbar_i(z):=\partial_z \tilde{K}^\hbar_i(z)$, ${F}^\hbar_{ij}(z):=\partial_z {R}^\hbar_{ij}(z)$.
 Notice that (\ref{q500}) does not follow from the ${\rm BC}_n$ four-term relation (\ref{a131}). Differentiating
 (\ref{a131}) also produces terms with ${\ti R}$ and its derivative, but these terms are absent in (\ref{q500}).
 That is, (\ref{q500}) holds only for particular choices of $R$- and $K$-matrices.

\subsection{Rational cases}
The generalisation of the Lax pair for case $I$
can be constructed using Yang's $R$-matrix (\ref{ratR}) and the following $K$-matrix:
\beq\label{q501}
 \begin{array}{c}
 	\displaystyle{
 		{\tilde K}^\hbar(z)
 		=(\frac{1}{\hbar}+\nu_0{z})\sigma_{0}}\,.
 \end{array}
 \eq
 The $R$-matrix (\ref{ratR}), $R_{12}^\hbar(z)
		=\frac{\sigma_{0}\otimes\sigma_{0}}{\hbar}+\frac{P_{12}}{z}$,
 and the $K$-matrix (\ref{q501}) satisfy the reflection equation and (\ref{q500}), but (\ref{a131})
  holds up to a constant matrix independent of $z$ and $\hbar$.
  The construction of the Lax pair in this case does not require (\ref{a131}) or its consequences.

 The entries of (\ref{q502}) are
\beq\label{q503}
 \begin{array}{c}
 	\displaystyle{
 		\mL^{11}_{ij}(z)=\delta_{ij}p_i+g(1-\delta_{ij})R_{ij}^z(q_{ij})\,,
 	}
 \end{array}
 \eq
 \beq\label{q504}
 \begin{array}{c}
 	\displaystyle{
 		\mL^{12}_{ij}(z)=\delta_{ij}\tilde{K}^z_i(q_i)\,,
 	}
 \end{array}
 \eq
 \beq\label{q505}
 \begin{array}{c}
 	\displaystyle{
 		\mL^{21}_{ij}(z)=-\delta_{ij}\tilde{K}^{-z}_i(q_i)\,,
 	}
 \end{array}
 \eq
 \beq\label{q506}
 \begin{array}{c}
 	\displaystyle{
 		\mL^{22}_{ij}(z)=-\delta_{ij}p_i-g(1-\delta_{ij})R^{-z}_{ij}(q_{ij})
 	}
 \end{array}
 \eq
 and
 \beq\label{q507}
 \begin{array}{c}
 	\displaystyle{
 		\mM^{11}_{ij}(z)=\delta_{ij}A_i+g(1-\delta_{ij})F^z_{ij}(q_{ij})\,,
 	}
 \end{array}
 \eq
 \beq\label{q508}
 \begin{array}{c}
 	\displaystyle{
 		\mM^{12}_{ij}(z)=\frac12\,\delta_{ij}\tilde{Y}_i^z(q_i)\,,
 	}
 \end{array}
 \eq
 \beq\label{q509}
 \begin{array}{c}
 	\displaystyle{
 		\mM^{21}_{ij}(z)=\frac12\,\delta_{ij}\tilde{Y}^{-z}_i(q_i)\,,
 	}
 \end{array}
 \eq
 \beq\label{q510}
 \begin{array}{c}
 	\displaystyle{
 		\mM^{22}_{ij}(z)=\delta_{ij}A_i+g(1-\delta_{ij})F^{-z}_{ij}(q_{ij})\,,
 	}
 \end{array}
 \eq
 \beq\label{q511}
 \begin{array}{c}
 	\displaystyle{
 		A_i=-g\sum_{k:k\neq i}^n F^0_{ik}(q_{ik})=g\sum_{k:k\neq i}^n\frac{P_{ik}}{(q_i-q_k)^2}\,,
 	}
 \end{array}
 \eq
 \beq\label{q512}
 \begin{array}{c}
 	\displaystyle{
 		\mH=g\sum_{k<l}^n F^0_{kl}(q_{kl})=-g\sum_{k< l}^n\frac{P_{kl}}{(q_k-q_l)^2}\,.
 	}
 \end{array}
 \eq

 Case $II$ can be described by $R$- and $K$-matrices without an additional parameter:
 \beq\label{q513}
 \begin{array}{c}
 	\displaystyle{
 		R_{12}(z)
 		=\frac{P_{12}}{z}\,,
 	}
 \end{array}
 \eq
 \beq\label{q514}
 \begin{array}{c}
 	\displaystyle{
 		{\tilde K}(z)
 		=(i\rho+\nu z+i\gamma z^2)\sigma_{0}\,.
        }
 \end{array}
 \eq
For $R$- and $K$-matrices without a Planck constant, the identity (\ref{q500}) transforms into
\beq\label{q5001}
-2{F}_{ij}\tilde{K}_j+2\tilde{K}_i
 {F}_{ij}+{R}_{ij}\tilde Y_j+\tilde Y_i{R}_{ij}=0\,,
 \eq
The expressions (\ref{q513}) and (\ref{q514}) satisfy all necessary conditions, and the Lax pair has the form:
 \beq\label{q5031}
 \begin{array}{c}
 	\displaystyle{
 		\mL^{11}_{ij}=\delta_{ij}p_i+g(1-\delta_{ij})R_{ij}(q_{ij})\,,
 	}
 \end{array}
 \eq
 \beq\label{q5041}
 \begin{array}{c}
 	\displaystyle{
 		\mL^{12}_{ij}=\delta_{ij}\tilde{K}_i(q_i)\,,
 	}
 \end{array}
 \eq
 \beq\label{q5051}
 \begin{array}{c}
 	\displaystyle{
 		\mL^{21}_{ij}=\delta_{ij}\tilde{K}_i(-q_i)\,,
 	}
 \end{array}
 \eq
 \beq\label{q5061}
 \begin{array}{c}
 	\displaystyle{
 		\mL^{22}_{ij}=-\delta_{ij}p_i-g(1-\delta_{ij})R_{ij}(q_{ij})\,,
 	}
 \end{array}
 \eq
 and
 \beq\label{q5071}
 \begin{array}{c}
 	\displaystyle{
 		\mM^{11}_{ij}=\delta_{ij}A_i+g(1-\delta_{ij})F_{ij}(q_{ij})\,,
 	}
 \end{array}
 \eq
 \beq\label{q5081}
 \begin{array}{c}
 	\displaystyle{
 		\mM^{12}_{ij}=\frac12\,\delta_{ij}\tilde{Y}_i(q_i)\,,
 	}
 \end{array}
 \eq
 \beq\label{q5091}
 \begin{array}{c}
 	\displaystyle{
 		\mM^{21}_{ij}=\frac12\,\delta_{ij}\tilde{Y}_i(-q_i)\,,
 	}
 \end{array}
 \eq
 \beq\label{q5101}
 \begin{array}{c}
 	\displaystyle{
 		\mM^{22}_{ij}=\delta_{ij}A_i+g(1-\delta_{ij})F_{ij}(q_{ij})\,,
 	}
 \end{array}
 \eq
 \beq\label{q51122}
 \begin{array}{c}
 	\displaystyle{
 		A_i=-g\sum_{k:k\neq i}^n F^0_{ik}(q_{ik})=g\sum_{k:k\neq i}^n\frac{P_{ik}}{(q_i-q_k)^2}\,,
 	}
 \end{array}
 \eq
 \beq\label{q51222}
 \begin{array}{c}
 	\displaystyle{
 		\mH=g\sum_{k<l}^n F^0_{kl}(q_{kl})=-g\sum_{k<l}^n\frac{P_{kl}}{(q_k-q_l)^2}\,.
 	}
 \end{array}
 \eq

\noindent{\bf Remark.} For the $K$-matrix (\ref{q514}), identity (\ref{q5001}) is a special case of the global conformal Ward identities in a two-dimensional conformal field theory (CFT) \cite{DiF}:
 \beq\label{ward}
\sum_i
 \big(\epsilon(w_i)\p_{w_i}+\p_{w_i}\epsilon(w_i)h_i\big)\langle \phi_1(w_1)\cdots\phi_n(w_n) \rangle=0\,, \qquad \epsilon(z)= c+b z+az^2\,,
 \eq
 or equivalently
\begin{equation}
\begin{aligned}
\sum_i \partial_{w_i} \langle \phi_1(w_1)\cdots\phi_n(w_n) \rangle &= 0\,, \\
\sum_i (w_i \partial_{w_i} + h_i) \langle \phi_1(w_1)\cdots\phi_n(w_n) \rangle &= 0\,, \\
\sum_i (w_i^2 \partial_{w_i} + 2w_i h_i) \langle \phi_1(w_1)\cdots\phi_n(w_n) \rangle &= 0\,,
\end{aligned}
\end{equation}
where $\langle \phi_1(w_1)\cdots\phi_n(w_n) \rangle$ is an $n$-point correlation function of primary operators $\phi_i(w_i)$ with conformal dimensions $h_i$. In particular, in the free-fermion model, the two-point correlation function is
\beq
\langle \psi(w_1, \bar{w}_1)\psi(w_2, \bar{w}_2) \rangle \propto \frac{1}{w_1-w_2}\,,
\eq
and the conformal dimension of the field $\psi$ is equal to $h=\frac{1}{2}$.

At a purely algebraic level, for the $R$- and $K$-matrices (\ref{q513}) and (\ref{q514}), relation (\ref{q5001}) reduces to (\ref{ward}) for the free-fermion two-point correlation function. The correspondence is therefore
\[
\begin{aligned}
R_{12}(z_1-z_2) &\longleftrightarrow \langle \psi(w_1, \bar{w}_1)\psi(w_2, \bar{w}_2) \rangle \\
\ti{K}(z) &\longleftrightarrow \epsilon(z) \\
-2{F}_{ij}\tilde{K}_j+2\tilde{K}_i
 {F}_{ij}+{R}_{ij}\tilde Y_j+\tilde Y_i{R}_{ij}=0 & \longleftrightarrow \sum^2_{i=1}
 \big(\epsilon(w_i)\p_{w_i}+\frac12\p_{w_i}\epsilon(w_i)\big)\langle \psi(w_1)\psi(w_2) \rangle=0\,.
\end{aligned}
\]
The relationship between Calogero-type models and two-dimensional CFT has been studied extensively; see, for example, \cite{EPSS} and the references therein. These connections concern quantum models, whereas the correspondence described here concerns a classical model. It would be interesting to clarify the analogous relationship between quantum Calogero models in an external field and CFT. Quantum models in an external field have received less attention; quantised models of types $II$ and $III$ were considered in \cite{Sas2}.

\subsection{Hyperbolic case}
The Lax pair for case $III$ has the same structure as in the previous case, (\ref{q5031})--(\ref{q51222}):
 \beq\label{q5032}
 \begin{array}{c}
 	\displaystyle{
 		\mL^{11}_{ij}=\delta_{ij}p_i+g(1-\delta_{ij})R_{ij}(q_{ij})\,,
 	}
 \end{array}
 \eq
 \beq\label{q5042}
 \begin{array}{c}
 	\displaystyle{
 		\mL^{12}_{ij}=\delta_{ij}\tilde{K}_i(q_i)\,,
 	}
 \end{array}
 \eq
 \beq\label{q5052}
 \begin{array}{c}
 	\displaystyle{
 		\mL^{21}_{ij}=\delta_{ij}\tilde{K}_i(q_i)\,,
 	}
 \end{array}
 \eq
 \beq\label{q5062}
 \begin{array}{c}
 	\displaystyle{
 		\mL^{22}_{ij}=-\delta_{ij}p_i-g(1-\delta_{ij})R_{ij}(q_{ij})
 	}
 \end{array}
 \eq
 and
 \beq\label{q5072}
 \begin{array}{c}
 	\displaystyle{
 		\mM^{11}_{ij}=\delta_{ij}A_i+g(1-\delta_{ij})F_{ij}(q_{ij})\,,
 	}
 \end{array}
 \eq
 \beq\label{q5082}
 \begin{array}{c}
 	\displaystyle{
 		\mM^{12}_{ij}=\frac12\,\delta_{ij}\tilde{Y}_i(q_i)\,,
 	}
 \end{array}
 \eq
 \beq\label{q5092}
 \begin{array}{c}
 	\displaystyle{
 		\mM^{21}_{ij}=-\frac12\,\delta_{ij}\tilde{Y}_i(q_i)\,,
 	}
 \end{array}
 \eq
 \beq\label{q5102}
 \begin{array}{c}
 	\displaystyle{
 		\mM^{22}_{ij}=\delta_{ij}A_i+g(1-\delta_{ij})F_{ij}(q_{ij})\,,
 	}
 \end{array}
 \eq

 \beq\label{q5112}
 \begin{array}{c}
 	\displaystyle{
 		A_i=\frac{g\be}{2}\sum_{k:k\neq i}^n \sinh^{-2}\big(\frac{\be q_{ik}}{2}\big) P_{ik}\,,
 	}
 \end{array}
 \eq
 \beq\label{q5122}
 \begin{array}{c}
 	\displaystyle{
 		\mH=-\frac{\be g}{2}\sum_{k<l}^n \sinh^{-2}\big(\frac{\be q_{kl}}{2}\big)P_{kl}\,.
 	}
 \end{array}
 \eq
 The $R$- and $K$-matrices are
 \beq\label{q5123}
 \begin{array}{c}
 	\displaystyle{
		R_{12}(z)
		=\frac{P_{12}}{\sinh(\frac{\beta z}{2})}\,,
	}
\end{array}
 \eq
\beq\label{q5124}
 \begin{array}{c}
	\displaystyle{
		{\tilde K}(z)
		=(\gamma\cosh(\beta z +\delta)+\rho) \sigma_{0}\,.
      }
\end{array}
\eq
Unlike the previous cases, this Lax pair contains the entries (\ref{q5112}) and (\ref{q5122}), which are not expressed in terms of derivatives of the $R$-matrix. In fact, the $R$-matrix (\ref{q5123}) does not satisfy the AYBE. The Lax equation nevertheless holds with the entries (\ref{q5112}) and (\ref{q5122}) defined above.

%%%%%%%%%%%%%%%%%%%%%%%%%%%%%%%%%%%%%%%%%%%%%%%%%%%%%%%%%%%%%%%%%%%%%%%%%%%%%%%%%%%%%%%%%%%%%%%%%
%%%%%%%%%%%%%%%%%%%%%%%%%%%%%%%%%%%%%%%%%%%%%%%%%%%%%%%%%%%%%%%%%%%%%%%%%%%%%%%%%%%%%%%%%%%%%%%%%
\section{A list of long-range spin chains}\label{sec6}
\setcounter{equation}{0}

In this section, we present explicit forms for the quantum
Hamiltonians (\ref{a07}) for the long-range spin chains entering the quantum Lax equation (\ref{a11493}).
The Hamiltonians are obtained by freezing the $R$-matrix-valued Lax pairs
for the Calogero--Inozemtsev system and its degenerations, including the
models in an external field considered in Section~\ref{sec5}.

The Hamiltonians of these spin chains have the general form
\beq\label{Hc}
{\mH}=g\sum_{k<l}^n \Big(F^0_{kl}(z_{k}-z_l)+{\ti F}^0_{kl}(z_{k}+z_l)\Big)
+\frac12\sum\limits_{k=1}^n \tilde{Y}_k^0(z_k)\,.
\eq
We consider here the models arising from ${\rm GL}_2$ $R$- and $K$-matrices (that is, ${\mH}\in{\rm Mat}(2,\mC)^{\otimes n}$),
 and use the results of the previous sections.

As noted above, the existence of a quantum Lax equation (\ref{a11493}) provides evidence for quantum integrability but does not by itself establish it.

\subsection{Elliptic case}
We consider the eight-vertex $R$-matrix (\ref{w810}) and the $K$-matrix (\ref{w811}). As noted above,
 ${\ti R}_{12}^\hbar(z)={R}_{12}^\hbar(z)$ in this case. That is
 \beq\label{w421}
 \begin{array}{c}
  \displaystyle{
 R_{12}^\hbar(z)
 =\frac12\sum\limits_{k=0}^3\vf_{k}\Big(z,\Om_{k}+\frac{\hbar}{2}\Big)\sigma_{4-k}\otimes\sigma_{4-k}
  }
 \end{array}
\eq
 and\footnote{Here we denote the coupling constants by $\nu_k$, rather than $\rho_k$ as in (\ref{w811}).
 An exact match with the general expression for the $K$-matrix (\ref{Ktilde}) is not needed here.}
  \beq\label{w422}
 \begin{array}{c}
  \displaystyle{
 {\tilde K}^\hbar(z)=\sum\limits_{k=0}^3\nu_{k}
 e^{-2\pi\imath\hbar\p_\tau\Om_{k}}\vf_{k}(z-\Om_{k},\hbar+\Om_{k})\sigma_{4-k}\,.
}
 \end{array}
\eq
Direct computation yields (see (\ref{a90627}))
  \beq\label{w423}
  \begin{array}{c}
  \displaystyle{
F^0_{12}(x)={\ti F}^0_{12}(x)=\p_x R_{12}^\hbar(x)\Big|_{\hbar=0}=
  }
  \\[4mm]
  \displaystyle{
=-\frac12\, E_2(x)\sigma_0\otimes\sigma_0+
\frac12\sum\limits_{a=1}^3\vf_a(x,\Om_a)\Big(E_1(x+\Om_a)-E_1(x)-E_1(\Om_a)\Big)
\sigma_{4-a}\otimes\sigma_{4-a}=
}
  \\[4mm]
  \displaystyle{
=-\frac12\, E_2(x)\sigma_0\otimes\sigma_0-
\frac12\sum\limits_{a=1}^3\frac{\vf_1(x,\Om_1)\vf_2(x,\Om_2)\vf_3(x,\Om_3)}{\vf_a(x,\Om_a)}
\sigma_{4-a}\otimes\sigma_{4-a}
 }
 \end{array}
 \eq
 and
  \beq\label{w424}
  \begin{array}{c}
  \displaystyle{
\tilde{Y}^0(x)=\p_x {\ti K}^\hbar(x)\Big|_{\hbar=0}=
  }
  \\ \ \\
  \displaystyle{
=-\nu_0 E_2(x)\sigma_0-\sum\limits_{a=1}^3\nu_a
\vf_a(x-\Om_a,\Om_a)\Big(E_1(x+\Om_a)-E_1(x)-E_1(\Om_a)\Big)\sigma_{4-a}\,.
 }
 \end{array}
 \eq
 Substituting these expressions into (\ref{Hc}), we obtain the Hamiltonian for a quantum elliptic long-range anisotropic spin chain associated with the
 ${\rm BC}_n$ root system:
  \beq\label{w425}
  \begin{array}{c}
   \displaystyle{
{\mH}=-\frac{g}2\sum_{k<l}^n
 \Big(E_2({z}_k-{z}_l)+E_2({z}_k+{z}_l)\Big)\stackrel{k}{\sigma}_0\otimes\stackrel{l}{\sigma}_0+
 }
 \\[5mm]
   \displaystyle{
+
\frac{g}2\sum_{k<l}^n\sum\limits_{a=1}^3\Big[\vf_a({z}_k-{z}_l,\Om_a)\Big(E_1({z}_k-{z}_l+\Om_a)-E_1({z}_k-{z}_l)-E_1(\Om_a)\Big)+
}
\\[5mm]
\displaystyle{
+\vf_a({z}_k+{z}_l,\Om_a)\Big(E_1({z}_k+{z}_l+\Om_a)-E_1({z}_k+{z}_l)-E_1(\Om_a)\Big)\Big]
\stackrel{k}\sigma_{4-a}\otimes\stackrel{l}\sigma_{4-a}-
 }
 \\[4mm]
   \displaystyle{
 -\frac12\sum\limits_{k=1}^n\Big[ \nu_0 E_2({z}_k)\stackrel{k}\sigma_0+\sum\limits_{a=1}^3\nu_a
\vf_a({z}_k-\Om_a,\Om_a)\Big(E_1({z}_k+\Om_a)-E_1({z}_k)-E_1(\Om_a)\Big)\stackrel{k}\sigma_{4-a}\Big]\,,
}
 \end{array}
 \eq
 where $z_k$ are equilibrium positions in the classical Calogero--Inozemtsev model (\ref{a071}). Therefore, $z_k$
 are solutions of the equations
\begin{equation}
 \label{w426}
   \displaystyle{
g^2
\sum_{k:k\neq i}^{n} \Big(\wp'({z}_i-{z}_k) + \wp'({z}_i+{z}_k)\Big) +
\frac{1}{2}\sum\limits_{a=0}^3\nu_a^2\wp'({z}_i+\Om_a)=0\,,\quad i=1,...,n\,.
}
 \end{equation}
 The existence of real solutions to these equations was proved in \cite{Tak} under certain conditions on the coupling constants.

\subsection{Trigonometric cases}
\subsubsection{6-vertex}
The Hamiltonian of the particle system is
\beq\label{Ht6}
\begin{array}{c}
\
\displaystyle{
 H= \frac{1}{2} \sum_{k=1}^{n} p_k^2 - g^2\pi^2
\sum_{i<j}^{n} \Big(\frac{1}{\sin^2(\pi q_{ij})} + \frac{1}{\sin^2(\pi q^+_{ij})}\Big) -
\frac{\pi^2}{2}\sum_{k=1}^{n}\Big(\frac{\nu_0^2}{\sin^2(\pi q_k)}+\frac{\nu_3^2}{\cos^2(\pi q_k)}\Big)+
}
 \\ \ \\
\displaystyle{
+\frac{\pi^2}{4}(\nu^2_1+\nu^2_2)\sum_{k=1}^{n}\cos{(2\pi q_k)}\,.
}
\end{array}
\eq
The equilibrium positions satisfy
\beq\label{Ht6eq}
\begin{array}{c}
\displaystyle{
  2g^2
\sum_{j\neq i}^{n} \Big(\frac{\cos(\pi z_{ij})}{\sin^3(\pi z_{ij})} + \frac{\cos(\pi z^+_{ij})}{\sin^3(\pi z^+_{ij})}\Big) +
\nu_0^2\frac{\cos(\pi z_i)}{\sin^3(\pi z_i)}-\nu_3^2\frac{\sin(\pi z_i)}{\cos^3(\pi z_i)}-
\frac{1}{2}(\nu^2_1+\nu^2_2)\sin{(2\pi z_i)}=0\,.
}
\end{array}
\eq
Under the change of variables $x_i=\cos{(2\pi z_i)}$, equations (\ref{Ht6eq}) take the rational form
\beq\label{Ht6eqr}
\begin{array}{c}
\displaystyle{
  8g^2
\sum_{j\neq i}^{n} \Big(\frac{-x^2_j-x_ix_j+2}{(x_j-x_i)^3}\Big) +
\frac{2\nu_0^2}{(1-x_i)^2}-\frac{2\nu_3^2}{(1+x_i)^2}-
\frac{1}{2}(\nu^2_1+\nu^2_2)=0\,.
}
\end{array}
\eq
%
%q_i+q_j<\frac12, (\forall i,j \text{ s.t. } i<j)\
In \cite{Tak} it was shown that the elliptic Inozemtsev system  (\ref{a170})
admits a real equilibrium configuration. An analogous statement holds
for the trigonometric degeneration of the Inozemtsev system (\ref{Ht6}).
\begin{proposition}\label{p71}
   For $g$, $\nu_0$, $\nu_3$ $\in$ $\mathbb{R}\backslash\{0\}$,  $\nu_1$, $\nu_2 $ $\in \mathbb{R}$
   equations (\ref{Ht6eq}) have a real solution in
   the domain $D=\{(q_1,\dots,q_n)\in\mathbb{R}^n|\, 0<q_1<q_2\dots<q_n<\frac12 \}$.
\end{proposition}
\begin{proof}
Consider the potential of the system (\ref{Ht6})
\beq\label{Ht61}
\begin{array}{c}
\displaystyle{
\Phi(q_1,\dots,q_n)=    - g^2\pi^2
\sum_{i<j}^{n} \Big(\frac{1}{\sin^2(\pi q_{ij})} + \frac{1}{\sin^2(\pi q^+_{ij})}\Big) -
\frac{\pi^2}{2}\sum_{k=1}^{n}\Big(\frac{\nu_0^2}{\sin^2(\pi q_k)}+\frac{\nu_3^2}{\cos^2(\pi q_k)}\Big)+
}
 \\ \ \\
\displaystyle{
+\frac{\pi^2}{4}(\nu^2_1+\nu^2_2)\sum_{k=1}^{n}\cos{(2\pi q_k)}\,.
}
\end{array}
\eq
$\Phi(q_1,\dots,q_n)$ is real-valued and smooth on the
 domain $D$. As $(q_1,\dots,q_n)$ tends to the boundary of
 $D$, $\Phi(q_1,\dots,q_n)$ tends to $-\infty$ if $\nu_0\neq0$, $\nu_3\neq0$.
  The potential is also bounded above. Indeed, for real arguments
\beq
\Phi(q_1,\dots,q_n)\leq -n(n-1)g^2\pi^2-\frac{n\pi^2}{2}(\nu^2_0+\nu^2_3)+\frac{n\pi^2}{4}(\nu^2_1+\nu^2_2)
\eq
Therefore, $\Phi(q_1,\dots,q_n)$ attains its maximum at a point $(z_1,\dots,z_n)\in D$.
Equations (\ref{Ht6eq}) then follow from the necessary condition
$\p_{q_i}\Phi(q_1,\dots,q_n)|_{(q_1,\dots,q_n)=(z_1,\dots,z_n)}=0$ for an interior extremum.
\end{proof}

The Hamiltonian of the spin chain can be found from the general formula (\ref{Hc}):
\beq\label{Ht62}
\begin{array}{c}
	\displaystyle{
		\mH=-\frac{\pi^{2}g}{2} \sum_{i<j} \Bigg[\left( \frac{\cos(\pi z_{ij}) }{\sin^{2}(\pi z_{ij})}+
\frac{\cos(\pi z^{+}_{ij})}{\sin^{2}(\pi z^{+}_{ij})}\right)\stackrel{i}{\sigma_{1}}\otimes\stackrel{j}{\sigma_{1}} +
\left( \frac{\cos(\pi z_{ij}) }{\sin^{2}(\pi z_{ij})}+
\frac{\cos(\pi z^{+}_{ij})}{\sin^{2}(\pi z^{+}_{ij})}\right)\stackrel{i}{\sigma_{2}}\otimes \stackrel{j}{\sigma_{2}}  +
        }
        \\ \ \\
        \displaystyle{
        +\left( \frac{1}{\sin^{2}(\pi z_{ij})}+\frac{1}{\sin^{2}(\pi z^{+}_{ij})}\right)\stackrel{i}{\sigma_{3}}
        \otimes \stackrel{j}{\sigma_{3}}\Bigg]-
        \frac{\pi^{2} g}{2} \sum_{i<j} \Big(\frac{1}{\sin^{2}(\pi z_{ij})}+
        \frac{1}{\sin^{2}(\pi z^+_{ij})}\Big) \sigma_{0}^{\otimes n}-
    }
    \\ \ \\
    \displaystyle{
    	-\frac{\pi^2}{2}{\sigma}^{\otimes n}_0\sum_{k=1}^n\frac{\nu_0}{\sin^2(\pi z_k)}+
    \frac{\pi^2}{2}\sum_{k=1}^n\Big(\frac{\nu_3}{\cos^2(\pi z_k)}\stackrel{k}{\sigma}_3+\cos(\pi z_k)(\nu_1\stackrel{k}{\sigma}_1+\nu_2\stackrel{k}{\sigma}_2)\Big)\,.
    }
\end{array}
\eq

\subsubsection{7-vertex}
\beq\label{Ht7}
\begin{array}{c}
\
\displaystyle{
 H= \frac{1}{2} \sum_{k=1}^{n} p_k^2 - g^2\pi^2
\sum_{i<j}^{n} \Big(\frac{1}{\sin^2(\pi q_{ij})} + \frac{1}{\sin^2(\pi q^+_{ij})}\Big) -
\frac{\pi^2}{2}\sum_{k=1}^{n}\Big(\frac{\nu_0^2}{\sin^2(\pi q_k)}+\frac{\nu_3^2}{\cos^2(\pi q_k)}\Big)\,.
}
\end{array}
\eq
The equilibrium positions satisfy
\beq\label{Ht7eq}
\begin{array}{c}
\
\displaystyle{
  2g^2
\sum_{j\neq i}^{n} \Big(\frac{\cos(\pi z_{ij})}{\sin^3(\pi z_{ij})} + \frac{\cos(\pi z^+_{ij})}{\sin^3(\pi z^+_{ij})}\Big) +
\nu_0^2\frac{\cos(\pi z_i)}{\sin^3(\pi z_i)}-\nu_3^2\frac{\sin(\pi z_i)}{\cos^3(\pi z_i)}=0\,.
}
\end{array}
\eq
The solutions of these equations are related to the zeros of the Jacobi polynomials \cite{Sas1}.

The Hamiltonian of the spin chain is as follows:
\beq\label{Ht71}
\begin{array}{c}
	\displaystyle{
		\mH=-\frac{\pi^{2}g}{2} \sum_{i<j}\Bigg[ \left( \frac{\cos(\pi z_{ij}) }{\sin^{2}(\pi z_{ij})}+
\frac{\cos(\pi z^{+}_{ij})}{\sin^{2}(\pi z^{+}_{ij})}\right)\stackrel{i}{\sigma_{1}}\otimes\stackrel{j}{\sigma_{1}}
 +\left( \frac{\cos(\pi z_{ij}) }{\sin^{2}(\pi z_{ij})}+
 \frac{\cos(\pi z^{+}_{ij})}{\sin^{2}(\pi z^{+}_{ij})}\right)\stackrel{i}{\sigma_{2}}\otimes \stackrel{j}{\sigma_{2}}  +
        }
        \\ \ \\
        \displaystyle{
        +\left( \frac{1}{\sin^{2}(\pi z_{ij})}+\frac{1}{\sin^{2}(\pi z^{+}_{ij})}\right)\stackrel{i}{\sigma_{3}}\otimes \stackrel{j}{\sigma_{3}}-8\alpha^2\left(\cos(\pi z_{ij})+
        \cos(\pi z^+_{ij})\right)\stackrel{i}{\sigma_{-}}\otimes \stackrel{j}{\sigma_{-}}\Bigg]-
    }
    \\ \ \\
    \displaystyle{
    	-\frac{\pi^{2} g}{2} \sum_{i<j} \Big(\frac{1}{\sin^{2}(\pi z_{ij})}+
    \frac{1}{\sin^{2}(\pi z^+_{ij})}\Big) \sigma_{0}^{\otimes n}-\frac{\pi^2}{2}{\sigma}^{\otimes n}_0\sum_{k=1}^n\frac{\nu_0}{\sin^2(\pi z_k)}+
     }
    \\ \ \\
    \displaystyle{
    +\frac{\pi^2}{2}\sum_{k=1}^n\Big(\frac{\nu_3}{\cos^2(\pi z_k)}\stackrel{k}{\sigma}_3+
    \alpha\nu\cos(\pi z_k)\stackrel{k}{\sigma}_{-}\Big)\,.
    }
\end{array}
\eq

\subsection{Rational case}
\subsubsection{6-vertex}
The Hamiltonian of the particle system is
\beq\label{Hr6}
\begin{array}{c}
\
\displaystyle{
 H= \frac{1}{2} \sum_{k=1}^{n} p_k^2 - g^2
\sum_{i<j}^{n} \Big(\frac{1}{( q_{ij})^2} + \frac{1}{(q^+_{ij})^2}\Big) -
\frac{1}{2}\sum_{k=1}^{n}\Big(\frac{\nu_0^2}{q^2_k}+(\nu_3^2+\nu^2_1+\nu^2_2)q^2_k\Big)\,.
}
\end{array}
\eq
The equilibrium positions satisfy
\beq \label{eqr}
\begin{array}{c}
\
\displaystyle{
  2g^2
\sum_{j\neq i}^{n} \Big(\frac{1}{(z_{ij})^3} + \frac{1}{(z^+_{ij})^3}\Big) +
\frac{\nu_0^2}{z^3_i}-(\nu_3^2+\nu^2_1+\nu^2_2)z_i=0\,.
}
\end{array}
\eq
The solutions of these equations are related to the zeros of the Laguerre polynomials \cite{Sas1}.
 Namely,
\beq\label{Lag}
x_i=\frac{\sqrt{\nu_3^2+\nu^2_1+\nu^2_2}}{g}z^2_i \quad \text{are zeroes of}
 \quad L^{\frac{\nu_0}{g}-1}_n(x),\quad z_i=\sqrt{\frac{gx_i}{\sqrt{\nu_3^2+\nu^2_1+\nu^2_2}}}\,.
\eq

The Hamiltonian of the spin chain is
\beq\label{Hr6sc1}
\begin{array}{c}
	\displaystyle{
		\mH=-\frac{g}{2} \sum_{i<j} \left(\frac{1}{(z_{ij})^2}+\frac{1}{(z^{+}_{ij})^2}\right)\Big(\stackrel{i}{\sigma_{1}}\otimes\stackrel{j}{\sigma_{1}} +\stackrel{i}{\sigma_{2}}\otimes \stackrel{j}{\sigma_{2}}
        +\stackrel{i}{\sigma_{3}}\otimes \stackrel{j}{\sigma_{3}}+\sigma_{0}^{\otimes n}\Big)-
        }

    \\[6mm]
    \displaystyle{
    	-\frac{1}{2}{\sigma}^{\otimes n}_0\sum_{k=1}^n\frac{\nu_0}{z^2_k}+\frac{1}{2}
    \sum_{k=1}^n\Big(\nu_3\stackrel{k}{\sigma}_3+\nu_1\stackrel{k}{\sigma}_1+\nu_2\stackrel{k}{\sigma}_2\Big)\,.
    }
\end{array}
\eq
This formalism also describes a spin chain similar to (\ref{Hr6sc1}):
\beq\label{Hr6sc2}
\begin{array}{c}
	\displaystyle{
		\mH=-\frac{g}{2} \sum_{i<j} \left(\frac{1}{(z_{ij})^2}+\frac{1}{(z^{+}_{ij})^2}\right)\Big(\stackrel{i}{\sigma_{1}}\otimes\stackrel{j}{\sigma_{1}} +\stackrel{i}{\sigma_{2}}\otimes \stackrel{j}{\sigma_{2}}
        +\stackrel{i}{\sigma_{3}}\otimes \stackrel{j}{\sigma_{3}}+\sigma_{0}^{\otimes n}\Big)+
        }
        \\[6mm]
    \displaystyle{
        +\frac{1}{2}{\sigma}^{\otimes n}_0\sum_{k=1}^n\Big(\nu-\frac{\nu_0}{z^2_k}\Big)\,,
        }
\end{array}
\eq
where $z_i$ also satisfy (\ref{eqr}). The Lax pair has the same form as in the previous cases, but uses a different $K$-matrix:
\beq\label{q601}
 \begin{array}{c}
 	\displaystyle{
 		{\tilde K}^\hbar(z)
 		=\Big(\frac{1}{\hbar}+\nu{z}+\frac{\nu_0}{z}\Big)\sigma_{0}}, \quad \nu=\sqrt{\nu_3^2+\nu^2_2+\nu^2_1}\,.
 \end{array}
 \eq
 The matrix (\ref{q601}) satisfies the unitarity condition and the reflection equation, while the $BC_{n}$ AYBE holds up to a matrix independent of $z$ and $\hbar$. The Lax equation nevertheless holds, since its verification requires only derivatives of (\ref{a131}). The integrability of the system defined by (\ref{Hr6sc2}) was proved in \cite{Ya}.

\subsubsection{11-vertex}
The Hamiltonian of the particle system is
\beq\label{Hr611}
\begin{array}{c}
\
\displaystyle{
 H= \frac{1}{2} \sum_{k=1}^{n} p_k^2 - g^2
\sum_{i<j}^{n} \Big(\frac{1}{( q_{ij})^2} + \frac{1}{(q^+_{ij})^2}\Big) -
\frac{1}{2}\sum_{k=1}^{n}\frac{\nu_0^2}{q^2_k}\,.
}
\end{array}
\eq
The equilibrium positions satisfy
\beq\label{eqr11}
\begin{array}{c}
\
\displaystyle{
  2g^2
\sum_{j\neq i}^{n} \Big(\frac{1}{(z_{ij})^3} + \frac{1}{(z^+_{ij})^3}\Big) +
\frac{\nu_0^2}{z^3_i}=0\,.
}
\end{array}
\eq
These equations are a special case of (\ref{eqr}) corresponding to $\nu_3^2+\nu^2_1+\nu^2_2=0$.
In this limit, the solutions (\ref{Lag}) are ill-defined.
Thus, equations (\ref{eqr11}) have no finite real solutions (the potential in (\ref{Hr611}) is negative for nonsingular real values
of the variables).
 Nevertheless, for any solution of (\ref{eqr11})  we can write down the spin chain Hamiltonian:
\beq\label{eqr12}
\begin{array}{c}
	\displaystyle{
		\mH=-\frac{g}{2} \sum_{i<j} \left(\frac{1}{(z_{ij})^2}+\frac{1}{(z^{+}_{ij})^2}\right)\Big(\stackrel{i}{\sigma_{1}}\otimes\stackrel{j}{\sigma_{1}} +\stackrel{i}{\sigma_{2}}\otimes \stackrel{j}{\sigma_{2}}
        +\stackrel{i}{\sigma_{3}}\otimes \stackrel{j}{\sigma_{3}}+\sigma_{0}^{\otimes n}\Big)+
        }
   \\ \ \\
  \displaystyle{
  	-g\sum_{i<j}\Big(\beta^2\big(3z^2_{ij} + 3(z_{ij}^+)^2 \big)\stackrel{i}{\sigma}_{-}\otimes\stackrel{j}{\sigma}_{-}
	+2\beta\big(\stackrel{i}{\sigma}_{-}\otimes\stackrel{j}{\sigma}_{3}+\stackrel{i}{\sigma}_{3}\otimes\stackrel{j}{\sigma}_{-}\big)\Big)-
    }
    \\[3mm]
    \displaystyle{
    -\frac{1}{2}\,{\sigma}^{\otimes n}_0\sum_{k=1}^n\frac{\nu_0}{z^2_k}\,+\,\frac{\beta\tilde{\nu}}{2}\sum_{k=1}^n\stackrel{k}{\sigma}_{-}\,.
    }
\end{array}
\eq
To demonstrate the existence of complex solutions of (\ref{eqr11}), consider the following example. For $n=2$ and $\nu_0^2=g^2$ there is a solution $z_1=a$, $z_2=\imath a$.
This is because
\beq\label{eqr121}
\begin{array}{c}
	\displaystyle{(a-ia)^{-3}+(a+ia)^{-3}=-\frac1{2a^3}\,.
    }
\end{array}
\eq

\subsection{Spin chains associated with Calogero models in an external field}
For these rational systems, the spin-chain Hamiltonian has the general form
\beq\label{q6121}
 \begin{array}{c}
 	\displaystyle{
 		\mH^{0'}=g\sum_{k<l}^n F^0_{kl}(z_{kl})\,.
 	}
 \end{array}
 \eq
 In the hyperbolic case, as noted above, the Hamiltonian (\ref{q5122}) is not given by differentiation of the $R$-matrix (\ref{q5123}).

\paragraph{The system $I$:}
\beq\label{q6122}
\begin{array}{c}
\displaystyle{
 H= \frac{1}{2} \sum_{k=1}^{n} p_k^2 - g^2
\sum_{i<j}^{n} \frac{1}{q_{ij}^2}  -
\frac{\nu_0^2}{2}\sum_{k=1}^{n}q^2_k\,.
}
\end{array}
\eq
leads to the well-known Polychronakos--Frahm spin chain \cite{Poly}, \cite{Frahm}:
\beq\label{PF}
\begin{array}{c}
	\displaystyle{
		\mH=-\frac{g}{2}
\sum_{i<j}^n \frac{1}{z_{ij}^2}\Big(\stackrel{i}{\sigma_{1}}\otimes\stackrel{j}
{\sigma_{1}} +\stackrel{i}{\sigma_{2}}\otimes \stackrel{j}{\sigma_{2}}
        +\stackrel{i}{\sigma_{3}}\otimes \stackrel{j}{\sigma_{3}}+\sigma_{0}^{\otimes n}\Big)\,,
        }
\end{array}
\eq
where $z_i$ are (scaled) zeros of the Hermite polynomials \cite{Sas1} and obey the following equations:
\beq\label{PF44}
\begin{array}{c}
\displaystyle{
  2g^2
\sum_{j\neq i}^{n}\frac{1}{z_{ij}^3}-\nu_0^2z_i=0\,.
}
\end{array}
\eq

\paragraph{The system $II$:}
\beq\label{PF4}
\begin{array}{c}
\
\displaystyle{
 H= \frac{1}{2} \sum_{k=1}^{n} p_k^2 - g^2
\sum_{i<j}^{n} \frac{1}{q_{ij}^2}  -
\frac{1}{2}\sum_{k=1}^{n}\Big((2\gamma\rho+\nu^2) q^2_k+\gamma^2q^4_k\Big)
}
\end{array}
\eq
yields the same form of spin-chain Hamiltonian as (\ref{PF}):
\beq\label{sPF4}
\begin{array}{c}
	\displaystyle{
		\mH=-\frac{g}{2} \sum_{i<j}^n
 \frac{1}{z_{ij}^2}\Big(\stackrel{i}{\sigma_{1}}\otimes\stackrel{j}{\sigma_{1}} +\stackrel{i}{\sigma_{2}}\otimes \stackrel{j}{\sigma_{2}}
        +\stackrel{i}{\sigma_{3}}\otimes \stackrel{j}{\sigma_{3}}+\sigma_{0}^{\otimes n}\Big)\,,
        }
\end{array}
\eq
with different spin positions $z_i$, determined by
\beq\label{x4eq}
\begin{array}{c}
\displaystyle{
  2g^2
\sum_{j\neq i}^{n}\frac{1}{z_{ij}^3} -(2\gamma\rho+\nu^2) z_i-2\gamma^2z_i^3=0\,.
}
\end{array}
\eq
For this model, real solutions of (\ref{x4eq}) exist under the conditions stated in the following proposition.
\begin{proposition}\label{prop62}
     For $g$, $\ga$, $\nu$ $\in$ $\mathbb{R}\backslash\{0\}$,  $\rho$ $\in \mathbb{R}$,
    %some fixed value of $\rho$ and sufficiently small value $|\gamma|$,
    equations (\ref{x4eq}) admit a real solution $z=(z_1,\dots,z_n)$ in
    the domain $D=\{(q_1,\dots,q_n)\in\mathbb{R}^n|-\infty<q_1<q_2\dots<q_n<\infty\}$.
    For sufficiently small values of $|\gamma|$, such a solution can be chosen
    close to the vector of scaled Hermite zeros.
\end{proposition}
% The proof of the statement totally reproduces the proof of the Proposition \ref{p71}.
%
%
%\begin{proposition}\label{prop62}
%Fix $n\geq 2$, $g,\nu\in\mathbb{R}\setminus\{0\}$ and
%$\rho\in\mathbb{R}$. There exists $\varepsilon>0$ such that, for every
%real $\gamma$ with $|\gamma|<\varepsilon$, equations~(\ref{x4eq})
%admit a real solution in the domain
%\[
% D=\{(q_1,\ldots,q_n)\in\mathbb{R}^n:q_1<\cdots<q_n\}.
%\]
%More precisely, the unique solution in $D$ at $\gamma=0$ extends to
%a real-analytic branch $z(\gamma)\in D$. This branch is locally unique
%near $z(0)$.
%\end{proposition}
%
The argument used in Proposition~\ref{p71} also applies here.
Alternatively, one can first consider the case $\gamma=0$ and then continue
the resulting solution to sufficiently small real values of $\gamma$.

\paragraph{The system $III$:} The classical Hamiltonian
\beq\label{HS}
\begin{array}{c}
\
\displaystyle{
H= \frac{1}{2} \sum_{k=1}^{n} p_k^2 - g^2
\sum_{i<j}^{n} \frac{1}{\sinh^2(\frac{\beta q_{ij}}{2})}  +
\frac{1}{2}\sum_{k=1}^{n}\Big(2\gamma\rho\cosh(\beta q_k+\delta)+\frac{\gamma^2}{2}\cosh(2\beta q_k+2\delta)\Big)
}
\end{array}
\eq
leads to the non-translation-invariant version of the Haldane--Shastry spin chain (recall (\ref{q5122}) and (\ref{r724})):
\beq
\begin{array}{c}\label{sHS}
	\displaystyle{
		\mH=-\frac{g\beta}{4} \sum_{i<j}^n \frac{1}{\sinh^2 \frac{\beta z_{ij}}{2}}(\stackrel{i}{\sigma_{1}}\otimes\stackrel{j}{\sigma_{1}} +\stackrel{i}{\sigma_{2}}\otimes \stackrel{j}{\sigma_{2}}
        +\stackrel{i}{\sigma_{3}}\otimes \stackrel{j}{\sigma_{3}}+\sigma_{0}^{\otimes n})\,.
       }
\end{array}
\eq
In this case, the equilibrium points $z_i$ satisfy the following equations:
\beq\label{sHS22}
\begin{array}{c}
\
\displaystyle{
\gamma\rho\be\sinh(\beta z_{i}+\delta)+\frac{\gamma^2\be}{2}\sinh(2\beta z_{i}+2\delta)
+\beta g^2\sum_{j:j\neq i}^{n} \frac{\cosh(\frac{\beta z_{ij}}{2})}{\sinh^3(\frac{\beta z_{ij}}{2})}  =0\,,\quad
i=1,...,n
}
\end{array}
\eq
and are related to polynomial solutions of a Hill equation \cite{Inoz08}.
For the standard Haldane--Shastry spin chain, $z_k=2\pi ik / (n\be)$.

A possible route to proving the integrability of (\ref{sPF4}) and (\ref{sHS}) is suggested by the connection between
 (\ref{a170}) and (\ref{HS}), (\ref{PF4}) described in \cite{Ta2}, \cite{Osh}. The chains (\ref{sHS}) and (\ref{sPF4})
 are expected to arise as limits of the isotropic $BC_n$ Inozemtsev spin chain,
 whose integrability was proved in \cite{Ch3}. We plan to clarify this limiting procedure
 and construct examples of higher Hamiltonians in future work.

%%%%%%%%%%%%%%%%%%%%%%%%%%%%%%%%%%%%%%%%%%%%%%%%%%%%%%%%%%%%%%%%%%%%%%%%%%%%%%%%%%%%%%%%%%%%%%%%%%%%%%%%%%
%%%%%%%%%%%%%%%%%%%%%%%%%%%%%%%%%%%%%%%%%%%%%%%%%%%%%%%%%%%%%%%%%%%%%%%%%%%%%%%%%%%%%%%%%%%%%%
%%%%%%%%%%%%%%%%%%%%%%%%%%%%%%%%%%%%%%%%%%%%%%%%%%%%%%%%%%%%%%%%%%%%%%%%%%%%%%%%%%%%%%%%%%%%%%

\subsection*{Acknowledgments}
%We are grateful to O. Chalykh, An. Kirillov, A. Liashyk, V. Prokofev and I. Sechin for useful discussions.
%
%This work was supported by the Russian Science Foundation under grant no.\ 25-11-00081\\
%(\texttt{https://rscf.ru/en/project/25-11-00081/})
%and was carried out at the Steklov Mathematical Institute of the Russian Academy of Sciences.

The work of M. Matushko and A. Zotov was performed at the Steklov International Mathematical Center and supported by the Ministry of Science and Higher Education of the Russian Federation (agreement no. 075-15-2025-303).
The work of A. Mostovskii was supported by the Foundation
for the Advancement of Theoretical Physics and Mathematics ''BASIS''.

%%%%%%%%%%%%%%%%%%%%%%%%%%%%%%%%%%%%%%%%%%%%%%%%%%%%%%%%%%%%%%%%%%%%%%%%%%%%%%%%%%%%%%%%%%%%%%%%%
\section{Appendix: elliptic functions}\label{secA}
\def\theequation{A.\arabic{equation}}
\setcounter{equation}{0}

\paragraph{Function $\phi$.}
The main object is the elliptic Kronecker function \cite{Weil}:
 \beq\label{a901}
  \begin{array}{l}
  \displaystyle{
 \phi(z,u)=\frac{\vth'(0)\vth(z+u)}{\vth(z)\vth(u)}=\phi(u,z)\,,\quad
 \res\limits_{z=0}\phi(z,u)=1\,,
  \quad \phi(-z, -u) = -\phi(z, u)\,,
 }
 \end{array}
 \eq
where $\vth(z)$ is the first Jacobi theta function:
 \beq\label{a902}
 \begin{array}{c}
  \displaystyle{
\vth(z)=\vth(z,\tau)\equiv-\theta{\left[\begin{array}{c}
1/2\\
1/2
\end{array}
\right]}(z|\, \tau )\,,
 }
 \end{array}
 \eq
\beq\label{a903}
 \begin{array}{c}
  \displaystyle{
\theta{\left[\begin{array}{c}
a\\
b
\end{array}
\right]}(z|\, \tau ) =\sum_{j\in \mZ}
\exp\left(2\pi\imath(j+a)^2\frac\tau2+2\pi\imath(j+a)(z+b)\right)\,,\quad {\rm Im}(\tau)>0
}
 \end{array}
 \eq
and $\tau$ is the modulus of the elliptic curve $\mC/(\mZ+\tau\mZ)$, ${\rm Im}(\tau)>0$.

Recall the Jacobi theta functions:
\beq\label{8v1}
\begin{array}{c}
\displaystyle{
\theta_1(u|\tau )=\vth(u,\tau)=-i\sum_{k\in  \mZ}
(-1)^k q^{(k+\frac{1}{2})^2}e^{\pi i (2k+1)u},
\qquad
\theta_2(u|\tau )=\sum_{k\in  \mZ}
q^{(k+\frac{1}{2})^2}e^{\pi i (2k+1)u},}
\\ \\
\displaystyle{
\theta_3(u|\tau )=\sum_{k\in  \mZ}
q^{k^2}e^{2\pi i ku},
\qquad
\theta_4(u|\tau )=\sum_{k\in  \mZ}
(-1)^kq^{k^2}e^{2\pi i ku},}
\end{array}
\eq
where
$q=e^{\pi i \tau}$ and
\beq\label{8v11}
\begin{array}{c}
\theta_2(u|\tau)=\theta_1(u+\frac{1}{2}|\tau), \! \quad \! \!
\theta_3(u|\tau)=q^{\frac{1}{4}}e^{\pi iu}\theta_1(u+\frac{\tau+1}{2}|\tau), \! \quad \!\!
\theta_4(u|\tau)=-iq^{\frac{1}{4}}e^{\pi iu}\theta_1(u+\frac{\tau}{2}|\tau).
\end{array}
\eq
The infinite product representation for the Jacobi theta function is
%Their infinite product representations are also useful:
\beq\label{8v2}
\begin{array}{l}
\displaystyle{
\theta_1(u|\tau )=2q^{\frac{1}{4}} \sin \pi u \prod_{n\geq 1}
(1-q^{2n})(1-q^{2n}e^{2\pi i u})(1-q^{2n}e^{-2\pi i u}),
}
\\ \\
\displaystyle{
\theta_2(u|\tau )=2q^{\frac{1}{4}} \cos \pi u \prod_{n\geq 1}
(1-q^{2n})(1+q^{2n}e^{2\pi i u})(1+q^{2n}e^{-2\pi i u}),
}
\\ \\
\displaystyle{
\theta_3(u|\tau )= \prod_{n\geq 1}
(1-q^{2n})(1+q^{2n-1}e^{2\pi i u})(1+q^{2n-1}e^{-2\pi i u}),
}
\\ \\
\displaystyle{
\theta_4(u|\tau )= \prod_{n\geq 1}
(1-q^{2n})(1-q^{2n-1}e^{2\pi i u})(1-q^{2n-1}e^{-2\pi i u}).
}
\end{array}
\eq
The partial derivative $f(z,u) = \partial_u \phi(z,u)$ equals
\beq\label{a904}
\begin{array}{c} \displaystyle{
    f(z, u) = \phi(z, u)(E_1(z + u) - E_1(u)), \qquad f(-z, -u) = f(z, u)\,,
}\end{array}\eq
where $E_1(z)$ is the first
  Eisenstein function. The first and second Eisenstein functions are defined as follows:
\beq\label{a905}
\begin{array}{c} \displaystyle{
    E_1(z)=\frac{\vth'(z)}{\vth(z)}=\zeta(z)+\frac{z}{3}\frac{\vth'''(0)}{\vth'(0)}\,,
    \quad
    E_2(z) = - \partial_z E_1(z) = \wp(z) - \frac{\vartheta'''(0) }{3\vartheta'(0)}\,,
}\end{array}\eq
\beq\label{a906}
\begin{array}{c}
 \displaystyle{
    E_1(- z) = -E_1(z)\,, \quad E_2(-z) = E_2(z)\,,
}\end{array}
\eq
where $\wp(z)$ and $\zeta(z)$ are the Weierstrass functions.
The local behaviour of $\phi(z,q)$ near $z=0$ is as follows:
  \beq\label{a9061}
  \begin{array}{l}
  \displaystyle{
\phi(z,q)=z^{-1}+E_1(q)+z\,(E_1^2(q)-\wp(q))/2+O(z^2)\,.
 }
 \end{array}
 \eq
Therefore,
  \beq\label{a9062}
  \begin{array}{l}
  \displaystyle{
 f(0,q)=-E_2(q)\,.
 }
 \end{array}
 \eq
 Addition formulae:
\beq\label{a907}
  \begin{array}{c}
  \displaystyle{
  \phi(z_1, u_1) \phi(z_2, u_2) = \phi(z_1, u_1 + u_2) \phi(z_2 - z_1, u_2) + \phi(z_2, u_1 + u_2) \phi(z_1 - z_2, u_1)\,.
 }
 \end{array}
 \eq
Applying $\p_{u_2}-\p_{u_1}$ gives
\beq\label{a9071}
  \begin{array}{c}
  \displaystyle{
  \phi(z_1, u_1) f(z_2, u_2)-f(z_1,u_1)\phi(z_2,u_2) =
  }
  \\ \ \\
  \displaystyle{
  =
  \phi(z_1, u_1 + u_2) f(z_2 - z_1, u_2) - \phi(z_2, u_1 + u_2) f(z_1 - z_2, u_1)\,.
 }
 \end{array}
 \eq
We also have
\beq\label{a908}
  \begin{array}{c}
  \displaystyle{
 \phi(z,u_1)\phi(z,u_2)=\phi(z,u_1+u_2)\Big(E_1(z)+E_1(u_1)+E_1(u_2)-E_1(z+u_1+u_2)\Big)\,,
 }
 \end{array}
 \eq
\beq\label{a909}
  \begin{array}{c}
  \displaystyle{
  \phi(z, u_1) f(z,u_2)-\phi(z, u_2) f(z,u_1)=\phi(z,u_1+u_2)\Big(\wp(u_1)-\wp(u_2)\Big)=
 }
 \\ \ \\
   \displaystyle{
 =\phi(z,u_1+u_2)\Big(E_2(u_1)-E_2(u_2)\Big)=\phi(z,u_1+u_2)\Big(f(0,u_2)-f(0,u_1)\Big)\,.
 }
 \end{array}
 \eq
\beq\label{a910}
  \begin{array}{c}
  \displaystyle{
  \phi(z, u) \phi(z, -u) = \wp(z)-\wp(u)=E_2(z)-E_2(u)\,,
 }
 \end{array}
 \eq
\beq\label{a911}
  \begin{array}{c}
  \displaystyle{
  \phi(z, u) f(z, -u)-\phi(z, -u) f(z, u)=\wp'(u)\,.
 }
 \end{array}
 \eq
Define the half-periods
 \beq\label{w209}
 \begin{array}{c}
  \displaystyle{
 \Om_0=0\,,\quad
 \Om_1=\frac{1}{2}\,,\quad
 \Om_2=\frac{1+\tau}{2}\,,\quad
 \Om_3=\frac{\tau}{2}
  }
 \end{array}
\eq
 and
 \beq\label{w210}
 \begin{array}{c}
  \displaystyle{
 \vf_{0}(z,\hbar)=\phi(z,\hbar)\,,
\qquad
 \vf_{1}(z,\hbar+\Om_1)=\phi(z,\Om_1+\hbar)\,,
 }
 \\ \ \\
  \displaystyle{
 \vf_{2}(z,\hbar+\Om_2)=e^{\pi\imath z}\phi(z,\Om_2+\hbar)\,,
\qquad
  \vf_{3}(z,\hbar+\Om_3)=e^{\pi\imath z}\phi(z,\Om_3+\hbar)\,.
 }
  \end{array}
 \eq
Then
 \beq\label{w211}
 \begin{array}{c}
  \displaystyle{
 \vf_{k}(z,\hbar+\Om_k)=\frac{\theta_1'(0)\theta_{k+1}(z+\hbar)}{\theta_{1}(z)\theta_{k+1}(\hbar)}\,,
 \quad k=0,1,2,3
  }
 \end{array}
\eq
 The functions introduced above transform under the matrix $I$ (\ref{q416})-(\ref{q417}) as follows:
  \beq\label{w215}
 \begin{array}{c}
  \left(\begin{array}{c}
 \vf_{0}(2z,u+\Om_0)
 \\
  \vf_{1}(2z,u+\Om_1)
 \\
  \vf_{2}(2z,u+\Om_2)
 \\
  \vf_{3}(2z,u+\Om_3)
 \end{array}\right)
 =
   \displaystyle{\frac12}
 \left(\begin{array}{cccc}
 1 & 1 & 1 & 1
 \\
  1 & 1 & -1 & -1
 \\
  1 & -1 & 1 & -1
 \\
  1 & -1 & -1 & 1
 \end{array}\right)
  \left(\begin{array}{c}
 \vf_{0}(2u,z+\Om_0)
 \\
  \vf_{1}(2u,z+\Om_1)
 \\
  \vf_{2}(2u,z+\Om_2)
 \\
  \vf_{3}(2u,z+\Om_3)
 \end{array}\right)\,.
 \end{array}
 \eq
  This transformation matrix can be written as
  \beq\label{q417}
  \begin{array}{c}
  \displaystyle{
I_{km}=\frac12\exp\Big(4\pi\imath\Big(\Om_{m-1}\p_\tau\Om_{k-1}
-\Om_{k-1}\p_\tau\Om_{m-1}\Big)\Big)\,,\quad k,m=1,...,4\,.
 }
 \end{array}
 \eq
It has the property
  \beq\label{q418}
  \begin{array}{c}
  \displaystyle{
I^{-1}=I\,.
 }
 \end{array}
 \eq
  We also use the standard identity:
  \beq\label{w216}
 \begin{array}{c}
  \displaystyle{
\sum\limits_{a=0}^3\wp(z+\Om_a)=4\wp(2z)\,.
  }
 \end{array}
 \eq
The quasi-periodic properties are as follows:
  \beq\label{a90621}
  \begin{array}{l}
  \displaystyle{
 \phi(z+1,u)=\phi(z,u)\,,\qquad  \phi(z+\tau,u)=e^{-2\pi\imath u}\phi(z,u)
 }
 \end{array}
 \eq
or
  \beq\label{a90622}
  \begin{array}{l}
  \displaystyle{
 \phi(z\pm 2\Om_a,u)=e^{\mp4\pi\imath\p_\tau\Om_a u}\phi(z,u)\,.
 }
 \end{array}
 \eq
For the theta function (\ref{a902}) we have
  \beq\label{a90623}
  \begin{array}{l}
  \displaystyle{
 \vth(z+2\Om_a)=-e^{-4\pi\imath(z+\Om_a)\p_\tau\Om_a}\vth(z)\,,\quad a\neq 0
 }
 \end{array}
 \eq
or
  \beq\label{a90624}
  \begin{array}{l}
  \displaystyle{
 \vth(z+\Om_a)=-e^{-4\pi\imath z\p_\tau\Om_a}\vth(z-\Om_a)\,,\quad a\neq 0\,.
 }
 \end{array}
 \eq
This gives
  \beq\label{a90625}
  \begin{array}{l}
  \displaystyle{
 E_2(z+2\Om_a)=E_2(z)\,,\qquad E_1(z+2\Om_a)=E_1(z)-4\pi\imath\p_\tau\Om_a
 }
 \end{array}
 \eq
and
  \beq\label{a90626}
  \begin{array}{l}
  \displaystyle{
E_1(\Om_a)=-2\pi\imath\p_\tau\Om_a\,,\quad a\neq 0\,.
 }
 \end{array}
 \eq
It follows from (\ref{a90626}) and (\ref{a904}), (\ref{a907}) that for $a=1,2,3$
  \beq\label{a90627}
  \begin{array}{l}
  \displaystyle{
\p_z\vf_a(z,\Om_a)=\vf_a(z,\Om_a)(E_1(z+\Om_a)-E_1(z)-E_1(\Om_a))=
-\frac{\vf_1(z,\Om_1)\vf_2(z,\Om_2)\vf_3(z,\Om_3)}{\vf_a(z,\Om_a)}\,.
 }
 \end{array}
 \eq

\paragraph{Discrete Fourier transform.}
Writing $\phi(z,u;\tau)$ for the Kronecker function (\ref{a901}), one has,
for every $m\in\mZ_N$,
\beq\label{a911m}
\begin{array}{c}
\displaystyle{
\frac1N\sum_{k\in\mZ_N}
\exp\!\Big(\frac{2\pi\imath\,mk}{N}\Big)
\phi\!\Big(z+\frac{k}{N},\,u;\,\tau\Big)
=
\exp\!\big(-2\pi\imath\,mz\big)\,
\phi(Nz,\,u-m\tau;\,N\tau)\,.
}
\end{array}
\eq
The companion sum along the $\tau$-direction is, for every $m\in\mZ_N$,
\beq\label{a911t}
\begin{array}{c}
\displaystyle{
\sum_{k\in\mZ_N}
\exp\!\Big(\frac{2\pi\imath\,k(u+m)}{N}\Big)
\phi\!\Big(z+\frac{k\tau}{N},\,u;\,\tau\Big)
=
\phi\!\Big(z,\,\frac{u+m}{N};\,\frac{\tau}{N}\Big)\,.
}
\end{array}
\eq

\paragraph{Function $v(z,u)$.}
Using
the functions (\ref{w209})-(\ref{w210}), define
  \beq\label{q413}
  \begin{array}{c}
  \displaystyle{
v(z,u)\equiv v(z,u|\nu)=\sum\limits_{a=0}^3\nu_a\vf_a(2z,u+\Om_a)=
\sum\limits_{a=0}^3\nu_a \exp(4\pi \imath z\p_\tau\Om_a)\phi(2z,u+\Om_a)
 }
 \end{array}
 \eq
and its partial derivative
  \beq\label{q414}
  \begin{array}{c}
  \displaystyle{
v'(z,u)=\p_uv(z,u)=\sum\limits_{a=0}^3\nu_a \exp(4\pi \imath z\p_\tau\Om_a)f(2z,u+\Om_a)\,.
 }
 \end{array}
 \eq
Some identities for the functions $v(z,u)$ and $v'(z,u)$ can be found in
 \cite{KH,KH2,Hikami,Z04}.
In particular,
\beq\label{q419a}
v(z,u)v(-z,u)=\sum\limits_{a=0}^3\Big(\nu_a^2\wp(u+\Om_a)-\bar\nu_a^2\wp(z+\Om_a)\Big)=-v(z,u)v(z,-u)\,.
\eq
where, as in (\ref{w215})
  \beq\label{q416}
 \begin{array}{c}
  \left(\begin{array}{c}
 \bar\nu_0
 \\
 \bar\nu_1
 \\
 \bar\nu_2
 \\
 \bar\nu_3
 \end{array}\right)
 =
   \displaystyle{\frac12}
 \left(\begin{array}{cccc}
 1 & 1 & 1 & 1
 \\
  1 & 1 & -1 & -1
 \\
  1 & -1 & 1 & -1
 \\
  1 & -1 & -1 & 1
 \end{array}\right)
  \left(\begin{array}{c}
 \nu_0
 \\
 \nu_1
 \\
 \nu_2
 \\
 \nu_3
 \end{array}\right)\,.
 \end{array}
 \eq
From (\ref{w215}) we conclude that
  \beq\label{q419}
  \begin{array}{c}
  \displaystyle{
v(z,u|\nu)=v(u,z|\bar\nu)=\sum\limits_{a=0}^3\bar\nu_a\vf_a(2u,z+\Om_a)\,.
 }
 \end{array}
 \eq
Obviously,
\beq\label{a9ntiphi}
v(-z,-u)=-v(z,u).
\eq
By differentiating (\ref{q419a}) with respect to $u$ we get:
  \beq\label{q420}
  \begin{array}{c}
  \displaystyle{
v(z,u)v'(z,-u)-v(z,-u)v'(z,u)=\sum\limits_{a=0}^3\nu_a^2\wp'(u+\Om_a)\,.
 }
 \end{array}
 \eq
In our previous paper \cite{MMZ} we used and proved the identities
\beq\label{hident1}
  \begin{array}{c}
  \displaystyle{
v(x,u)\phi(x+y,w-u)+v(x,w)\phi(x-y,u-w)+v(y,-u)\phi(x+y,u+w)=
 }
 \\ \ \\
   \displaystyle{
=v(y,w)\phi(x-y,u+w)
}
\end{array}
\eq
and
\beq\label{hident3}
\begin{array}{c}
\phi(z,u-w)(v'(0,w)-v'(0,u))=
\\ \\
2v(-z,w)f(z,u+w)+2v(z,u)f(-z,u+w)+v'(-z,w)\phi(z,u+w)+v'(z,u)\phi(-z,u+w)\,,
\end{array}
\eq
where
  \beq\label{q441}
  \begin{array}{c}
  \displaystyle{
v'(0,u)=\sum\limits_{a=0}^3\nu_a f(0,u+\Om_a)=-\sum\limits_{a=0}^3\nu_a E_2(u+\Om_a)=
-\sum\limits_{a=0}^3\nu_a \wp(u+\Om_a)+\frac{\vth'''(0)}{3\vth'(0)}\sum\limits_{a=0}^3\nu_a\,.
 }
 \end{array}
 \eq
%

%\subsection*{Acknowledgments}
%We are grateful to O. Chalykh, An. Kirillov, A. Liashyk, V. Prokofev and I. Sechin for useful discussions.

%This work was performed at the Steklov International Mathematical Center and supported by the Ministry of Science %and Higher Education of the Russian Federation (agreement no. 075-15-2022-265).
%\addcontentsline{toc}{section}{\hspace{6mm}Acknowledgments}

%This work was supported by the Russian Science Foundation under grant no. 19-11-00062,\\ %https://rscf.ru/en/project/19-11-00062/ .

%This work was performed at the Steklov International Mathematical Center and supported by the Ministry of Science %and Higher Education of the Russian Federation (agreement no. 075-15-2022-265).

%%%%%%%%%%%%%%%%%%%%%%%%%%%%%%%%%%%%%%%%%%%%%%%%%%%%%%%%%%%%%%%%%%%%%%%%%%%%%%%%%%%%%%%%%%%%%%
%%%%%%%%%%%%%%%%%%%%%%%%%%%%%%%%%%%%%%%%%%%%%%%%%%%%%%%%%%%%%%%%%%%%%%%%%%%%%%%%%%%%%%%%%%%%%%

\begin{small}

\end{small}


\begin{thebibliography}{99}
\addcontentsline{toc}{section}{References}

\bibitem{Aguiar}
M. Aguiar,
{\it Infinitesimal Hopf algebras},
New trends in Hopf algebra theory,
Contemporary Mathematics 267
(2000) 1--29.

M. Aguiar,
{\it On the associative analog of Lie bialgebras},
J. Algebra 244 (2001) 492--532.

%\bibitem{Ahmed} S. Ahmed,
% {\it A general technique to obtain nonlinear equations
% for the zeros of classical orthogonal polynomials}, Lett. Nuovo Cimento, 26 (1979) 285--288.

\bibitem{AHZ} A. Antonov, K. Hasegawa, A. Zabrodin,
{\it On trigonometric intertwining vectors and non-dynamical R-matrix for the Ruijsenaars model},
Nuclear Physics B, 503 (1997), 747--770; arXiv:hep-th/9704074.

\bibitem{Ch2} Ph. Argyres, O. Chalykh, Y. L\"u,
{\it Inozemtsev System as Seiberg-Witten Integrable system},
 J. High Energ. Phys. 2021, 51 (2021);
	arXiv:2101.04505 [hep-th].

\bibitem{AtZ5} K.R. Atalikov, A.V. Zotov,
 {\it Higher-rank generalization of the 11-vertex rational R-matrix:
 IRF–vertex relations and the associative Yang–Baxter equation},
 Theoret. and Math. Phys., 216:2 (2023), 1083--1103; 	arXiv:2303.02391 [math-ph].

\bibitem{BB}
R.J. Baxter,
 {\it Partition function of the eight-vertex lattice model},
  Ann. Phys. 70 (1972) 193--228.

A.A. Belavin,
{\it Dynamical symmetry of integrable quantum systems},
Nucl. Phys. B, 180 (1981) 189--200.



%\bibitem{BD}
%A. Belavin, V. Drinfeld,
%{\it Solutions of the classical Yang–Baxter equation for simple Lie algebras},
%Functional Analysis and Its Applications, 16:3 (1982) 159--180.

\bibitem{BPS} D. Bernard, V. Pasquier, D. Serban,
{\it Exact solution of long-range interacting spin chains with boundaries},
Europhys. Lett. 30 (1995) 301--306; arXiv:hep-th/9501044.

\bibitem{BCR} M. Bertola, M. Cafasso, V. Rubtsov,
 {\it Noncommutative Painlevé Equations and Systems of Calogero Type}, Commun. Math. Phys. 363 (2018) 503--530;
 arXiv:1710.00736 [math-ph].



\bibitem{BS} R. Bittleston, D. Skinner,
{\it Gauge theory and boundary integrability. Part II.
Elliptic and trigonometric cases}, JHEP 06 (2020) 080; arXiv:1912.13441 [hep-th].

R. Bittleston, {\it Integrability from Chern-Simons
Theories}, PhD thesis, 2021.

\bibitem{BoSa} A.J. Bordner, R. Sasaki, {\it Calogero-Moser Models III: Elliptic Potentials and Twisting},
	Prog. Theor. Phys. 101 (1999) 799--829; 	arXiv:hep-th/9812232.

\bibitem{BFV} V.M. Buchstaber, G. Felder, A.V. Veselov,
{\it Elliptic Dunkl operators, root systems, and functional equations},
Duke Math. J. 76 (1994) 885--911; 	arXiv:hep-th/9403178.

\bibitem{Ch1} O. Chalykh,
{\it Quantum Lax pairs via Dunkl and Cherednik operators},
Commun. Math. Phys. 369, 261–316 (2019);
	arXiv:1804.01766 [math.QA].

\bibitem{Ch3}
O. Chalykh, {\it Integrability of the Inozemtsev spin chain}, arXiv:2407.03276 [nlin.SI].

\bibitem{ChM} O. Chalykh, M. Matushko,
{\it R-matrix Dunkl operators and spin Calogero--Moser system},
Annales Henri Poincar\'e (2026);
DOI:10.1007/s00023-026-01671-6;
arXiv:2509.18989 [math.QA].


%\bibitem{DP}
%E. D' Hoker,  D.H. Phong,
% {\it Calogero-Moser Lax Pairs with Spectral Parameter for General Lie Algebras},
%Nuclear Physycs B 530 (1998) 537--610; hep-th/9804124.

\bibitem{Chered} I.V. Cherednik,
  {\it On a method of constructing factorized S-matrices in elementary functions},
 Theoret. and Math. Phys., 43:1 (1980) 356--358.

\bibitem{Cherednik} I.V. Cherednik,
{\it Factorizing particles on a half-line and root systems},
Theoret. and Math. Phys. 61 (1984) 977--983.

\bibitem{Cong} Cong-xin Liu, Guo-xing Ju, Shi-kun Wang, Ke Wu,
{\it Classification of Solutions to Reflection Equation of Two-Component Systems},
	J. Phys. A: Math. Gen. 32 (1999) 3505--3523;
	arXiv:hep-th/9808083.

\bibitem{Vega} H.J. de Vega, A. Gonzalez-Ruiz, {\it Boundary K-matrices for the XYZ, XXZ and XXX spin chains},
 J. Phys. A: Math. Gen. 27 (1994) 6129--6138; 	arXiv:hep-th/9306089.

 H.J. de Vega, A. Gonzalez-Ruiz, {\it Boundary K-Matrices for the Six Vertex and the $n(2n-1)$ $A_{n-1}$ Vertex Models},
 J. Phys. A: Math. Gen. 26 (1993) L519--L524; arXiv:hep-th/9211114.

  \bibitem{DiF} P. Di Francesco, P. Mathieu and D. S{\'e}n{\'e}chal,
{\it Conformal Field Theory}, (1999).

\bibitem{EFGR} A. Enciso, F. Finkel, A. González-Lopez, M.A. Rodríguez,
 {\it Haldane–Shastry spin chains of $BC_N$ type}, Nuclear Physics B 707:3 (2005) 553--576;
 	arXiv:hep-th/0406054.

J.C. Barba, F. Finkel, A. González-Lopez, M.A. Rodríguez,
{\it The Polychronakos-Frahm spin chain of $BC_N$ type and Berry-Tabor's conjecture},
Physical Review B, 77 (2008) 214422;
	arXiv:0803.0922 [cond-mat.stat-mech].

\bibitem{EPSS} B. Estienne, V. Pasquier, R. Santachiara, D. Serban,
{\it Conformal blocks in Virasoro and W theories: Duality and the Calogero–Sutherland model},
Nuclear Physics B,
860:3 (2012) 377--420;	arXiv:1110.1101 [hep-th].

\bibitem{FP} G. Felder, V. Pasquier,
{\it A simple construction of elliptic R-matrices},
Lett. Math. Phys. 32 (1994) 167--171; arXiv:hep-th/9402011.

\bibitem{FK} S. Fomin, A.N. Kirillov,
{\it Quadratic algebras, Dunkl elements, and Schubert calculus},
Advances in geometry; Progress in Mathematics book series, Vol. 172
 (1999) 147--182.

\bibitem{Frahm} H. Frahm,
 {\it Spectrum of a spin chain with inverse square exchange}, J. Phys. A: Math. Gen. 26 (1993) L473.


\bibitem{FM} L. Freidel, J.M. Maillet,
{\it Quadratic algebras and integrable systems},
Physics Letters B, 262 (1991) 278--284.


%\bibitem{GW} F. Gesztesy, R. Weikard,
%{\it Treibich-Verdier potentials and the stationary (m) KdV hierarchy},
%Mathematische Zeitschrift, 219 (1995) 451--476.
%
%V. Prokofev, unpublished.

\bibitem{GrZ18} A. Grekov, A. Zotov,
{\it On R-matrix valued Lax pairs for Calogero–Moser models},
J. Phys. A: Math. Theor., 51 (2018), 315202;
	arXiv:1801.00245 [math-ph].


\bibitem{GSZ}
A. Grekov, I. Sechin, A. Zotov,
 {\it Generalized model of interacting integrable tops},  JHEP 10 (2019) 081;
 	arXiv:1905.07820 [math-ph].




\bibitem{HS1}
 F.D.M. Haldane, {\it Exact Jastrow-Gutzwiller resonating-valence-bond ground state of the spin-$\frac{1}{2}$ antiferromagnetic Heisenberg chain with $1/r^2$ exchange}, Phys. Rev. Lett. 60 (1988) 635--638.

B.S. Shastry, {\it Exact solution of an S=1/2 Heisenberg %antiferromagnetic
chain with long-ranged interactions}, Phys. Rev. Lett. 60 (1988) 639--642.


\bibitem{Hikami} K. Hikami, {\it Boundary K-matrix, elliptic Dunkl operator and quantum
many-body systems}, J. Phys. A: Math. Gen. 29 (1996) 2135--2147.




\bibitem{IK} T. Inami, H. Konno,
{\it Integrable XYZ spin chain with boundaries},
J. Phys. A Math. Gen. 27 (1994) L913--L918.

\bibitem{Inoz83}V.I. Inozemtsev,
{\it On the motion of classical integrable systems of interacting particles in an external field},
Physics Letters A, 98 (1983) 316-318.

\bibitem{Inoz89} V.I. Inozemtsev,
 {\it Lax representation with spectral parameter on a torus for integrable particle systems},
Lett. Math. Phys. 17 (1989) 11--17.

\bibitem{Inoz90} V.I. Inozemtsev,
 {\it On the connection between the one-dimensional S=1/2 Heisenberg chain and Haldane-Shastry model},
Journal of Statistical Physics, 59 (1990) 1143--1155.

J. Dittrich, V. Inozemtsev,
 {\it The commutativity of integrals of motion for quantum spin chains and elliptic functions identities},
Regular and Chaotic Dynamics, 13:1 (2008) 19--26; 	arXiv:0711.1973 [math-ph].

\bibitem{Inoz08} V.I. Inozemtsev,
 {\it Equilibrium points of classical integrable particle systems,
  factorization of wave functions of their quantum analogs and
  polynomial solutions of the Hill equation}, Regul. Chaot. Dyn. 13 (2008) 588--592.

\bibitem{Jimbo} M. Jimbo, {\it A $q$-difference analogue of $U(g)$ and the Yang--Baxter equation},
Lett. Math. Phys. 10 (1985) 63--69.

\bibitem{Kir} Anatol N. Kirillov, {\it On Some Quadratic Algebras I $\frac12$:
Combinatorics of Dunkl and Gaudin Elements, Schubert, Grothendieck, Fuss-Catalan, Universal Tutte and Reduced Polynomials}, 	
SIGMA 12 (2016), 002; arXiv:1502.00426 [math.RT].

\bibitem{KL} R. Klabbers, J. Lamers,
{\it The deformed Inozemtsev spin chain},
SciPost Phys. 17 (2024) 155; arXiv:2306.13066 [math-ph].

R. Klabbers, J. Lamers,
{\it Landscapes of integrable long-range spin chains},
	arXiv:2405.09718 [math-ph].

R. Klabbers, J. Lamers,
{\it Modular families of elliptic long-range spin chains from freezing},
Lett. Math. Phys. 116 (2026) 58; DOI:10.1007/s11005-026-02075-7;
arXiv:2507.13104 [math-ph].

\bibitem{KH} Y. Komori, K. Hikami,
{\it Quantum integrability of the generalized elliptic Ruijsenaars models},
1997 J. Phys. A: Math. Gen. 30, 4341.

Y. Komori, K. Hikami,
{\it Conserved operators of the generalized elliptic Ruijsenaars models},
J. Math. Phys. 39, 6175--6190 (1998).


\bibitem{KH2} Y. Komori, K. Hikami, {\it Elliptic K-matrix associated with Belavin's symmetric R-matrix}, Nucl. Phys. B 494 (1997) 687--701.


\bibitem{KrZ}
T. Krasnov, A. Zotov,
{\it Trigonometric integrable tops from solutions of associative Yang-Baxter equation},
 Annales Henri Poincar\'e, 20:8 (2019)
2671--2697;
 arXiv:1812.04209 [math-ph].


\bibitem{Krich1}
I.  Krichever,
{\it Elliptic solutions of the Kadomtsev--Petviashvili equation and integrable systems of particles},
Funct.  Anal.  Appl., 14:4 (1980) 282--290.


\bibitem{KuSkl} P.P. Kulish, E.K. Sklyanin,
{\it Algebraic Structures Related to Reflection Equations},
J. Phys. A 25 (1992) 5963--5976; 	arXiv:hep-th/9209054.

\bibitem{Lam} J. Lamers, V. Pasquier, D. Serban,
 {\it Spin-Ruijsenaars, q-deformed Haldane-Shastry and Macdonald polynomials},
 Commun. Math. Phys. 393 (2022) 61--150;
	arXiv:2004.13210 [math-ph].

\bibitem{LOZ14}
A. Levin, M. Olshanetsky, A. Zotov,
{\it Planck constant as spectral parameter in integrable systems and KZB equations},
JHEP 10 (2014) 109; arXiv:1408.6246 [hep-th].

%\bibitem{LOZ15}
%A.M. Levin, M.A. Olshanetsky, A.V. Zotov,
%{\it Quantum Baxter–Belavin R-matrices and multidimensional Lax pairs for Painlev\'e VI},
%    Theoret. and Math. Phys., 184:1 (2015), 924--939;	arXiv:1501.07351 [math-ph].

\bibitem{LOZ161}
%	A. Levin, M. Olshanetsky, A. Zotov,
% {\it Yang–Baxter equations with two Planck constants},
%  J. Phys. A: Math. Theor., 49:1 (2016), 14003 , 19 pp.,
%  Exactly Solved Models and Beyond: a special issue
%  in honour of R. J. Baxter's 75th birthday, arXiv:1507.02617 [math-ph].
%
	A. Levin, M. Olshanetsky, A. Zotov,
 {\it Noncommutative extensions of elliptic integrable Euler–Arnold tops and Painlev\'e VI equation},
  J. Phys. A, 49:39 (2016), 395202; 	arXiv:1603.06101 [math-ph].


\bibitem{LRS} A. Liashyk, N. Reshetikhin, I. Sechin,
{\it Quantum Integrable Systems on a Classical Integrable Background}, Commun. Math. Phys. (2026) 407:15;
arXiv:2405.17865 [math-ph].

\bibitem{MMZ} M. Matushko, A. Mostovskii, A. Zotov,
{\it R-matrix valued Lax pair for elliptic Calogero-Inozemtsev system and associative Yang-Baxter equations of ${\rm BC}_n$ type},
 Lett. Math. Phys. 116 (2026) 11;	arXiv:2503.22659 [math-ph].


\bibitem{MZ} M. Matushko, A. Zotov,
{\it Anisotropic spin generalization of elliptic Macdonald–Ruijsenaars operators and R-matrix identities},
Ann. Henri Poincar\'e, 24 (2023) 3373–3419;
	arXiv:2201.05944 [math.QA].

M. Matushko, A. Zotov,
{\it Elliptic generalisation of integrable q-deformed anisotropic Haldane–Shastry long-range spin chain},
Nonlinearity, 36:1 (2023) 319;
	arXiv:2202.01177 [math-ph].

\bibitem{MZ2} M. Matushko, A. Zotov, {\it Supersymmetric generalization
of q-deformed long-range spin chains of Haldane–Shastry type and trigonometric ${GL(N|M)}$
 solution of associative Yang–Baxter equation}, Nuclear Physics B, 1001 (2024) 116499;
 	arXiv:2312.04525 [math-ph].

M. Matushko,
{\it A solution of the associative Yang-Baxter equation related to the queer Lie superalgebra},
 Journal of Geometry and Physics, 218 (2025) 105671; 	arXiv:2412.19214 [math-ph].

\bibitem{Sas1} S. Odake, R. Sasaki,
{\it Polynomials associated with equilibrium positions in Calogero–Moser systems},
 J. Phys. A: Math. Gen. 35 (2002) 8283, arXiv:hep-th/0206172.

E. Corrigan, R. Sasaki,
{\it Quantum versus classical integrability in Calogero–Moser systems},
J. Phys. A: Math. Gen. 35 (2002) 7017.


%\bibitem{ORS} A. Odesskii, V. Rubtsov, V. Sokolov,
%{\it Parameter-dependent associative Yang-Baxter equations and Poisson brackets},
%International Journal of Geometric Methods in Modern Physics, 11:09 (2014) 1460036;
%arXiv:1311.4321 [math-ph].

\bibitem{OP} M.A. Olshanetsky, A.M. Perelomov,
 {\it Completely integrable Hamiltonian systems connected with semisimple Lie  algebras},
Inventiones mathematicae, 37:2 (1976) 93--108.

M.A. Olshanetsky, A.M. Perelomov,
{\it Completely integrable classical systems connected with semisimple Lie algebras},
Lett. Math. Phys.,
Vol. 1 (1976) 187--193.

M.A. Olshanetsky, A.M. Perelomov,
 {\it Classical integrable finite-dimensional systems related to Lie algebras},
Physics Reports, 71 (1981) 313--400.




\bibitem{Osh} T. Oshima,
 {\it Completely integrable systems associated with classical root systems}, SIGMA 3 (2007) 061;
 	arXiv:math-ph/0502028.


\bibitem{Pol} A. Polishchuk,
 {\it Classical Yang–Baxter equation and the $A^\infty$-constraint},
Advances in Mathematics 168:1 (2002)  56--95;
arXiv:math/0008156 [math.AG].

A. Polishchuk,
 {\it $A_\infty$-structures on an elliptic curve},
 Commun. Math. Phys. 247 (2004) 527--551;
arXiv: math/0001048 [math.AG].

\bibitem{Poly} A.P. Polychronakos,
 {\it Lattice integrable systems of Haldane-Shastry type}, Phys. Rev. Lett. 70 (1993) 2329--2331.

\bibitem{Poly2}
A.P. Polychronakos, {\it Exact Spectrum of SU(n) Spin Chain with Inverse-Square Exchange}, Nucl. Phys. B419 (1994) 553--566;
	arXiv:hep-th/9310095.

\bibitem{Poly3} A.P. Polychronakos, {\it A new integrable system with a quartic potential}, Physics Letters B, 276 (1992), 341--346.

\bibitem{Resh} N. Reshetikhin,
{\it Periodic and open classical spin Calogero--Moser chains},
Surv. Diff. Geom. 26 (2024); arXiv:2302.14281 [math-ph].

\bibitem{Sas2}R. Sasaki, K. Takasaki,
 {\it Quantum Inozemtsev model, quasi-exact solvability and $\mathcal{N}$-fold supersymmetry},
  J. Phys. A: Math. Gen. 34 (2001) 9533; arXiv:hep-th/0109008.

\bibitem{SP}
T. Schedler, {\it Trigonometric solutions of the associative Yang-Baxter equation},
Mathematical Research Letters, 10:3 (2003) 301--321;
arXiv:math/0212258 [math.QA].

A. Polishchuk, {\it Massey products on cycles of projective lines and
trigonometric solutions of the Yang–Baxter equations},
Algebra, Arithmetic, and Geometry, Volume 270 of the series Progress
in Mathematics (2010) 573--617; arXiv:math/0612761 [math.QA].


\bibitem{SeZ18} I. Sechin, A. Zotov,
{\it R-matrix-valued Lax pairs and long-range spin chains},
Phys. Lett. B, 781 (2018), 1–7; arXiv:1801.08908 [math-ph].

\bibitem{SU}
Y. Shibukawa, K. Ueno. {\it Completely $\mathbb Z$ symmetric R matrix}, Letters in Mathematical Physics 25 (1992): 239-248.

\bibitem{Skl-refl} E. Sklyanin,
{\it Boundary conditions for integrable quantum systems},
J. Phys. A: Math. Gen. 21 (1988) 2375--2389.



\bibitem{Smirnov}  A. Smirnov,
{\it Degenerate Sklyanin algebras}, Central European Journal of Physics, 8 (2010), 542-554;
arXiv:0903.1466 [math-ph].

\bibitem{Ta} K. Takasaki,
{\it Elliptic Calogero-Moser Systems and Isomonodromic Deformations},
	J. Math. Phys. 40 (1999) 5787--5821;
	arXiv:math/9905101 [math.QA].

\bibitem{Ta2} K. Takasaki, {\it Painlev\'e-Calogero Correspondence Revisited},
J. Math. Phys. 42, (2001) 1443–1473;
arXiv: math/0004118 [math.QA].

\bibitem{Tak} K. Takemura, {\it The Hermite–Krichever Ansatz for Fuchsian equations with applications to the sixth Painlev\'e equation and to finite-gap potentials}, Mathematische Zeitschrift  263 (2009) 149--194;

K. Takemura, {\it On finite-gap potential}, math/0504540 [math.CA].

\bibitem{Weil} A. Weil, {\it Elliptic functions according to Eisenstein and
Kronecker}, Springer-Verlag, %Berlin- Heidelberg-New York,
 (1976).

 D. Mumford, {\it Tata Lectures on Theta I, II},
Birkh\"auser, Boston, Mass. (1983, 1984).


\bibitem{Woj} S. Wojciechowski,
{\it On the integrability of the Calogero-Moser system in an external quartic potential and other many-body systems}, Physics Letters A, 102 (1984) 85--88.

%\bibitem{ZZ} A. Zabrodin, A. Zotov,
%{\it Field analogue of the Ruijsenaars-Schneider model},
%JHEP 07 (2022) 023;	arXiv:2107.01697 [math-ph].


\bibitem{Ya} T. Yamamoto, O. Tsuchiya,
 {\it Integrable $1/r^2$ spin chain with reflecting end}, J. Phys. A: Math. Gen. 29 (1996) 3977.

\bibitem{Yang} C.N. Yang, {\it Some exact results for the many-body problem in
one dimension with repulsive delta-function interaction},  Phys. Rev. Lett. 19 (1967)
1312--1315.

\bibitem{Z04} A. Zotov,
{\it Elliptic linear problem for the Calogero-Inozemtsev model and Painlev\'e VI equation},
 Lett. Math. Phys., 67:2 (2004), 153--165; 	arXiv:hep-th/0310260.

\bibitem{Z18}  A.V. Zotov, {\it Calogero–Moser model and R-matrix identities},
 Theoret. and Math. Phys., 197:3 (2018), 1755--1770.






%\bibitem{Tak} K. Takemura, {\it The Hermite–Krichever Ansatz for Fuchsian equations with applications to the sixth Painlev\'e equation and to %finite-gap potentials}, Mathematische Zeitschrift  263 (2009) 149--194; arXiv:math/0504540 [math.CA].


%A.M. Levin, M.A. Olshanetsky, A.V. Zotov,
% {\it Classification of isomonodromy problems on elliptic curves},
% Russian Math. Surveys, 69:1 (2014) 35--118; 	arXiv:1311.4498 [math-ph].
%
% A.M. Levin, M.A. Olshanetsky, A.V. Zotov,
% {\it Painlev\'e VI, rigid tops and reflection equation}, Commun. Math. Phys., 268:1 (2006) 67--103;
% 	arXiv:math/0508058 [math.QA].




%\bibitem{ZFourier} A. Zotov, {\it Relativistic elliptic matrix tops and finite Fourier transformations},
	%Modern Physics Letters A, 32 (2017) 1750169, 	arXiv:1706.05601 [math-ph].

\end{thebibliography}
\end{document}